\documentclass[authorcolumns,numberwithinsect]{no-lipics-v2022}
\usepackage{graphicx} 
\usepackage{amsthm,amssymb,amsfonts}
\usepackage{tikz}
\nolinenumbers

\usepackage{todonotes}

\usepackage{calrsfs}

\newtheorem{observation}[theorem]{Observation}

\newcommand{\Hom}{\mathsf{Hom}}

\newcommand{\homs}{\#\mathsf{Hom}}

\newcommand{\upp}{\mathsf{upp}}

\newcommand{\fulltup}{\mathsf{FullTuple}}

\newcommand{\exthom}[2]{#1^\varepsilon_{#2}}

\newcommand{\rcrFrac}{\textup{$k$-frac-RCR}}

\newcommand{\rcr}{\textup{$k$-proRCR}}

\newcommand{\ar}{\operatorname{ar}}
\newcommand{\cA}{\mathcal{A}}
\newcommand{\cB}{\mathcal{B}}
\newcommand{\cC}{\mathcal{C}}

\newcommand{\CT}{\mathsf{CT}}
\newcommand{\profile}{\mathsf{Profile}}
\newcommand{\bagtup}{\mathsf{BagTuple}}

\newcommand{\dom}{\mathsf{dom}}

\newcommand{\enc}{\mathbf{B}^\mathsf{split}}
\newcommand{\fenc}{\mathbf{F}^\mathsf{split}}

\newcommand{\MCT}{\mathsf{MCT}}

\newcommand{\atp}{\mathsf{atp}}
\newcommand{\stp}{\mathsf{stp}}
\newcommand{\flatt}{\mathsf{flat}}

\newcommand{\hmrcr}{\textup{$k$-RCR}}
\newcommand{\bij}{\mathfrak{E}}
\newcommand{\hmcol}{\kappa}

\newcommand{\rcrcol}{\chi}

\title{Homomorphism Indistinguishability Beyond Graphs: Relational Weisfeiler--Leman and Hypertree Width}
\titlerunning{Homomorphism Indistinguishability Beyond Graphs}

\authorrunning{P. Aivasiliotis, H. Chen, A. G\"obel, M. Lanzinger, M. Roth, M. Schneider}

\author{Panagiotis Aivasiliotis}{Hasso Plattner Institute, University of Potsdam}{panos.aivasiliotis@hpi.de}{}{}

\author{Hubie Chen}{Department of Informatics, King’s College London}{hubie.chen@kcl.ac.uk}{}{EPSRC grant EP/V039318/1, EPSRC grant EP/X03190X/1}

\author{Andreas Göbel}{Hasso Plattner Institute, University of Potsdam}{andreas.goebel@hpi.de}{}{Postdoc Network Brandenburg}

\author{Matthias Lanzinger}{Institute of Logic and Computation, TU Wien}{matthias.lanzinger@tuwien.ac.at}{}{}

\author{Marc Roth}{School of Electronic Engineering and Computer Science, Queen Mary University of London}{m.roth@qmul.ac.uk}{}{}

\author{Markus Schneider}{Independent}{markus.schneider@gmx.co.uk}{}{EPSRC grant EP/V039318/1; the bulk of their research contributions were performed while M. Schneider was employed at King's College London.}

\begin{document}

\maketitle

\begin{abstract}
    The Weisfeiler-Leman (WL) algorithm is one of the most influential heuristics for the graph isomorphism problem and constitutes a cornerstone of Babai's celebrated quasi-polynomial-time isomorphism test [STOC'16]. Starting with the seminal work of Cai, Fürer, and Immerman [Combinatorica'92] the expressive power of WL has been extensively studied over the past 35 years in the contexts of descriptive complexity, logics, graph neural networks, and the theory of homomorphism indistinguishability, the latter of which dates back to early works of Lov{\'a}sz. In a landmark result, Dvo\v{r}{\'{a}}k [J.\ Graph Theory '10], and Dell, Grohe, and Rattan [ICALP'18] proved that two graphs are indistinguishable by the $k$-dimensional WL algorithm if and only if they are indistinguishable by homomorphism counts from graphs of treewidth at most $k$.

    A central open question is whether there is a natural version of the WL algorithm which operates on hypergraphs and relational structures of higher arity. Scheidt [ICALP'24] argues that a ``proper'' version of a $k$-dimensional relational WL should admit an equivalent characterisation via homomorphism indistinguishability along bounded generalised hypertree width (GHW) and poses the search for such a version for all $k\geq1$ as an open problem. In follow-up work, Scheidt and Schweikardt [MFCS'25] confirmed this for $k=1$ by defining relational colour refinement (RCR), a 1-dimensional WL variant, and by showing the respective homomorphism-indistinguishability result from ($\alpha$-)acyclic structures (i.e. structures with GHW=1). Nevertheless, the question remained open for all $k>1$.

    In this work, we provide a definitive and affirmative resolution: we develop a $k$-dimensional version of RCR ($k$-RCR for short) and show that two structures $\mathcal A$ and $\mathcal B$ are indistinguishable by $k$-RCR if and only if they have the same number of homomorphisms from all structures $\mathcal C$ of generalised hypertreewidth at most $k$. Furthermore, we show that on a structure $\mathcal A$ over a fixed signature, with $\mathsf{ghw}(\mathcal{A})\leq k$, $k$-RCR runs in $O(\|\mathcal A\|^k\cdot\log(\|\mathcal A\|))$, where $\|\mathcal A\|$ is the number of all tuples that appear in the relations of $\mathcal A$.
\end{abstract}

\vspace{3em}

\noindent
\textbf{Note on the origin of this paper.}
This paper emerged from the integration of two originally independent preprints and unpublished manuscripts: Aivasiliotis, Göbel, Lanzinger, and Roth~\cite{AivasiliotisGLR26}, as well as Chen and Schneider~\cite{ChenSchneiderUnpublished} independently discovered higher-arity versions of the WL algorithm which we found to be equivalent in expressive power. This manuscript consolidates the respective contributions into a unified presentation.

\vspace{0.6em}

\noindent
\textbf{Declaration of AI use.}
AI has been used exclusively for final proof and type checking. Specifically, the research and writing of this paper has been done without any assistance of AI. Please consult the full AI statement at the end of the manuscript.

\newpage

\tableofcontents

\newpage
\section{Extended Abstract}

The \emph{colour refinement} algorithm is a famous heuristic for the graph isomorphism problem and its inception dates back to 1965 when Morgan introduced an iterative neighbourhood-refinement method for encoding chemical structures~\cite{Morgan1965}. In its modern form, the algorithm tests, in near linear time~\cite{BerkholzBG17}, whether two graphs could be isomorphic, or simply put, whether two graphs are ``identical'' up to relabelling of the vertices---a particularly intriguing problem in NP as, up to this day, it is not known to be polynomial-time solvable and at the same time not expected to be NP-hard either, in part due to Babai's seminal quasi-polynomial-time isomorphism test~\cite{10.1145/2897518.2897542}. Colour refinement is sound, that is, any two graphs distinguished by colour refinement are non-isomorphic. However, it is not complete, that is, there are pairs of non-isomorphic graphs that cannot be distinguished by colour refinement. Nevertheless, colour refinement is a considerably powerful, and fast, heuristic: for example, the probability that it correctly distinguishes a given pair of non-isomorphic random $n$-vertex graphs converges to $1$ as $n$ approaches infinity~\cite{BabaiES80}.

In a nutshell, the  algorithm is a finite iterative process that colours the vertices of a graph according to the following rule: all vertices get the same colour at round $0$, and two vertices $u, v$ get the same colour at round $i+1$ if and only if for every colour $c$ that was assigned to vertices at round $i$, vertices $u, v$ have the same number of $c$-coloured neighbours. The process terminates when the partitioning induced by the colouring becomes stable. Lifting this iterative refinement process from individual vertices to $k$-tuples of vertices gives rise to a strictly more powerful graph isomorphism heuristic: the $k$-dimensional \emph{Weisfeiler-Leman} algorithm~\cite{weisfeiler1968reduction} ($k$-WL, for short).\footnote{Note that the original work by Weisfeiler and Leman considered the case $k=2$, while the general case was studied in later works, see e.g.\ \cite{Cai1992}.} $k$-WL can be run in time $O(n^k \log n)$~\cite{ImmermanL90} and has been subject of a flurry of results on understanding its expressibility and descriptive complexity theory~\cite{Cai1992,dell_et_al:LIPIcs.ICALP.2018.40,KieferPS19,GroheK19,BartoBD26}, as well as its connections to the theory of machine learning~\cite{Morris19,Grohe21,Barceloetal22,Morrisetal23}. Despite the 1992 result of Cai, Fürer, and Immerman~\cite{Cai1992} showing that k-WL is not a complete isomorphism test for any fixed $k$, the $k$-WL algorithm remains a key component of Babai's quasi-polynomial-time graph isomorphism test~\cite{10.1145/2897518.2897542}.
For what follows, we say that $k$-WL \emph{distinguishes} two graphs $G, G'$ if the partitioning of the set of $k$-vertex tuples of $G$ into \textit{colour classes} induced by the \textit{stable} colouring, which  $k$-WL eventually produces, is different for $G$ and $G'$. 

\paragraph*{Lov{\'a}sz meets Weisfeiler and Leman: $k$-WL and Homomorphism Indistinguishability}
A remarkable connection between graph isomorphism and counts of graph homomorphisms was
established in early works of Lov{\'a}sz \cite{lovasz1967operations} (see also Lov{\'a}sz's textbook~\cite{lovasz2012large}): 
\begin{center}
\textit{Two graphs $G$ and $G'$ are isomorphic if and only if, for each graph $H$ the number of graph homomorphisms from $H$ to $G$ is equal to the number of graph homomorphisms from $H$ to $G'$.}
\end{center}
Here, a homomorphism $h$ from $H$ to $G$ is a mapping from the vertex set of $H$ to the vertex set of $G$ that preserves the edges of $H$, i.e., for each edge $\{u, v\}$ of $H$, we have that $\{h(u), h(v)\}$ is an edge of $G$. We write $\Hom(H, G)$ for the set of homomorphisms from $H$ to $G$. 

Lov{\'a}sz's theorem induces a similarity measure of graphs coarser than isomorphism by restricting the class of graphs from which homomorphisms are counted: we say that two graphs $G$ and $G'$ are \emph{homomorphism indistinguishable} over a class of graphs $\mathfrak C$ if, for all $H \in\mathfrak C$, we have $\homs(H,G)=\homs(H,G')$. In that way, every class $\mathfrak C$ yields a heuristic for graph isomorphism that is sound, but not necessarily complete. Homomorphism indistinguishability has received a significant amount of attention in recent years (see, for instance, ~\cite{BokerCGR19,Seppelt23,Scheidt24,Neuen24,RobersonS24,KarRS025,CernyS26}), most notably in the context of Man\v{c}inska's and Roberson's celebrated result on the equivalence of quantum isomorphism and homomorphism indistinguishability over planar graphs~\cite{MancinskaR20}.

In two independently discovered landmark results, Dvo\v{r}{\'{a}}k~\cite{Dvorak10}, and Dell, Grohe, and Rattan~\cite{dell_et_al:LIPIcs.ICALP.2018.40} found a surprising connection between $k$-WL and homomorphism indistinguishability over graphs of bounded treewidth:\footnote{Informally, treewidth measures how similar a graph is to a tree. We refer the reader to~\cite[Chapter 7]{Cyganetal15} for a detailed exposition and note that, for the present work, we will only rely on hypertree decompositions which are formally introduced in Section~\ref{sec:prelims}.}
\begin{theorem}[\cite{Dvorak10,dell_et_al:LIPIcs.ICALP.2018.40}]\label{thm:DGR}
    For all $k\geq 1$, two graphs~$G$ and~$G'$ are indistinguishable by $k$-WL if, and only if,~$G$ and~$G'$ are homomorphism indistinguishable over the class of all graphs of treewidth at most $k$.\qed
\end{theorem}
This result created an interface between $k$-WL (including its connections to logics~\cite{Dvorak10} and Graph Neural Networks~\cite{Morris19,xu2018how,Grohe21}) and algorithmic applications of homomorphism counting which subsequently lead to advances in graph motif counting~\cite{ArvindFKV19,Seppelt24,CurticapeanN25}, database query evaluation~\cite{gobel2024weisfeiler,FockeGRZ25}, and Machine Learning~\cite{lanzinger2024on,BouritsasFZB23, XiaLL23a}. 

\paragraph*{Towards a Relational WL Algorithm}
Given the importance of $k$-WL, both as an isomorphism heuristic and through its connections to homomorphism indistinguishability, it is natural to ask whether Theorem~\ref{thm:DGR} and its applications extend beyond graphs, for instance to hypergraphs and relational structures.
Specifically, an instance of this new endeavour is to design a \textit{relational} $k$-WL algorithm that is executed on relational structures and which is also capable of hosting similar characterisations in terms of homomorphism-indistinguishability. 

Butti and Dalmau \cite{butti_et_al:LIPIcs.MFCS.2021.27} proposed such a relational $k$-WL algorithm\footnote{We also highlight the contribution of the authors of~\cite{Barceloetal22}, published around the same time, which investigates $k$-WL for relational structures, but only applies to arity 2, as well as the work of B\"oker~\cite{Boker19} who considers colour refinement on the incidence graphs of hypergraphs.} for which they established a similar equivalence with respect to homomorphism indistinguishability via the treewidth of the Gaifman graphs of relational structures. While this constitutes a significant first step, the treewidth of the Gaifman graphs exhibits a variety of limitations: for example, a relational structure consisting of only one tuple with $r$ distinct elements already yields a Gaifman graph of treewidth $r-1$; thus applying Butti's and Dalmau's result would need $r-1$ dimensions for their relational WL-algorithm, despite the structure being very simple.

Apart from the treewidth of the Gaifman graph, there is a plethora of other structural parameters that are algorithmically more useful for relational structures, including, specifically, \emph{generalised hypertreewidth}. 
This measure generalises hypergraph ($\alpha$-)acyclicity and often functions as the algorithmic analogue to treewidth in general relational structures.
Further motivation for characterisations of (relational) homomorphism indistinguishability that involve other structural parameters than the treewidth of the Gaifman graphs will be made more explicit later, where we will discuss the algorithmic implications of such characterisations for the problem of counting homomorphisms between relational structures that depend on the specific width measure that is being considered. In this context it is already worth noting that the aforementioned trivial structure consisting of only one tuple of size $r$, which we recall to have a Gaifman graph of treewidth $r-1$, has generalised hypertreewidth $1$ for any $r$.

\noindent Pursuing this direction, Scheidt and Schweikardt \cite{scheidt_et_al:LIPIcs.MFCS.2025.88}, very recently introduced a relational version of the  algorithm---which they called \textit{relational colour refinement (RCR)}---for relational structures and showed that 
\begin{center}
\textit{RCR-indistinguishability is equivalent to homomorphism-indistinguishability over the class of all \textit{$\alpha$-acyclic} relational structures,}
\end{center}
which are the structures with generalised hypertreewidth equal to 1. Note that, as desired, their result precisely matches the equivalence statement of Theorem~\ref{thm:DGR} for the case of 1-WL, with 1-WL being replaced by RCR and with treewidth being replaced by generalised hypertreewidth.

While their result answers the question in our discussion above for $k = 1$, the more general case $k > 1$ was left open. Moreover, in view of the current state of the art for graphs, Scheidt explicitly points out that ``the distinguishing power
of [a $k$-dimensional relational WL algorithm] should match homomorphism indistinguishability over the class [...]
of hypergraphs of generalised hypertree width at most $k$''~\cite{Scheidt24}.

\subsection{Our Results}
In this work, we resolve the quest of finding a relational WL variant that matches homomorphism indistinguishability over structures of bounded generalised hypertreewidth. Specifically, for all $k > 1$, we present a novel relational $k$-WL algorithm --- which we call $\hmrcr$, and which we formally introduce in Section~\ref{sec:intro_krcr} --- and identify its expressive power in terms of 
homomorphism-indistinguishability over the class of all relational structures with generalised hypertreewidth at most~$k$. Our main theorem is formally stated as follows.

\begin{theorem}[Main Theorem] \label{maintheorem:boundedGHDforStructures_neq}
Fix a signature $\sigma$ and a positive integer $k$. For any two $\sigma$-structures $\mathcal{A}, \mathcal{B}$, the following are equivalent:
\begin{enumerate}
\item $\hmrcr$ distinguishes $\mathcal{A}$ and $\mathcal{B}$;
\item There is a connected $\sigma$-structure $\mathcal{C}$ that admits a generalised hypertree decomposition of width at most $k$ such that $\#\Hom(\mathcal{C}, \mathcal{A}) \neq \#\Hom(\mathcal{C}, \mathcal{B})$.\qed
\end{enumerate}    
\end{theorem}

What we find particularly notable about this result is that it connects notions that come from rather different areas. Generalised hypertreewidth was developed in the algorithmic theory of constraint satisfaction and database queries as a width measure capturing structural tractability and with limited and only very recently discovered connections to logic~\cite{scheidt_et_al:LIPIcs.MFCS.2025.88}. While generalised hypertree width is a generalisation of treewidth for higher arity structures, the width represents a very different structural aspect of the structure. In treewidth the width of a bag in the tree decomposition corresponds to the number of vertices in the bag, thus closely connected to the variables in a logical expression of the bag's structure. For generalised hypertree width, the width is a covering parameter, measuring how well the bag can be covered by edges in the structure. Algorithmically, this combinatorially bounds the number of bag assignments that need to be considered, but it is not known to have deeper structural meaning.

It is therefore somewhat surprising that bounding this algorithmic width parameter is exactly reflected by a variant of colour refinement, and that the same condition admits a homomorphism indistinguishability characterisation. The case $k=1$ already hints at this connection, since generalised hypertreewidth $1$ coincides with $\alpha$-acyclicity~\cite{GottlobLS02,GottlobMS09}, our result shows that this is not merely a special case, but part of a uniform correspondence for every bound $k$.

\subsection{Our Relational Weisfeiler-Leman Algorithm: $\hmrcr$}\label{sec:intro_krcr}

To formally state the $\hmrcr$ algorithm involved in \Cref{maintheorem:boundedGHDforStructures_neq}, we first need to introduce some notation.


Fix a finite relational signature $\sigma$, consisting of relation symbols, each of which is associated with a positive integer that is called the \textit{arity} of the relation symbol.  Let $\cA$ be a $\sigma$-structure, with domain $A$ and relations $R^\cA \subseteq A^{\ar(R)}$, for each $R \in \sigma$ where $\ar(R)$ denotes the arity of $R$. We also fix a positive integer $k$. 

In a nutshell, for $k = 1$, the $\hmrcr$ algorithm is exactly the Relational Color Refinement (RCR) algorithm introduced by Scheidt and Schweikardt \cite{scheidt_et_al:LIPIcs.MFCS.2025.88}. For $k > 1$, $\hmrcr$ can be seen as running RCR on a structure obtained from $\cA$ by enriching its tuple set with what we call \textit{flat tuples}, defined as follows. A flat tuple (w.r.t. $k$) is a tuple that is obtained by concatenating any collection of up to $k$ tuples in $\cA$. In particular, given a set of up to $k$ tuples, all possible arrangements of tuples are considered. For example, w.r.t. $k=2$, the set of tuples $\{ab, cac\}$ yields the flat tuples $ab, cac, abab, caccac, abcac, cacab$. We remark that the set of flat tuples includes all tuples that were originally in $\cA$. We write $\mathcal{T}_k(\cA)$ for the set of flat tuples of $\cA$ (w.r.t. $k$). 

Following the definition of RCR, the last ingredient is to define \textit{atomic types} and \textit{similarity types} of flat tuples. For $t \in \mathcal{T}_k(\cA)$, we write $\atp(t) = \{R \in \sigma : t \in R^\cA\}$ for the \textit{atomic type} of $t$, that captures which relations contain $t$. Furthermore, given $s, t \in \mathcal{T}_k(\cA)$ with respective lengths $\ell_s, \ell_t$, we write $\stp(s, t) = \{(i, j) \mid 1 \leq i \leq \ell_s, 1 \leq j \leq \ell_t, s[i] = t[j]\}$, for the \textit{similarity type} of tuples $s, t$ ---where $s[i]$ denotes the domain element appearing at index $i$ in $s$--- that captures all pairs of indices of $s, t$ respectively that $s, t$ agree on with each other. In particular, the similarity type of a single flat tuple $t$ is given as $\stp(t, t)$. 

\begin{definition}[$\hmrcr$]
We fix $k\ge 1$ and define the $k$-relational  algorithm ($\hmrcr$) that iteratively colours the elements of $\mathcal{T}_k(\cA)$ by the colouring $\kappa_i^\cA$ computed as follows (where $i$ is the iteration counter)
  \[\kappa_0^\cA(t):=(\atp(t),\stp(t, t));\]
  \[
\kappa_{i+1}^\cA(t):=
\Bigl(
\kappa_i^\cA(t),\
\{\!\{\,(\stp(t, s),\kappa_i^\cA(s))\mid
s\in\mathcal{T}_k(\cA),\ \stp(t, s)\neq\emptyset\,\}\!\}
\Bigr),
\]
where $\{\!\{\cdot \}\!\}$ denotes a multiset.
We say that $\hmrcr$ \emph{stabilises} after iteration $j$ if the colour classes formed before iteration $j$ do not change after iteration $j$, that is, if the partitions of the set $\mathcal{T}_k(\cA)$ induced by $\kappa_{j-1}^\cA$ and $\kappa_{j}^\cA$, respectively, are equal. Let $\kappa_\infty^\cA$ be the stable colouring, and for a stable colour $c$ let
\[E_k^\cA(c):=|\{t\in\mathcal{T}_k(\cA)\mid \kappa_\infty^\cA(t)=c\}|\,.\]
For $\sigma$-structures $\cA,\cB$, write $\cA\equiv_{\hmrcr}\cB$ iff
$E_k^\cA(c)=E_k^\cB(c)$ for all stable colours $c$, in which case we say that $\hmrcr$ \emph{does not distinguish} structures $\mathcal{A}$ and $\mathcal{B}$.
\end{definition}

\paragraph*{The running time of $\hmrcr$}

We remind the reader that we are working with a fixed signature $\sigma$ and a fixed dimension $k$ of the $\hmrcr$ algorithm. Scheidt and Schweikardt have shown in \cite{scheidt_et_al:LIPIcs.MFCS.2025.88} that RCR can be implemented on a given $\sigma$-structure $\cA$ (for fixed $\sigma$) in time $O(\|\cA\|\cdot\log(\|\cA\|))$, where $\|\cA\| \coloneq |\bigcup_{R \in \sigma}R^\cA|$ denotes the \textit{size} of $\cA$. Using the aforementioned algorithm, we show that the stable colour classes produced by $k$-RCR when run on a $\sigma$-structure $\cA$, can be computed in time $O(\|\cA\|^k\cdot\log(\|\cA\|))$. 
\begin{remark}
The $O$-notation above hides terms of size $2^{O(r_{\mathsf{max}}\cdot\log(r_{\mathsf{max}}))}$, where $r_{\mathsf{max}}$ denotes the maximum length of a tuple in $\mathcal{T}_k(\cA)$, which depends on the maximum arity of symbols in $\sigma$ as well as $k$, both of which are considered fixed. 
\end{remark}

\begin{theorem}\label{thm:intro_rt}
Fix a relational signature $\sigma$ and a positive integer $k \geq 1$. Given a $\sigma$-structure $\cA$, we can compute the stable colour classes produced by $\hmrcr$ when run on $\cA$ in time $O(\|\cA\|^k\cdot\log(\|\cA\|))$, where $\|\cA\| = |\bigcup_{R \in \sigma}R^\cA|$ is the number of all tuples that appear in the relations of $\cA$.   \qed
\end{theorem}
The proof of Theorem~\ref{thm:intro_rt} follows fairly easily from the running time bound of Scheidt and Schweikardt's RCR algorithm and is thus deferred to Appendix~\ref{appendix_rt}. 

We highlight that the very recent announcement of the (unconditional) lower bound for WL on graphs by OpenAI~\cite{openai2026wllowerbounds}, suggests that a linear dependency on $k$ in the exponent of the running time of $\hmrcr$ might be unavoidable since $\hat{k}$-WL (on graphs) is a special case of $\hmrcr$, where $k$ and $\hat{k}$ are linearly related. However, we do not state an explicit lower bound for $\hmrcr$ as the latter announcement has yet to be scrutinised and verified by the community.

\subsection{Technical Overview}
In order to prove Theorem~\ref{maintheorem:boundedGHDforStructures_neq}, we introduce a \emph{profiled} version of $\hmrcr$, where we not only keep track of tuples but also their \textit{profile} --- which is a sequence of relation symbols --- defined momentarily. This version gives us more control for arguing about homomorphism indistinguishability. We call this profiled version $\rcr$, and we will show in Section~\ref{sec:rcr_equivalence} that both versions, $\hmrcr$ and $\rcr$, are equivalent.

We note that the equivalence with $\rcr$ is not purely of technical interest. Arguably, both $\hmrcr$ and $\rcr$ represent natural ways of generalising plain relational colour refinement, even though the procedures are superficially rather different. It is then remarkable that these two approaches have the exact same distinguishing power and we consider this an indicator for the mathematical naturalness of the underlying equivalence relations defined by  $\hmrcr$ indistinguishability.

\subsubsection{$\rcr$: A profiled version of $\hmrcr$}
We refer the reader to \Cref{sec:prelims} for a quick background on relational structures which we assume familiarity with for the rest of this section. We fix a finite relational signature $\sigma$, and let $\ar(\sigma)$ be the maximum arity in $\sigma$. We also \emph{globally} fix a total order $\preceq_\sigma$ on the relation symbols in $\sigma$.
A \emph{coloured tuple} of $\cA$ is a pair $(\bar a, R)$ where $\bar a \in R^{\cA}$, which we denote by $R(\bar{a})$. The need of introducing the notion of coloured tuples is for appropriate bookkeeping of tuples that occur in multiple relations. We write $\CT(\cA)$ for the set of all coloured tuples in $\cA$. Moreover, we fix a total order $<$ on the domain and define the total order $\prec_\cA$ over its coloured tuples as follows:
\[
R(\bar a)\prec_\cA S(\bar b)
\iff
\bigl(R\prec_\sigma S\bigr)\ \lor\ \bigl(R=S\ \land\ \bar a<\bar b\bigr).
\]
Let $\bot$ be a fresh dummy symbol that we will use to represent ``no relation symbol'', and let $\varepsilon$ be a fresh domain element that we will use to represent ``no domain element''.

\paragraph*{The $\rcr$ Algorithm}

Fix a signature $\sigma$ and let $\sigma^{(\le k)}
:=
\left\{
(R_1,\dots,R_m)\ \middle|\
1\le m\le k,\ R_i\in \sigma\ \text{for all }i,\ 
R_1 \preceq_\sigma \cdots \preceq_\sigma R_m
\right\}$.
For a $\sigma$-structure $\cA$, define
\[
\MCT_k(\cA)=
\{(\alpha;\bar t_1,\ldots,\bar t_m)\mid
\alpha=(R_1,\ldots,R_m)\in\sigma^{(\leq k)},\
\bar t_i\in R_i^\cA\}.
\]
For $\Omega=(\alpha;\bar t_1,\ldots,\bar t_m)$ set
$\profile(\Omega):=\alpha$, $\flatt(\Omega):=\bar t_1+\cdots+\bar t_m$,
and $L(\Omega):=|\flatt(\Omega)|$.

There is also the notion of \textit{similarity types} that is central to relational $k$-WL algorithms (see \cite{scheidt_et_al:LIPIcs.MFCS.2025.88}), which we extend appropriately to elements of $\MCT_k(\cA)$ as follows.

\begin{restatable}[Similarity types for $\rcr$]{definition}{rcrTypesDefinition}
For $\Omega,\Psi\in\MCT_k(\cA)$, we define
$\stp(\Omega,\Psi):=\{(p,q)\in [L(\Omega)]\times[L(\Psi)]\mid \flatt(\Omega)[p]=\flatt(\Psi)[q]\}$,
and $\stp(\Omega):=\stp(\Omega,\Omega)$.
\end{restatable}

We are now ready to state $\rcr$ which is essentially a colouring function from elements of $\MCT_k(\mathcal{A})$ that is iteratively updated until it reaches a point where the \textit{colour classes} that it induces do not change in which case we say that the algorithm stabilises.

\begin{restatable}[\rcr]{definition}{kRCRdefinition}\label{def:kRCR}
We fix $k\ge 1$ and define the $k$-relational  algorithm ($\rcr$) that iteratively colours the elements of $\MCT_k(\cA)$ by the colouring $\chi_i^\cA$ computed as follows (where $i$ is the iteration counter)
\begin{itemize}
    \item[] \[\chi_0^\cA(\Omega):=(\profile(\Omega),\stp(\Omega));\]
    \item[] \[
\chi_{i+1}^\cA(\Omega):=
\Bigl(
\chi_i^\cA(\Omega),\
\{\!\{\,(\stp(\Omega,\Psi),\chi_i^\cA(\Psi))\mid
\Psi\in\MCT_k(\cA),\ \stp(\Omega,\Psi)\neq\emptyset\,\}\!\}
\Bigr).
\]
\end{itemize}
We say that $\rcr$ \emph{stabilises} after iteration $j$ if the colour classes formed before iteration $j$ do not change after iteration $j$.\footnote{Note that similarly to the definition of $k$-WL algorithm (see \Cref{sec:prelims}), $\rcr$ refines at each iteration the colour classes, which in turn implies that the algorithm stabilises after at most $|\MCT_k(\mathcal{A})|$ many iterations.} Let $\chi_\infty^\cA$ be the stable colouring, and for a stable colour $c$ let
\[\mathsf{mult}_k^\cA(c):=|\{\Omega\in\MCT_k(\cA)\mid \chi_\infty^\cA(\Omega)=c\}|\,.\]

For $\sigma$-structures $\cA,\cB$, write $\cA\equiv_{\rcr}\cB$ iff
$\mathsf{mult}_k^\cA(c)=\mathsf{mult}_k^\cB(c)$ for all stable colours $c$, in which case we say that $\rcr$ \emph{cannot distinguish} structures $\mathcal{A}$ and $\mathcal{B}$.
\end{restatable}

\subsubsection{Equivalence of $\hmrcr$ and $\rcr$}
The following result establishes the said equivalence and, therefore, enables us to focus on $\rcr$ for the remainder of the technical overview.
\begin{theorem}[Equivalence of \rcr{} and \hmrcr]\label{thm:intro_rcr-hmrcr-equivalence}
    Let $k\geq 1$ be an integer and let $\sigma$ be a signature. For any pair of $\sigma$-structures $\cA$ and $\cB$ we have $\cA\equiv_{\rcr}\cB$ if and only if $\cA\equiv_{\hmrcr}\cB$. \qed
\end{theorem}

We provide a brief exposition of the high-level ideas of the proof of \Cref{thm:intro_rcr-hmrcr-equivalence} here and refer the reader to \Cref{sec:rcr_equivalence} for the technical details. To begin with, observe that a stable colouring produced by a WL algorithm can be equivalently seen as an equivalence relation over the elements coloured by the algorithm, corresponding to the structures $\cA, \cB$ to be tested. Our toolkit consists of some special equivalence relations over the elements coloured by $\hmrcr$ (resp. $\rcr$), which we call \textit{robust} (resp. \textit{pro-robust}) and a suitably defined type of an \textit{isomorphism} between robust and pro-robust equivalence relations. In particular, robust (resp. pro-robust) equivalence relations correspond to refinements of the equivalence relation that comes from the stable colouring produced by $\hmrcr$ (resp. $\rcr$). In other words, we show that the equivalence relation corresponding to the stable colouring produced by $\hmrcr$ (resp. $\rcr$) is the \textit{coarsest} robust (resp. pro-robust) equivalence relation---that is, every robust (resp. pro-robust) equivalence relation is a class-wise subset of it. Finally, we show that the aforementioned isomorphism maps the coarsest robust equivalence relation to the coarsest pro-robust equivalence relation, and its specifications guarantee that this is sufficient to show that $\hmrcr$ and $\rcr$ have the same distinguishing power.

We remark that our techniques for establishing the equivalence between $\hmrcr$ and $\rcr$ (in terms of distinguishing power) may be of independent interest, as we consider the framework we developed general enough so as to be used for establishing similar equivalences between variants of the rich family of WL algorithms.

\subsubsection{The central construction: Binary Structures Induced by Generalised Hypertree-Decompositions}\label{sec:intro-ghd}

With \cref{thm:intro_rcr-hmrcr-equivalence} at hand, in order to prove \cref{maintheorem:boundedGHDforStructures_neq} it suffices to prove the analogous statement for \rcr{} instead. That is, we show the following statement.

\begin{theorem}\label{maintheorem:boundedGHDforStructures}
Fix a signature $\sigma$ and a positive integer $k$. For any two $\sigma$-structures $\mathcal{A}, \mathcal{B}$, the following are equivalent:
\begin{enumerate}
\item $\rcr$ distinguishes $\mathcal{A}$ and $\mathcal{B}$, that is, $\cA \not\equiv_{\rcr}\cB$.
\item There is a connected $\sigma$-structure $\mathcal{C}$ that admits a generalised hypertree decomposition of width at most $k$ such that $\#\Hom(\mathcal{C}, \mathcal{A}) \neq \#\Hom(\mathcal{C}, \mathcal{B})$.\qed
\end{enumerate}    
\end{theorem}

For the special case of $k=1$, the authors of \cite{scheidt_et_al:LIPIcs.MFCS.2025.88} devised a construction in which the structure (which is binary in their case) is induced by a tree-decomposition. To prove \cref{maintheorem:boundedGHDforStructures}, we extend their definition in an intricate way so as to induce binary structures by generalised hypertree-decompositions of relational structures of \textit{any} arity.

\paragraph*{Generalised Hypertree Decompositions for Relational Structures}

Typically, tree decompositions of a relational structure decompose either the Gaifman graph of the structure or the hypergraph that is associated with it. In this work, we consider decompositions of the latter case which we slightly refine: such tree decompositions usually involve a mapping $\lambda$ from the nodes of the underlying tree to a subset of hyperedges of the hypergraph of the corresponding relational structure; in our refined definition, $\lambda$ maps tree-nodes to sequences of \emph{coloured tuples} of the corresponding relational structure instead. Such a refinement in the definition does not affect the \emph{width} of the decomposition (compared to the width of the standard decomposition).

\begin{definition}[Edge cover based on coloured tuples]

Fix a signature $\sigma$ and let $\mathcal{A}$ be a $\sigma$-structure. Recall that we write $\mathsf{CT}(\mathcal{A})$ for the set of coloured tuples of $\mathcal{A}$. Let $X \subseteq A$. We say that a set $C\subseteq \mathsf{CT}(\mathcal{A})$ of coloured tuples is an \emph{edge cover} of $X$, if $X \subseteq \bigcup_{R(\bar{t})}\mathsf{set}(\bar{t}) \coloneq  \{x \mid x \in \bar{t}\}$. The \emph{edge cover number} of a relational structure $\mathcal{A}$ is the minimum size among all edge covers $C \subseteq \mathsf{CT}(\mathcal{A})$ of $A$.
\end{definition}

\begin{definition}[Generalised hypertree-decompositions of relational structures]
Fix a signature $\sigma$. A \emph{generalised hypertree decomposition} (GHD, for short) of a $\sigma$-structure $\mathcal{A}$ is an ordered triplet $(T, B, \lambda)$ where
\begin{itemize}
\item $T$ is a tree;
\item $B : V(T) \to 2^A$ is a mapping from the nodes of $T$ to subsets of the domain elements;
\item $\lambda$ is a mapping from the nodes of $T$ to sequences of coloured tuples in $\CT(\cA)$, where we allow for multiple occurrences of the same coloured tuple,
\end{itemize}
such that the following conditions are met.
\begin{itemize}
\item[(1)] For every coloured tuple $\bar{t} \in \mathsf{CT}(\cA)$, there is a node $v \in V(T)$ such that $\mathsf{set}(\bar{t}) \subseteq B(v)$.
\item[(2)] For every domain element $x \in A$, the subtree induced by the nodes $\{v \in V(T) \mid x \in B(v)\}$ is connected.
\item[(3)] For every node $v \in V(T)$, $\lambda(v)$ is an edge-cover of $B(v)$.
\end{itemize} 

The \emph{width} of a GHD is given by $\max_{v \in V(T)}|\lambda(v)|$. The \emph{generalised hypertreewidth} of $\mathcal{A}$, denoted by $\mathsf{ghw}(\mathcal{A})$, is the \emph{minimum} width over all possible GHDs of $\mathcal{A}$.

We say that a GHD $D = (T, B, \lambda)$ is \emph{full} if for
every coloured tuple $R(\bar a)$ of $\cA$, there is a node $v \in V(T)$ such that ($a$) $\mathsf{set}(\bar a) \subseteq B(v)$ (which matches Condition (1) above) and ($b$) $R(\bar a) \in \lambda(v)$.
\end{definition}

Note that, as argued in detail in \Cref{remark:GHDforRelationalStructures}, every GHD can be made into a full GHD while preserving its width. 
We proceed by defining, for each node $u\in V(T)$ two particular representations of the bag $B(u)$ corresponding to $u$. 

\begin{restatable}{definition}{bagTuplesDefinition}\label{def:BagTuples}
Let $\mathsf{Ord\lambda}(u)=R_1(\bar a_1), \dots, R_{k'}(\bar a_{k'})$ be the coloured tuples of $\lambda(u)$ in increasing $\preceq_\sigma$-order.\footnote{Here, we mean that the coloured tuples inside $\mathsf{Ord}\lambda(u)$ need only be ordered by their relations, that is, $R_1 \preceq_\sigma \dots \preceq_\sigma R_{k'}$.}
\begin{itemize}
\item[(1)] We define 
\[
\profile(u) := (R_1,\dots,R_{k'})\,.
\]

\item[(2)] Let $\bar a_i^\star$ be the tuple obtained from $\bar{a}_i$ by replacing every element in $\bar a_i$ that is not in $B(u)$ with $\varepsilon$. We then define $\fulltup(u)$ as the tuple constructed by concatenating $\bar{a}_1, \bar{a}_2, \dots, \bar{a}_{k'}$ and similarly define $\bagtup(u)$ as the  tuple constructed by concatenating $\bar a_1^\star, \bar a_2^\star,\dots,\bar a_{k'}^\star$.
\end{itemize}
\end{restatable}

As a second-to-last step, we define the underlying signature of the intended binary structure.

\begin{restatable}[$k$-exploded binary signature]{definition}{explodedSignatureDefinition}\label{def:explodedSignature}
Fix a signature $\sigma$ and an integer $k$. Let \[
\sigma^{(\le k)}
:=
\left\{
(R_1,\dots,R_m)\ \middle|\
1\le m\le k,\ R_i\in \sigma\ \text{for all }i,\ 
R_1 \preceq_\sigma \cdots \preceq_\sigma R_m
\right\}.
\]
We define the \emph{$k$-exploded} binary signature $\widehat{\sigma}_k$ (wrt. $\sigma$) as
\[
\widehat{\sigma}_{k}
:=
\{U_\alpha\mid \alpha\in\sigma^{(\le k)}\}
\ \cup\
\{E_{i,j}\mid i,j\in\{1,\dots,k\cdot \ar(\sigma)\} \}.\qedhere
\]
\end{restatable}

Now we have all the technical background to define the binary structure induced by GHD, key to proving \cref{maintheorem:boundedGHDforStructures}.

\begin{restatable}[Binary structure induced by a GHD]{definition}{ghdBinaryStructureDefinition}
\label{def:ghd.bin.struct}
Let $\cA$ be a $\sigma$-structure and let $D=(T,B,\lambda)$ be a GHD of $\cA$ with width at most $k$.
Define $\cA^D$ as the $\widehat{\sigma_k}$-structure with universe
 $V(T)$ and the minimal interpretation s.t. (where we omit the superscript $\cA^D$):
\begin{enumerate}
\item $u \in U_\alpha$ iff $\alpha = \profile(u)$.
\item $E_{i,j}(u,w)$ if $\{u,w\}\in E(T)\,\cup\,\{\{x,x\} \mid x \in V(T)\}$, and $\bar a[i] = \bar b[j]$ and $\bar a [i] \neq \varepsilon$, where $\bar a = \bagtup(u)$ and $\bar b = \bagtup(w)$.
\end{enumerate}
\end{restatable}

\subsubsection{Proof Sketch of \Cref{maintheorem:boundedGHDforStructures}}

As mentioned earlier in the introduction, our $\rcr$ algorithm is \textit{implicitly} run on an intermediate binary structure associated with the input structure. In fact, this is not a mere artefact of our algorithm. As we make it explicit momentarily, it is crucial for the proof of \Cref{maintheorem:boundedGHDforStructures} to interpret equivalence of structures under $\rcr$ (in symbols, $\equiv_\rcr$) in terms of equivalence of suitably defined \textit{binary} structures ---here we refer to the intermediate binary structures mentioned above---under 1-WL (in symbols, $\equiv_{1\textup{-WL}}$) and our $\rcr$ is precisely capable of accommodating such interpretation. 

We define these intermediate binary structures, which we call \textit{canonical $k$-exploded encodings}--- the underlying signature of which is the $k$-exploded binary signature from \Cref{def:explodedSignature}---as follows.

\begin{restatable}[$k$-exploded encoding]{definition}{explodedEncodingDefinition}\label{def:ExplodedEncoding}
We fix a signature $\sigma$. Let $\cA$ be a $\sigma$-structure.
Let $\mathsf{Prod}^{(\le k)}(\cA)$ be the set that contains all relations of the form $P=R_1^\cA \times \cdots \times R_n^\cA$
\footnote{Here we mean relational algebra interpretation of the product. That is, the product of two relations with arities $a_1,a_2$ contains tuples with arity $a_1+a_2$.}
, where $(R_1,\ldots,R_n)\in\sigma^{(\leq k)}$.
For every such relation $P$ we write $U_P$ as shorthand for the unary relation symbol in $\widehat{\sigma_k}$ for $(R_1,\dots,R_n) \in \sigma^{(\le k)}$.

Define
$\mathbf{B}(\cA,k)$ as the $\widehat{\sigma_k}$-structure with universe
\[
\dom(\mathbf{B}(\cA,k)) = \{\bar a \mid  \exists P \in \mathsf{Prod}^{(\le k)}(\cA) \text{ s.t. }\ \bar a \in P\}.
\]
For the interpretation set we define (where we omit the superscript $\mathbf{B}(\cA,k)$ for readability).
\begin{itemize}
\item For each $P \in \mathsf{Prod}^{(\le k)}(\cA)$ we have $\bar a \in U_P$ iff $\bar a \in P$.
\item $E_{i,j}(\bar a, \bar b)$ iff $\bar a[i] = \bar b[j]$.
\end{itemize}
\end{restatable}

We are now ready to state the aforementioned connection between 1-WL-indistinguishability and $\rcr$-indistinguishability (the proof of which can be found in \Cref{sec:Generalised}).

\begin{proposition}
\label{prop:kRCR.1WL}
Let $\cA,\cB$ be $\sigma$-structures and let $\mathbf{B}(\cA,k),\mathbf{B}(\cB,k)$
be the canonical $k$-exploded encodings of $\mathcal{A}$ and $\mathcal{B}$ respectively.
Then
\( \enspace 
\cA \equiv_{\rcr} \cB
\quad\Longleftrightarrow\quad
\mathbf{B}(\cA,k) \equiv_{\textup{1-WL}} \mathbf{B}(\cB,k).
\) \qed
\end{proposition}

We now proceed with the proof sketch of \Cref{maintheorem:boundedGHDforStructures}. Our analysis is inspired by the work of Scheidt and Schweikardt \cite{scheidt_et_al:LIPIcs.MFCS.2025.88} that established the case $k = 1$. However addressing the case $k > 1$ requires a significantly more elaborate analysis. The backbone of the proofs of both cases $k = 1$ (as shown in \cite{scheidt_et_al:LIPIcs.MFCS.2025.88}) and $k > 1$ (as shown in this work) essentially consists of 
\begin{enumerate}
\item[(I)] interpreting RCR-equivalence (resp. $\rcr$-equivalence) of structures in terms of 1-WL-equivalence of intermediate binary structures.

\item[(II)] interpreting homomorphism counts between structures of \textit{any} arity as homomorphism counts between appropriately defined binary structures.
\end{enumerate}

The first item of the list will be satisfied by the proposition stated above where the associated binary structures are precisely the $k$-exploded encodings of the original structures. For the second item, it is the binary structures induced by GHDs that will allow for such interpretations additionally to the $k$-exploded encodings. To make this more precise, given two structures $\mathcal{A}, \mathcal{B}$ for which we wish to interpret $\homs(\mathcal{A}, \mathcal{B})$ in terms of counting homomorphisms between binary structures, we achieve this by considering the binary structure associated with the left-hand side of $\homs(\mathcal{A}, \mathcal{B})$ to be given by the binary structure induced by a GHD $D$ of $\mathcal{A}$ and respectively by considering the binary structure associated with the right-hand side to be given by the canonical $k$-exploded encoding $\mathbf{B}(\mathcal{B}, k)$ of $\mathcal{B}$. Hence, we relate $\homs(\mathcal{A}, \mathcal{B})$ to $\homs(\mathcal{A}^D, \mathbf{B}(\mathcal{B}, k))$.

The reason as to why \textit{binary} structures appear to play a vital role in the proof is that, in using them, we may---loosely speaking---translate our theorem in terms of the classification of 1-WL of Dvo\v{r}{\'{a}}k~\cite{Dvorak10} and Dell, Grohe and Rattan \cite{dell_et_al:LIPIcs.ICALP.2018.40} (which we state formally below) in regard to homomorphism-indistinguishability over trees.

\begin{theorem}[\cite{Dvorak10,dell_et_al:LIPIcs.ICALP.2018.40}]\label{thm:IntroDellGroheRattan}
Let $\sigma'$ be a finite binary signature and $\cA, \cB$ be $\sigma'$-structures. The following are equivalent.
\begin{enumerate}
\item Color Refinement distinguishes $\cA$ and $\cB$;
\item There exists a $\sigma'$-structure $\mathcal{T}$, the Gaifman graph of which is a tree, such that $\#\Hom(\mathcal{T}, \cA) \neq \#\Hom(\mathcal{T}, \cB)$.\qed
\end{enumerate}
\end{theorem}

We first give a high-level proof sketch of \Cref{maintheorem:boundedGHDforStructures} in the case $k = 1$ 
(in the spirit of \cite{scheidt_et_al:LIPIcs.MFCS.2025.88} but employing our notation for the sake of consistency) 
such that we may later highlight the key conceptual and technical differences compared to our analysis.

\begin{proof}[Proof sketch of the case $k = 1$ of \Cref{maintheorem:boundedGHDforStructures}.]
At the heart of the proof lies the technical observation that for any $\sigma$-structure $\hat{\mathcal{A}}$ and any $\sigma$-structure $\mathcal{F}$ that admits a full and pure GHD $D$ of width 1, it holds
\begin{equation}\label{eq:HomsPreserved}
\homs({\mathcal{F}, \hat{\mathcal{A}}}) = \homs(\mathcal{F}^D, \mathbf{B}(\hat{\mathcal{A}}, 1))\,,
\end{equation}
where a pure GHD satisfies the property that for every node of the decomposition, the coloured tuples inside its edge-cover cover precisely (and no other elements than) the elements inside the bag of the node.
First, for the direction $(1) \implies (2)$, assume that $\mathcal{A} \not\equiv_{\rcr} \mathcal{B}$. Then, \Cref{prop:kRCR.1WL} implies that $\mathbf{B}(\mathcal{A}, 1) \not\equiv_{1\textup{-WL}} \mathbf{B}(\mathcal{B}, 1)$. In turn, it follows from \Cref{thm:IntroDellGroheRattan} that there is a binary structure $\mathcal{T}$ (with the same signature as that of $\mathbf{B}(\mathcal{A}, 1)$ and $\mathbf{B}(\mathcal{B}, 1)$) such that $\homs(\mathcal{T}, \mathbf{B}(\mathcal{A}, 1)) \neq \homs(\mathcal{T}, \mathbf{B}(\mathcal{B}, 1))$. The proof of this direction is then completed by showing that $\mathcal{T}$ can be used as a ``template'' for constructing a $\sigma$-structure $\mathcal{C}$ and a full and pure GHD $D$ of $\mathcal{C}$ of width 1 such that $\mathcal{T}$ is isomorphic to $\mathcal{C}^D$ (the details of which are omitted) and then using \eqref{eq:HomsPreserved} to derive $\homs(\mathcal{C}^D, \mathbf{B}(\mathcal{A}, 1)) \neq \homs(\mathcal{C}^D, \mathbf{B}(\mathcal{B}, 1)) \implies \homs(\mathcal{C}, \mathcal{A}) \neq \homs(\mathcal{C}, \mathcal{B})$.

Finally, the direction $(2) \implies (1)$ is more straightforward. Let $\mathcal{C}$ be a $\sigma$-structure with generalised hypertreewidth 1 such that $\homs(\mathcal{C}, \mathcal{A}) \neq \homs(\mathcal{C}, \mathcal{B})$. Choose a full and pure GHD $D$ of width 1. It then follows from \eqref{eq:HomsPreserved} that $\homs(\mathcal{C}^D, \mathbf{B}(\mathcal{A}, 1)) \neq \homs(\mathcal{C}^D, \mathbf{B}(\mathcal{B}, 1))$. Finally, \Cref{thm:IntroDellGroheRattan} yields the claim.
\end{proof}

For the more general case $k > 1$, we crucially observe that \Cref{eq:HomsPreserved} (and in general, our technical observations above) does not necessarily hold any more. Interestingly, the direction $(1) \implies (2)$ in \Cref{maintheorem:boundedGHDforStructures} may still be shown in a similar fashion as in the special case $k = 1$. This is mainly due to the fact that, as we show, the structure yielded by the adaptation of the proof for the same direction in the case $k = 1$, satisfies the ``pureness'' property mentioned above which, together with fullness, ensures that \Cref{eq:HomsPreserved} holds. More precisely, the generalised hypertree decomposition $D = (T, B, \lambda)$ of the respective structure satisfies that for every node $v \in V(T)$, the elements covered by the coloured tuples in $\lambda(v)$ are precisely those elements contained in bag $B(v)$.

Hence, it is the direction $(2) \implies (1)$ that requires substantially more elaborate effort to establish. It will also be more convenient to show the claim via contraposition, that is, we assume that for the $\sigma$-structures $\mathcal{A}, \mathcal{B}$ of the theorem's statement, $\mathcal{A} \equiv_{\rcr} \mathcal{B}$ holds and then show that this assumption implies homomorphism-indistinguishability over all $\sigma$-structures with generalised hypertreewidth at most $k$.

The pureness property that was met in the case $k = 1$ ensured that it is possible for any $\sigma$-structure $\mathcal{C}$ and a full and pure GHD $D = (T, B, \lambda)$ of $\mathcal{C}$ to induce a homomorphism $h \in \Hom(\mathcal{C}, \mathcal{A})$ by a homomorphism $\hat{h} \in \Hom(\mathcal{C}^D, \mathbf{B}(\mathcal{A}, k))$ in a well-defined way (the inverse is also true and it is more straightforward). More precisely, recalling \Cref{def:BagTuples}, given a node $v \in V(T)$, with $\bagtup(v) \coloneq a_1a_2 \dots a_\ell$, where $\ell \leq k \cdot r$ (where $r$ is the maximum arity among relation symbols in $\sigma$) that is mapped to the tuple $\hat{h}(v) \coloneq  b_1b_2\dots b_\ell$, the mapping $a_i \mapsto b_i$ agrees with the similarly defined mapping of any other node on the entries that they share, implying a well-defined mapping $\mathcal{C} \mapsto \mathcal{A}$ that is also a homomorphism.

However, if the pureness property is not met, $\bagtup(v)$ may feature $\varepsilon$-entries, i.e. elements covered by $\lambda(v)$ that are not present in $B(v)$ that can be mapped arbitrarily. This breaks the consistency that was guaranteed in the way we described earlier for inducing $h \in \Hom(\mathcal{C}, \mathcal{A})$ from $\Hom(\mathcal{C}^D, \mathbf{B}(\mathcal{A}, k))$.

In this more intricate case, we are able to show that for any $\sigma$-structure $\mathcal{C}$ and a full GHD $D$ of $\mathcal{C}$ of width at most $k$, the set $\Hom(\mathcal{C}^D, \mathbf{B}(\mathcal{A}, k))$ (resp. $\Hom(\mathcal{C}^D, \mathbf{B}(\mathcal{B}, k))$) can be instead partitioned into $\homs(\mathcal{C}, \mathcal{A}) = |\Hom(\mathcal{C}, \mathcal{A})|$ (resp. $\homs(\mathcal{C}, \mathcal{B})$) many (non-empty) equivalence classes. In particular, we see homomorphisms $f$ in $\Hom(\mathcal{C}^D, \mathbf{B}(\mathcal{A}, k))$ (resp. $\Hom(\mathcal{C}^D, \mathbf{B}(\mathcal{B}, k))$) as suitably defined \textit{extensions} of homomorphisms $h$ in $\Hom(\cC, \cA)$ (resp. $\Hom(\cC, \cB)$), in a way that every homomorphism $f$ is the extension of a \textit{unique} homomorphism $h$. In particular, every extension $f$ also has a \textit{unique} \textit{colour} $\mathfrak{Col}$ (which is a sequence of stable colours produced by $\rcr$, indexed by the nodes of $T$), associated with it, in which case we say that $f$ is $\mathfrak{Col}$-coloured. Given a colouring $\mathfrak{Col}$, we show that $\rcr$-equivalence for $\sigma$-structures $\mathcal{A}, \mathcal{B}$ implies that it is possible to construct a bijective mapping between homomorphisms $h \in \Hom(\cC, \cA)$ and $h' \in \Hom(\cC, \cB)$, where $h, h'$ can be extended into $\mathfrak{Col}$-coloured extensions. The above step can be iteratively implemented for the remaining homomorphisms, considering new colours $\mathfrak{Col}^*$. In particular, it is guaranteed by construction that the aforementioned mapping, does not map homomorphisms considered at step $i$ of the procedure to homomorphisms considered at step $j < i$, from which in turn follows that it is indeed a bijection, completing the proof of the direction $(2) \implies (1)$.

\subsection{Future Work: Towards Fractional Hypertreewidth}
In view of Theorem~\ref{maintheorem:boundedGHDforStructures} we also ask whether there exists an analogue concerning other width measures, with the most prominent candidate being \textit{fractional hypertreewidth}. In Appendix~\ref{sec:Fractional}, we provide a preliminary answer to this question, that is, we give a first characterisation of a \emph{fractional} version of a relational $k$-WL algorithm---which we call $\rcrFrac$---in regard to homomorphism-indistinguishability over the class of relational structures with \textit{pure} fractional hypertreewidth at most $k$. Here \emph{pure} fractional hypertreewidth is a restricted version of fractional hypertreewidth, a restriction that first appeared in the context of queries \cite[Definition 3.1]{GLS02}, and also referred to as a strong hypertree decomposition~\cite[Definition~9]{GOTTLOB2003775}, that helps us to overcome technical obstructions intrinsic to (not necessarily pure) fractional hypertree decompositions. We establish the following.

\begin{theorem}\label{maintheorem:boundedFHDforStructures}
Let $\sigma$ be a signature and let $k$ be a positive integer. For any two $\sigma$-structures $\mathcal{A}, \mathcal{B}$, the following are equivalent:
\begin{enumerate}
\item $\rcrFrac$ distinguishes $\mathcal{A}$ and $\mathcal{B}$;
\item There is a connected $\sigma$-structure $\mathcal{C}$ that admits a pure fractional hypertree decomposition of width at most $k$ such that $\#\Hom(\mathcal{C}, \mathcal{A}) \neq \#\Hom(\mathcal{C}, \mathcal{B})$.\qed
\end{enumerate}    
\end{theorem}

We leave the strengthening of the previous result to not necessarily pure fractional hypertree width, as well as possible equivalence results for even wider classes of structures, such as classes of bounded \emph{adaptive width}, for future work.

\subsection{Organisation of the Paper}
We introduce preliminary concepts and definitions in Section~\ref{sec:prelims}. We introduce the profiled version $\rcr$ and prove Theorem~\ref{maintheorem:boundedGHDforStructures} in Section~\ref{sec:Generalised}. Finally, in Section~\ref{sec:rcr_equivalence}, we show the equivalence of $\rcr$ and $\hmrcr$, which will then yield the proof of our main result (Theorem~\ref{maintheorem:boundedGHDforStructures_neq}).

Our treatment of the preliminary result on fractional hypertreewidth (Theorem~\ref{maintheorem:boundedFHDforStructures}) as well as the proof of an analogue of a technical lemma of~\cite{scheidt2026colorrefinementrelationalstructures} (which is needed in Section~\ref{sec:Generalised}) are deferred to the appendix.

\newpage

\section{Preliminaries}\label{sec:prelims}

\paragraph*{Notation}
For the cardinality of a finite set $S$, we write $|S|$ or $\#S$. For an ordered tuple $\bar{t}$ and an element $x$, we write $x \in \bar{t}$ to denote that $x$ appears in $\bar{t}$ and we also write $\mathsf{set}(\bar{t})$ for the set of elements that appear in $\bar{t}$ (note that tuples allow for multiple occurrences of the same element). We sometimes write $t$ instead of $\bar t$, if it is clear from context that $\bar t$ is a tuple. A \emph{multiset} is a set that allows repetitions of elements and can be otherwise seen as a tuple without any ordering. We use double brackets $\{\{\}\}$ to denote a multiset as in $\{\{a, b, b, a\}\}$. We may sometimes refer to the number of occurrences of an element $x$ in a multiset or a tuple as the \emph{multiplicity} of $x$.

We assume familiarity with standard concepts of computational complexity. We emphasize that we sometimes use $\tilde{O}$ notation instead of $O$ notation, whenever we wish to omit polylogarithmic terms, for the sake of simplicity. E.g., we write $\tilde{O}(n)$ instead of $O(n\log(n)^c)$, where $c$ is any constant.

\paragraph*{Signatures, Relational Structures and the Gaifman Graph}
A \emph{signature} is a finite set of symbols $R_1, \dots, R_\ell$ which are called \textit{relation} symbols. Each relation symbol $R_i, i \in [\ell]$ is associated with a positive integer $\mathsf{ar}(R_i)$ which is known as the \emph{arity} of the relation symbol $R_i$. The \emph{arity} of a signature $\mathsf{\sigma}$, denoted by $\mathsf{ar}(\sigma)$, is the maximum arity among the relation symbols in $\sigma$.

A \emph{relational structure $\mathcal{A}$ over a signature $\sigma$} (or simply a $\sigma$-structure $\mathcal{A}$) is a pair consisting of a set $A$ of elements --- known as the \emph{domain} of $\mathcal{A}$ --- and a collection of \emph{relations} $\{R_i^\mathcal{A}\}_{i \in [\ell]}$, where $R_i \in \sigma$ and $R_i^\mathcal{A} \subseteq A^{\mathsf{ar}(R_i)}$, for each $i \in [\ell]$. In other words, each relation $R_i^\mathcal{A}$ consists of \emph{ordered} tuples of $\mathsf{ar}(R_i)$ (not necessarily distinct) elements of $A$. We also refer to the domain-elements of $\mathcal{A}$ as the \emph{vertices} of $\mathcal{A}$.

We will also refer to \emph{coloured tuples} (instead of mere tuples) and write $R(\bar{t})$ instead of $\bar{t}$ (given that $\bar{t} \in R^{\mathcal{A}})$ such that we can distinguish between relation elements corresponding to the same (uncoloured) tuple. We write $\mathsf{CT}(\mathcal{A})$ for the set of all coloured tuples in $\mathcal{A}$.

Relational structures are denoted by calligraphic characters $\mathcal{A}, \mathcal{B}, \mathcal{C}, ...$ with their domains denoted respectively by $A, B, C, ...$

A substructure $\mathcal{A}'$ of a $\sigma$-structure $\mathcal{A}$ is a $\sigma$-structure such that $A' \subseteq A$ and $R^{\mathcal{A}'} \subseteq R^{\mathcal{A}}$ holds for every relation symbol $R \in \sigma$. A substructure may be \emph{induced} by (1) some set $X \subseteq A$, in which case we take $A' = X$ or (2) a set $Y \subseteq \mathsf{CT}(\mathcal{A})$ in which case we take $\mathsf{CT}(\mathcal{A}') = Y$ and $A'$ consisting of all elements that appear in a tuple in $Y$. We denote the substructure induced by $X \subseteq A$ (resp. $Y \subseteq \mathsf{CT}(\mathcal{A}))$ by $\mathcal{A}[X]$ (resp. $\mathcal{A}[Y]$) and it will be clear from context whether a structure is induced by domain elements or coloured tuples. 

The \emph{Gaifman graph} of a $\sigma$-structure $\mathcal{A}$ is a simple, undirected graph with vertex set given by the domain $A$ of the relational structure and edges formed according to the following rule: two vertices $x, y$ are adjacent if and only if there is a relation $R^\mathcal{A}$ and a tuple $\bar{t} \in R^{\mathcal{A}}$ that contains both $x$ and $y$ (where $R \in \sigma$).

Any $\sigma$-structure $\mathcal{A}$ is also associated to a hypergraph $H(\mathcal{A})$ with vertices given by $A$ and hyperedges formed according to the following rule: $e = (x_1, \dots, x_\ell)$ is a hyperedge of $H$ if and only if there is an ordering $x_{i_1}, \dots, x_{i_\ell}$ of $x_1, \dots, x_\ell$ and a relation symbol $R \in \sigma$ such that $(x_{i_1}, \dots, x_{i_\ell}) \in R^\mathcal{A}$.

\begin{remark}\label{remark:no_isolated_elements}
In this work, we consider only relational structures $\mathcal{A}$ that feature \emph{no} isolated elements, that is, we assume that for each domain element $x \in A$, there is a relation symbol $R \in \sigma$ such that $x \in R^\mathcal{A}$ (where $\sigma$ is the underlying signature). 
\end{remark}

\paragraph*{Homomorphisms Between Relational Structures}
Given two relational structures $\mathcal{A}, \mathcal{B}$ over the same signature $\sigma$, we say that a mapping $h : A \to B$ ---where $A$ (resp. $B$) is the domain of $\mathcal{A}$ (resp. $\mathcal{B}$)--- defines a \emph{homomorphism} (from $\mathcal{A}$ to $\mathcal{B}$) if for every relation symbol $R \in \sigma$ and every tuple $\bar{t} = (t_1, \dots, t_{\mathsf{ar}(R)}) \in R^\mathcal{A}$, it holds $h(\bar{t}) \coloneq (h(t_1), \dots, h(t_{\mathsf{ar}(R)})) \in R^\mathcal{B}$.
We write $\mathsf{Hom}(\mathcal{A}, \mathcal{B})$ for the set of all homomorphisms from $\mathcal{A}$ to $\mathcal{B}$.

A homomorphism $h : \mathcal{A} \to \mathcal{B}$ between two relational structures $\mathcal{A}, \mathcal{B}$ is \emph{injective} if for any two elements $x, y \in A$, $h(x) = h(y) \implies x = y$ holds. A homomorphism $h : \mathcal{A} \to \mathcal{B}$ is \emph{edge-surjective} if for every coloured tuple $R(\bar{t})$ of $\mathcal{B}$, there is a coloured tuple $R(\bar{t}')$ of $\mathcal{A}$ such that $h(\bar{t}') = \bar{t}$. 

Finally, a homomorphism $h : \mathcal{A} \to \mathcal{B}$ between two relational structures \textit{that feature no isolated elements}, that is both injective and edge-surjective is called an \emph{isomorphism} between structures $\mathcal{A}$ and $\mathcal{B}$. Hence, two structures $\mathcal{A}$ and $\mathcal{B}$ (with no isolated vertices) are isomorphic if only if there is an injective and edge-surjective homomorphism between them.

\paragraph*{The $k$-dimensional Weisfeiler-Leman algorithm for binary structures}

Let $\sigma$ be a \emph{binary} signature, that is, every relation symbol of $\sigma$ has arity at most 2. In turn, a $\sigma$-structure $\mathcal{A}$ can be then seen as a directed graph the vertices of which are coloured by the unary relation symbols of $\sigma$ and its (directed) edges are coloured by the binary relation symbols of $\sigma$. We call such a structure $\mathcal{A}$ a \emph{binary structure} or a \emph{coloured multigraph} (borrowing the latter term from \cite{scheidt_et_al:LIPIcs.MFCS.2025.88}).

We present below the $k$-dimensional Weisfeiler-Leman algorithm ($k$-WL for short) for coloured multigraphs. We distinguish between the cases $k = 1$ and $k > 1$ as the algorithm is defined slightly different in the former case compared to the latter case.
 
The 1-dimensional Weisfeiler-Leman algorithm (1-WL for short) is executed on coloured multigraphs $\mathcal{A}$, iteratively assigning colours to the domain elements (henceforth referred to as vertices) of $\mathcal{A}$ according to the following rule: two vertices $x, y$ get different colours at round $i$ if and only if there is a colour $c$, produced at round $i-1$, such that $x$ and $y$ have different number of neighbours coloured with $c$. 

At each round, 1-WL partitions the set of vertices into colour classes which are refined at each round, that is, two vertices $x$ and $y$ may be in the same colour class after iteration $i$ only if they were already in the same colour class before iteration $i$. The algorithm terminates once the colour classes have been stabilised, that is, they remain the same after some iteration, hence 1-WL always terminates after at most $|A|$ iterations. 

For each vertex $v$, we write $\mathsf{atp}(v)$ ---that stands for \emph{atomic type of} $v$--- for the union of all unary relation symbols $R \in \sigma$ such that $v \in R^\mathcal{A}$ \emph{and} all binary relation symbols $R \in \sigma$ such that $(v,v) \in R^\mathcal{A}$. In other words, by seeing $\mathcal{A}$ as a coloured directed graph, then $\mathsf{atp}(v)$ contains the vertex-colours of $v$ along with the edge-colours of the loop $(v,v)$. 

Let $\mathsf{Gaif}(\mathcal{A})$ denote the Gaifman graph of $\mathcal{A}$. Given two vertices $v, w$ such that $\{v, w\} \in E(\mathsf{Gaif}(\mathcal{A}))$, we write $\mu(v, w)$ for the edge labels of the directed edge $(v, w)$.

Formally, the colour that each vertex $v$ gets at round $i$ ---which we denote here by $c^1_i(v)$--- is given as follows:

\[c^1_i(v) = \left\{
\begin{array}{cc}
      \mathsf{atp}(v) & \mbox{if } i = 0 \\
    (c^1_{i-1}(v), \{\{(c^1_{i-1}(w), \mu(v,w)) \mid \{v,w\} \in \mathsf{Gaif}(\mathcal{A})\}\}) & \mbox{if } i > 0
\end{array}
\right.
\]

For $k > 1$, the $k$-dimensional Weisfeiler-Leman algorithm  colours size-$k$ tuples (or, $k$-tuples for short) of vertices. $k$-WL also follows the same idea as 1-WL did in iteratively refining the colour classes (now of $k$-tuples) until it stabilises after at most $|A|^k$ iterations. The main difference lies in that the colours of vertices in 1-WL are decided by the neighbourhood of each vertex, which is not the case for $k > 1$.

For each $k$-tuple $\bar{v} = (v_1, \dots, v_k) \in A^k$, we write $\mathsf{atp}(\bar{v})$ for the \emph{isomorphism type} of the substructure of $\mathcal{A}$ induced by the vertices in $\bar{v}$.
We could simply define $\mathsf{atp}(\bar{v})$ as the substructure 
$\mathcal{A}[\mathsf{set}(\bar{v})]$ but 
since $k$-WL is used for comparing the colour classes produced by the algorithm when executed on two relational structures, we need to consider the isomorphism type of structures that is invariant under relabelling of vertices.

Given a $k$-tuple, an index $j \in [k]$ and a vertex $w$, we write $\bar{v}[w \to j]$ for the $k$-tuple obtained from $\bar{v}$ by replacing its $j$-th entry with $w$.

Formally, the colour that each $k$-tuple $\bar{v}$ gets at round $i$---which we denote by $c_i^k(\bar{v})$--- is given as follows:

\[
c_i^k(\bar{v}) = \left\{
\begin{array}{cc}
    \mathsf{atp}(\bar{v}) & \mbox{if } i = 0 \\
    (c_{i-1}^k(\bar{v}), \{\{\mathsf{ct}(w, i-1, \bar{v}) \mid w \in A\}\}) & \mbox{if } i > 0
\end{array}
\right.
\]

where $\mathsf{ct}(w, \ell, \bar{v}) = (c_{\ell}^k(\bar{v}[w \to 1]), \dots, c_{\ell}^k(\bar{v}[w \to k]))$ is the \emph{colour tuple of $\bar{v}$ and $w$ at round $\ell$}.
\paragraph*{Generalised and Fractional Hypertree Decompositions for Relational Structures}

Typically, tree decompositions of a relational structure decompose either the Gaifman graph of the structure or the hypergraph that is associated with it. In this work, we consider decompositions of the latter case which we slightly refine as we explain in detail momentarily. In essence, such tree decompositions involve a mapping $\lambda$ that maps the nodes of the underlying tree to a subset of hyperedges of the hypergraph of the corresponding relational structure. In our refined definition, $\lambda$ maps tree-nodes to subsets of \emph{coloured tuples} of the corresponding relational structure instead. As we explain below, such a refinement in the definition does not affect the \emph{width} of the decomposition (compared to the width of the standard decomposition).

\begin{definition}[Fractional edge cover based on coloured tuples]\label{def:fracEdgeCoverNumberPrelims}
Fix a signature $\sigma$ and let $\mathcal{A}$ be a $\sigma$-structure. Recall that we write $\mathsf{CT}(\mathcal{A})$ for the set of coloured tuples of $\mathcal{A}$. Let $X \subseteq A$ (where $A$ is the domain of $\mathcal{A}$) and $C \subseteq \mathsf{CT}(\mathcal{A})$. We say that a mapping $\rho : C \to \mathbb{Q}^+$ is a \emph{fractional edge cover}\footnote{Although $\rho$ involves coloured tuples and not edges (as it is the case with fractional edge covers in hypergraphs), we refrain from introducing a new name for the sake of simplicity.} of $X$ if the following inequality holds for every $x \in X$:
\[
\sum_{R(\bar{t}) \in C : x \in \bar{t}}\rho(R(\bar{t})) \geq 1\,.
\]

In the special case in which $\rho$ ranges over $\{0,1\}$, we simply call $\rho$ an \emph{edge cover} since it can be equivalently seen as the set $\{R(\bar{t}) \in C \mid \rho(R(\bar{t})) = 1\}$ that ``covers'' all the elements of $X$.

The \emph{size} of a fractional edge cover $\rho$ is $\rho(C) \coloneq \sum_{R(\bar{t}) \in C}\rho(R(\bar{t}))$. We say that fractional edge cover $\rho$ is \emph{minimum} if it has the smallest size among all fractional edge covers $\rho' : C \to \mathbb{Q}^+$ of $X$.

The \emph{fractional edge cover number} (resp. \emph{edge cover number}) of a relational structure $\mathcal{A}$ is the size of the minimum fractional edge cover $\rho : \mathsf{CT}(\mathcal{A}) \to \mathbb{Q}^+$ (resp. edge cover $\rho : \mathsf{CT}(\mathcal{A}) \to \{0,1\}$) of $A$.
\end{definition}

\begin{definition}[Fractional and generalised hypertree-decompositions of relational structures]\label{def:FHD}
Fix a signature $\sigma$. A \emph{fractional hypertree decomposition} (FHD, for short) of a $\sigma$-structure $\mathcal{A}$ is an ordered triplet $(T, B, \lambda)$ where
\begin{itemize}
\item $T$ is a tree;
\item $B : V(T) \to 2^A$ is a mapping from the nodes of $T$ to subsets of the domain elements;
\item $\lambda$ is a mapping from the nodes of $T$ to sequences (or equivalently, ordered multisets) of coloured tuples in $\CT(\cA)$, where we allow for multiple occurrences of the same coloured tuple,
\end{itemize}
such that the following conditions are met.
\begin{itemize}
\item[(1)] For every coloured tuple $\bar{t} \in \mathsf{CT}(\cA)$, there is a node $v \in V(T)$ such that $\mathsf{set}(\bar{t}) \coloneq  \{x \mid x \in \bar{t}\} \subseteq B(v)$.
\item[(2)] For every domain element $x \in A$, the subtree induced by the nodes $\{v \in V(T) \mid x \in B(v)\}$ is connected.
\item[(3)] For every node $v \in V(T)$, there is a fractional edge cover $\rho_v : \lambda(v) \to \mathbb{Q}^+$ of $B(v)$. \footnote{Note that here, a fractional edge cover is similarly defined for sequences/multisets of coloured tuples.}
\end{itemize} 

The \emph{width} of a FHD is given by $\max_{v \in V(T)}\rho_v(\lambda(v))$. The \emph{fractional hypertreewidth} of $\mathcal{A}$, denoted by $\mathsf{fhw}(\mathcal{A})$, is the \emph{minimum} width over all possible FHDs of $\mathcal{A}$.

In the special case in which condition (3) is restricted to considering \emph{only edge covers} (which is not equivalent to assuming that every minimum fractional edge cover is also an edge cover), then we call $(T, B, \lambda)$ a \emph{generalised hypertree decomposition} (GHD, for short). The width of a GHD is defined similarly to that of a FHD. The \emph{generalised hypertreewidth} of $\mathcal{A}$, denoted by $\mathsf{ghw}(\mathcal{A})$, is the minimum width over all possible GHDs of $\mathcal{A}$.

We also define the following special fractional/generalised hypertree-decompositions.

\begin{itemize}
\item We say that a fractional/generalised decomposition $D = (T, B, \lambda)$ is \emph{full} if for
every coloured tuple $R(\bar a)$ of $\cA$, there is a node $v \in V(T)$ such that ($a$) $\mathsf{set}(\bar a) \subseteq B(v)$ (which matches Condition (1) above) and ($b$) $R(\bar a) \in \lambda(v)$.

\item We say that a fractional/generalised hypertree decomposition $D = (T, B, \lambda)$ is \emph{pure} if for every $v \in V(T)$, it holds $B(v) = \bigcup_{R(\bar{t}) \in \lambda(v)}\mathsf{set}(\bar{t})$.\footnote{Note that, by definition it follows that in every fractional/generalised hypertree decomposition $(T, B ,\lambda)$ and for every $v \in V(T)$ it holds $B(v) \subseteq \bigcup_{R(\bar{t}) \in \lambda(v)}\mathsf{set}(\bar{t})$.}

\item We say that a fractional hypertree decomposition $D = (T, B, \lambda)$ of width $k$ is \emph{semi-pure} if for every $v \in V(T)$, there is a fractional edge cover of $\bigcup_{R(\bar{t}) \in \lambda(v)}\mathsf{set}(\bar{t})$ that is of size at most $k$.\footnote{Note that, this is not necessarily true for every FHD, but it is true for every \textit{pure} FHD.}
\end{itemize}

Accordingly, we define the full/pure/semi-pure fractional/generalised hypertreewidth as the minimum width over all full/pure/semi-pure fractional/generalised hypertree-decompositions.
\end{definition}

\begin{remark}\label{remark:GHDforRelationalStructures}
It can be easily verified that every fractional/generalised hypertree decomposition can be made into a full decomposition of the same width. In particular, for any coloured tuple $R(\bar{t})$ that does not satisfy condition (b) of the definition of a full decomposition, let $v$ be a node such that $\mathsf{set}(\bar{t}) \subseteq B(v)$ (which is guaranteed to exist by condition (1)). Then, we attach a new child-node $v'$ to $v$ where $B(v') = \mathsf{set}(\bar t)$ and $\lambda(v') = (R(\bar{t}))$. Further note that the resulting decomposition has size that is linear in the size of the original one since every tuple belongs to at most $|\sigma|$ relations, and signatures will always be assumed to be finite and fixed. Also, note that if the original decomposition is (semi-)pure, then the modified decomposition is also (semi-)pure.
\end{remark}

\section{The $\rcr$ algorithm and the proof of \Cref{maintheorem:boundedGHDforStructures}}\label{sec:Generalised}

In this section, we first define the profiled version of $\hmrcr$, which we call the $\rcr$ algorithm, or just $\rcr$. Afterwards, we proceed with the proof of \Cref{maintheorem:boundedGHDforStructures}.

\subsection{Technical set-up}\label{sec:orderingSetUp}

We fix a finite relational signature $\sigma$, and let $\ar(\sigma)$ be the maximum arity in $\sigma$. We also \emph{globally} fix a total order $\preceq_\sigma$ on the relation symbols in $\sigma$, and an auxiliary total order $<$ on domain elements (extended to tuples in the usual way, that is, in lexicographic order).
Recall that a \emph{coloured tuple} of $\cA$ is a pair $(\bar a, R)$ where $\bar a \in R^{\cA}$ which we denote by $R(\bar{a})$. The need of introducing the notion of coloured tuples is for appropriate bookkeeping of tuples that occur in multiple relations. Also, recall that we write $\CT(\cA)$ for the set of all coloured tuples in $\cA$.

For a $\sigma$-structure $\cA$, we define the total order $\prec_\cA$ over its coloured tuples as follows:
\[
R(\bar a)\prec_\cA S(\bar b)
\iff
\bigl(R\prec_\sigma S\bigr)\ \lor\ \bigl(R=S\ \land\ \bar a<\bar b\bigr).
\]

\begin{remark}
We consider only relational structures that feature \emph{no} isolated elements, that is, elements that do not appear in any of the tuples of the relational structure.    
\end{remark}


\subsection{The $\rcr$ Algorithm}

Before we define the $\rcr$ algorithm, we recall the following definition from the introduction.

\explodedSignatureDefinition*

Towards defining $\rcr$, fix a signature $\sigma$ and, for a $\sigma$-structure $\cA$, define

\[
\MCT_k(\cA):=
\left\{(\alpha;\bar t_1,\ldots,\bar t_m)\ \middle|\
\alpha=(R_1,\ldots,R_m)\in\sigma^{(\leq k)},\
\bar t_i\in R_i^\cA\right\}.
\]
For $\Omega=(\alpha;\bar t_1,\ldots,\bar t_m)$, set
$\profile(\Omega):=\alpha$, $\flatt(\Omega):=\bar t_1+\cdots+\bar t_m$ (where the $+$ denotes concatenation here), and
$L(\Omega):=|\flatt(\Omega)|$.

There are two concepts, namely \textit{atomic types} and \textit{similarity types}, that are central to the definition of our $\rcr$ algorithm. These notions were introduced in \cite{scheidt_et_al:LIPIcs.MFCS.2025.88} for the definition of their relational Colour Refinement (RCR) algorithm which we extend appropriately here.
\rcrTypesDefinition

We are now ready to state $\rcr$ which is essentially a colouring function of the elements of $\MCT_k(\mathcal{A})$ that is iteratively updated---by incorporating in its encoding of new colours, the colours of the previous round---until it reaches a point where the \textit{colour classes} that it induces do not change, in which case we say that the algorithm \textit{stabilises}.

\kRCRdefinition

\subsection{Towards the Proof of \Cref{maintheorem:boundedGHDforStructures}}
\subsubsection{Binary Structures Induced by Generalised Hypertree-Decompositions}
In this part, we introduce the key technical tool used in the proof of \Cref{maintheorem:boundedGHDforStructures} which is binary structures that are induced by GHDs. Our construction builds on the approach of \cite{scheidt_et_al:LIPIcs.MFCS.2025.88}, where a binary structure is induced by a tree decomposition. We generalise this approach significantly so as to induce appropriate binary structures by generalized hypertree decompositions instead. We gave a brief overview in \Cref{sec:intro-ghd} but let us recall the definitions.

First, we introduce, for each node $u\in V(T)$ two particular representations of the bag $B(u)$ corresponding to $u$. 

\bagTuplesDefinition*

As a second-to-last step, we also recall the underlying signature of the intended binary structure.

\ghdBinaryStructureDefinition*

\subsubsection{$k$-Exploded Binary Encodings}

As we make it explicit momentarily, it is crucial for the proof of \Cref{maintheorem:boundedGHDforStructures} to interpret equivalence of structures under $\rcr$ (in symbols, $\equiv_\rcr$) in terms of equivalence of suitably defined intermediate \textit{binary} structures under 1-WL (in symbols, $\equiv_{1\textup{-WL}}$). In this section, we show that our $\rcr$ algorithm is precisely capable of accommodating such interpretation. 

We define the aforementioned intermediate binary structures, which we call \textit{canonical $k$-exploded encodings}---the underlying signature of which is the $k$-exploded binary signature from \Cref{def:explodedSignature}---as follows.

\explodedEncodingDefinition*

For technical reasons, it will be convenient to sometimes use a more structured version of the $k$-exploded encodings. To that end we also define the following profile-split $k$-exploded encoding.
\begin{definition}[Profile-split $k$-exploded encoding]\label{def:OrderExplodedEncoding}
Define
$\enc(\cA,k)$ as the $\widehat{\sigma}_k$-structure with universe
$\dom(\enc(\cA,k))$ given by
\[
\left\{
(\alpha,\bar a)
\ \middle|\
\alpha=(R_1,\ldots,R_n)\in\sigma^{(\leq k)},
\bar a\in R_1^\cA\times\cdots\times R_n^\cA
\right\}.
\]
For $x=(\alpha,\bar a)$, write $\mathsf{flat}(x):=\bar a$ and $x[i]:=\bar a[i]$; for split elements $x,y$, write
$\stp(x,y):=\{(i,j)\mid\mathsf{flat}(x)[i]=\mathsf{flat}(y)[j]\}$.
For the interpretation set (we omit the superscript $\enc(\cA,k)$ for readability).
\begin{itemize}
\item
$(\alpha,\bar a)\in U_\beta$ iff $\alpha=\beta$.
\item $E_{i,j}((\alpha,\bar a),(\beta,\bar b))$ iff $\bar a[i]=\bar b[j]$.
\end{itemize}
\end{definition}

For the different directions of our characterisation proof it will be more natural to use either the normal or the profile-split variants of the exploded encoding. However, the two notions are equivalent with respect to WL distinguishability and we will use this fact tacitly from here on.

\begin{proposition}
\label{prop:B-enc-1WL}
For all $\sigma$-structures $\cA,\cB$,
\[
\enc(\cA,k)\equiv_{1\textup{-WL}}\enc(\cB,k)
\quad\Longleftrightarrow\quad
\mathbf B(\cA,k)\equiv_{1\textup{-WL}}\mathbf B(\cB,k).
\]
\end{proposition}
\begin{proof}
Fix a $\sigma$-structure $\mathcal D$. For
$\bar a\in\dom(\mathbf B(\mathcal D,k))$, define
\[
P(\bar a):=\{\alpha\in\sigma^{(\leq k)}
\mid(\alpha,\bar a)\in\dom(\enc(\mathcal D,k))\}.
\]
For tuples $\bar a,\bar b$, define
\[
\tau(\bar a,\bar b):=
\{(p,q)\mid p\leq|\bar a|,\ q\leq|\bar b|,\ \bar a[p]=\bar b[q]\}.
\]
Let $\chi_\infty$ and $\widehat\chi_\infty$ be the stable $1$-WL colours
on $\mathbf B(\mathcal D,k)$ and $\enc(\mathcal D,k)$.

We claim that, $\widehat\chi_\infty(\alpha,\bar a)$ and $(\alpha,\chi_\infty(\bar a))$
determine each other.

The direction from right to left is immediate by induction on the refinement
rounds: if $\chi_i(\bar a)$ is known, then $P(\bar a)$ is known, and the
refinement multiset of $(\alpha,\bar a)$ is obtained from the refinement
multiset of $\bar a$ by replacing each neighbour $\bar b$ by all split
copies $(\beta,\bar b)$ with $\beta\in P(\bar b)$.

For the converse, one round of refinement in the split encoding determines
the missing set $P(\bar a)$. Indeed, let
$\eta:=\tau(\bar a,\bar a)$. For every profile $\beta$ with
$|\beta|=|\bar a|$, we have
$\beta\in P(\bar a)$ if and only if the refinement multiset of $(\alpha,\bar a)$ contains an
entry of the form $(\eta,d_{\beta,\eta})$,
where $d_{\beta,\eta}$ denotes the initial split colour with profile
$\beta$ and equality type $\eta$. Such an entry can only be contributed by
$(\beta,\bar a)$: if it is contributed by $(\beta,\bar b)$, then
$\tau(\bar a,\bar b)=\eta$, $|\bar b|=|\bar a|$, and
$\tau(\bar b,\bar b)=\eta$, hence $\bar b=\bar a$.

Thus, after one split-refinement round, the colour of $(\alpha,\bar a)$
determines $P(\bar a)$ and $\tau(\bar a,\bar a)$, i.e. the initial colour
of $\bar a$ in $\mathbf B(\mathcal D,k)$. More generally, suppose that $\widehat\chi_{i+1}(\alpha,\bar a)$ determines
$\chi_i(\bar a)$. Then the refinement multiset of
$(\alpha,\bar a)$ at round $i+1$ determines, for every round-$i$ colour $c$
of $\mathbf B(\mathcal D,k)$ and every edge pattern $\tau$, the number of
tuples $\bar b$ with $\chi_i(\bar b)=c$ and
$\tau(\bar a,\bar b)=\tau$: choose any profile
$\beta_c\in P(\bar b)$ for tuples of colour $c$, and count precisely those
split neighbours $(\beta_c,\bar b)$ whose round-$(i+1)$ split colour
determines $c$. Hence
$\widehat\chi_{i+2}(\alpha,\bar a)$ determines
$\chi_{i+1}(\bar a)$. In the stable colouring, the shift by one round is irrelevant. Hence
$\widehat\chi_\infty(\alpha,\bar a)$ determines $(\alpha,\chi_\infty(\bar a))$,
and the claim follows.

Consequently the stable colour classes of $\enc(\mathcal D,k)$ are exactly
the stable colour classes obtained by splitting each stable colour class
$c$ of $\mathbf B(\mathcal D,k)$ according to the profiles
$\alpha\in P_c$, where $P_c=P(\bar a)$ for any $\bar a$ of colour $c$.
Therefore the stable colour multiplicities of
$\mathbf B(\mathcal D,k)$ and $\enc(\mathcal D,k)$ determine each other.
Applying this to $\mathcal D=\mathcal A$ and $\mathcal D=\mathcal B$ gives the statement.
\end{proof}

We are now ready to state the aforementioned connection between $\rcr$-indistinguishability and 1-WL-indistinguishability.

\begin{proposition}\label{prop:kRCR.1WL}
For all $\sigma$-structures $\cA,\cB$,
\[
\cA\equiv_{\rcr}\cB
\quad\Longleftrightarrow\quad
\enc(\cA,k)\equiv_{1\textup{-WL}}\enc(\cB,k).
\]
\end{proposition}

\begin{proof}
Identify
$\Omega=(\alpha;\bar t_1,\ldots,\bar t_m)\in\MCT_k(\cA)$ with
$(\alpha,\bar t_1+\cdots+\bar t_m)\in\enc(\cA,k)$.
At round zero, both colours record exactly the profile and the equality
pattern of the flattened tuple.  Moreover, for corresponding
$\Omega,\Psi$, the set $\stp(\Omega,\Psi)$ is precisely $\{(i,j)\mid E_{i,j}(\Omega,\Psi)\}$.

Consequently the multiset used in the $\rcr$ update is exactly the
multiset of neighbour colours grouped by the full binary edge-label
pattern in the $1$-WL update.  Induction on the refinement round gives
the same colour classes, with the same multiplicities, and proves the
claim.
\end{proof}

\begin{remark}
Recall, that in the definition of the (version) of the 1-WL algorithm for binary structures that we use in this work (see, \Cref{sec:prelims}), the colouring (of each iteration) considers the edge labels along a directed edge. We are aware that the authors of a related work \cite{scheidt_et_al:LIPIcs.MFCS.2025.88} use another version of 1-WL. In particular, in their version each respective colouring considers for each directed edge $(u, v)$ the disjoint union of the edge labels of \textit{both} directed edges $(u, v)$ and $(v, u)$. We observe that when running the aforementioned algorithms on our canonical $k$-exploded encodings, it is easy to see that the algorithms are in fact equivalent (though we stress that this is not necessarily true for every pair of binary structures), because of the following symmetry: for any edge label $E_{i, j}$ of the directed edge $(u, v)$, the edge $(v, u)$ is labelled by $E_{j, i}$ and vice versa.
\end{remark}

\subsection{Proof of \Cref{maintheorem:boundedGHDforStructures}}

In this section, we formally prove \Cref{maintheorem:boundedGHDforStructures} with which we establish a characterisation of our $\rcr$ algorithm in terms of homomorphism indistinguishability over relational structures of generalised hypertreewidth at most $k$.

For what follows we fix a signature $\sigma$. Let $\mathcal{A}, \mathcal{B}$ be $\sigma$-structures and let $D = (T, B, \lambda)$ be a full GHD of $\mathcal{A}$ of width at most $k$. By \Cref{remark:GHDforRelationalStructures}, fullness can be assumed without loss of generality. The first step of the proof is to show that $\Hom(\cA^D, \mathbf{B}(\mathcal{B}, k))$ can be partitioned into $\homs(\cA, \cB)$-many equivalence classes. To this end, we first describe how to appropriately \textit{extend} any homomorphism $h \in \Hom(\cA, \cB)$ into a homomorphism $\widehat{h} \in \Hom(\cA^D, \mathbf{B}(\mathcal{B}, k))$ which would subsequently allow us to define the desired equivalence classes.

\begin{definition}\label{def:Extensions}
Let $\cA, \cB$ be $\sigma$-structures and let $D = (T, B, \lambda)$ be a full GHD of $\cA$ with width at most $k$. For $u \in V(T)$, write $\bagtup(u) = \bar{a}^*_1 + \dots + \bar{a}^*_\ell$, \footnote{Recall that $\bagtup(u)$ may contain $\varepsilon$-entries.} and $\mathsf{Ord}\lambda(u) = R_1(\bar{a}_1), \dots, R_\ell(\bar{a}_\ell)$, where $1 \leq \ell \leq k$. Given a mapping $h : \dom(\cA) \to \dom(\cB)$ and $u \in V(T)$, let $\exthom{h}{u}$ be a mapping that maps the tuple $\bar{a}^*_1 + \dots + \bar{a}^*_\ell$ entry-wise to some $\bar{t} \in \dom(\mathbf{B}(\cB, k))$ of the same size such that $(1)$ $\exthom{h}{u}$ agrees with $h$ on all entries $x \neq \varepsilon$ of $\bagtup(u)$ and $(2)$ $\exthom{h}{u}$ maps all $\varepsilon$'s to arbitrary elements in $\dom(\cB)$ that need not be equal.

For $h \in \Hom(\cA, \cB)$, we say that the multiset $\{\{\exthom{h}{u} : u \in V(T)\}\}$ $D$-\emph{extends} $h$ w.r.t $\cB$ if and only if for each $u \in V(T)$, it holds $\exthom{h}{u}(\bagtup(u)) \in U^{\mathbf{B}(\cB, k)}_{\profile(u)}$, where $\profile(u) = (R_1, \dots, R_\ell)$. We write $\mathsf{Ext}^D_\cB(h)$ for all $D$-extensions of $h$ (w.r.t $\cB$).    
\end{definition}

\begin{lemma}\label{lem:GHD-non-empty-extension-set}
Let $\cA$ be a $\sigma$-structure and let $D = (T, B, \lambda)$ be a \emph{full} GHD of $\cA$ with width at most $k$, where $k \in \mathbb{N}$. For every $\sigma$-structure $\cB$ and $h \in \Hom(\cA, \cB)$, it holds $\mathsf{Ext}^D_\cB(h) \neq \emptyset$.
\end{lemma}
\begin{proof}
For $u \in V(T)$, write $\fulltup(u) = \bar{a}_1 + \dots + \bar{a}_\ell$, $\bagtup(u) = \bar{a}^*_1 + \dots + \bar{a}^*_\ell$ and $\mathsf{Ord}\lambda(u) = R_1(\bar{a}_1), \dots, R_\ell(\bar{a}_\ell)$, where $1 \leq \ell \leq k$. Let $\exthom{h}{u}$ denote the mapping that maps $\bagtup(u)$ entry-wise to $h(\fulltup(u))$, that is, in particular if $\bar{a}_i^*[j] = \varepsilon$, then $\exthom{h}{u}$ maps $\bar{a}_i^*[j]$ to $h(\bar{a}_i[j])$.

Clearly, $\exthom{h}{u}(\bagtup(u)) \in \mathbf{B}(\cB, k)$. Furthermore, since $h$ is a homomorphism we have $h(\bar{a}_i) \in U_{R_i}^\cB$, for each $1 \leq i \leq \ell$ which in turn implies that $h(\fulltup(u)) \in U_{\profile(u)}^{\mathbf{B}(\cB, k)}$. Finally, since $\exthom{h}{u}(\bagtup(u)) = h(\fulltup(u))$, we also have $\exthom{h}{u}(\bagtup(u)) \in U_{\profile(u)}^{\mathbf{B}(\cB, k)}$ and so $\mathsf{Ext}^D_\cB(h) \neq \emptyset$.
\end{proof}

\begin{lemma}\label{lem:HomSumOfExt}
Let $\cA$ be a $\sigma$-structure and let $D = (T, B, \lambda)$ be a \emph{full} GHD of $\cA$ with width at most $k$.
Then for every $\sigma$-structure $\cB$,
\[
\homs(\cA^{D},\mathbf{B}(\cB,k)) = \sum_{h \in \Hom(\cA, \cB)}|\mathsf{Ext}^D_\cB(h)|\,. 
\]
In particular, $\{\mathsf{Ext}^D_\cB(h)\}_{h \in \Hom(\cA, \cB)}$ partition $\Hom(\cA^D, \mathbf{B}(\cB, k))$ into $\homs(\cA, \cB)$ equivalence classes.
\end{lemma}
\begin{proof}

Let $h \in \Hom(\cA, \cB)$ and assume that a given multiset $\{\{\exthom{h}{u} : u \in V(T)\}\}$ of mappings $D$-extends $h$ w.r.t. $\cB$. Note that $\{\{\exthom{h}{u} : u \in V(T)\}\}$ induces a mapping $\hat{h} : \cA^D \to \mathbf{B}(\cB, k)$, where $\hat{h}(u) = \exthom{h}{u}(\bagtup(u))$. We have
\begin{enumerate}
\item for each $u \in V(T)$, $\hat{h}(u) = \exthom{h}{u}(\bagtup(u)) \in U^{\mathbf{B}(\cB, k)}_{\profile(u)}$, following from the definition of a $D$-extension of $h$;
\item for each $(u, v)$ such that $u = v$ or $\{u, v\} \in E(T)$, if $(u, v) \in E_{i,j}^{\cA^D}$, then $(\hat{h}(u), \hat{h}(v)) \in E_{i, j}^{\mathbf{B}(\cB, k)}$. To see this, recall that $(u, v) \in E_{i,j}^{\cA^D}$ if and only if $\bagtup(u)[i] = \bagtup(v)[j] \neq \varepsilon$. Equivalently, $\bagtup(u)[i] = \bagtup(v)[j] \in B(u) \, \cap \, B(v)$. Since $\exthom{h}{u}, \exthom{h}{v}$ agree on $B(u)\,\cap\,B(v)$, we deduce that $(\hat{h}(u), \hat{h}(v)) = (\exthom{h}{u}(\bagtup(u)), \exthom{h}{v}(\bagtup(v))) \in E_{i, j}^{\mathbf{B}(\cB, k)}$. 
\end{enumerate}
Hence, we deduce that $\hat{h} \in \Hom(\cA^D, \mathbf{B}(\cB, k))$ which in turn implies that there is a natural mapping $\pi : \bigcup_{h \in \Hom(\cA, \cB)}\mathsf{Ext}^D_\cB(h) \to \Hom(\cA^D, \mathbf{B}(\cB, k))$ that maps $\{\{\exthom{h}{u} : u \in V(T)\}\}$ to $\hat{h}$ which is injective by definition. 

Also, it is clear that if $\cA$ has no isolated domain elements (see our assumption in \Cref{remark:no_isolated_elements}), then for any two $h, h' \in \Hom(\cA, \cB)$ such that $\mathsf{Ext}^D_\cB(h) \,\cap\, \mathsf{Ext}^D_\cB(h') \neq \emptyset$, we have $h = h'$. To see this, note that for each $u \in V(T)$, the mapping of an intersecting $D$-extension, that corresponds to $u$, must agree on $B(u)$ with both $h$ and $h'$. Since, $\bigcup_{u \in V(T)}B(u) = \dom(\cA)$, we deduce that $h = h'$. Hence, the sets $\{\mathsf{Ext}^D_\cB(h)\}_{h \in \Hom(\cA, \cB)}$ are pair-wise disjoint.

Next let $\hat{g} \in \Hom(\cA^D, \mathbf{B}(\cB, k))$. Recall that by assumption, each element $z \in \dom(\cA)$ is contained in some tuple of $\cA$ and thus it is also contained in some bag $B(u)$ of $T$. For each $z \in \dom(\cA)$, we fix a node $u_z \in V(T)$ such that $z \in B(u_z)$. Recall that if $\bagtup(u_z)[i] = \bagtup(u_z)[j] \neq \varepsilon$ then we have $(u_z, u_z) \in E^{\cA^D}_{i,j}$. Since $\hat{g}$ is a homomorphism, it also follows that $(\hat{g}(u_z), \hat{g}(u_z)) \in E^{\mathbf{B}(\cB, k)}_{i,j}$ implying that there is $x_z \in \dom(\cB)$ such that for any index $i$ with $\bagtup(u_z)[i] = z$, we have $\hat{g}(u_z)[i] = x_z$. We consider the well-defined mapping $\rho : \dom(\cA) \to \dom(\cB)$ that maps each $z \in \dom(\cA)$ to its corresponding element $x_z \in \dom(\cB)$ obtained from the procedure described above.

\begin{claim}\label{claim:SurjectivityNonPure}
For $u \in V(T)$ with $\bar{a} = \bagtup(u)$ such that $\hat{g}(u) = \bar{d} \in \dom(\mathbf{B}(\cB, k))$, we have $\rho(\bar{a}[i]) = \bar{d}[i]$, whenever $\bar{a}[i] \neq \varepsilon$.    
\end{claim}

\begin{proof}
Let $z \in \bar{a}$ (where $z \neq \varepsilon$). Recall that we have fixed $u_z \in V(T)$ such that $z \in B(u_z)$. By the definition of $D$, it follows that for the path $(u_0, u_1, \dots, u_{q-1}, u_q)$ in $T$ connecting $u_0 = u$ and $u_q = u_z$, it holds that $z \in B(u_j)$, for each $0 \leq j \leq q$. For each $0 \leq j \leq q$, write $\bar{t}_j  = \bagtup(u_j)$ and let $i_j$ denote any index such that $\bar{t}_j[i_j] = z$. Then, it easy to see that for each $0 \leq j < q$, we have by definition that $(u_j, u_{j+1}) \in E^{\cA^D}_{i_j,i_{j+1}}$ which implies that we also have $(\hat{g}(u_j), \hat{g}(u_{j+1})) \in E^{\mathbf{B}(\cB, k)}_{i_j, i_{j+1}}$, since $\hat{g}$ is a homomorphism. Concretely we have,
\begin{enumerate}
\item $\bar{t}_q[i_q] = \bar{t}_0[i_0] = \bar{a}[i_0] = z$;
\item $\hat{g}(u_0)[i_0] = \hat{g}(u_z)[i_q]$;
\item $\hat{g}(u_z)[i_q] = x_z$ due to the fact that whenever $\bagtup(u_z)[i] = \bar{t}_q[i] = z$, we have $\hat{g}(u_z)[i] = x_z$.
\end{enumerate}
Hence, $\bar{d}[i_0] \overset{\textrm{def}}= \hat{g}(u_0)[i_0] \overset{(2)}= \hat{g}(u_z)[i_q] \overset{(3)}= x_z \overset{\textrm{def}}= \rho(\bar{a}[i_0])$ which also shows the claim.
\end{proof}

\begin{claim}\label{claim:isHom}
$\rho \in \Hom(\cA, \cB)$.    
\end{claim}
\begin{proof}
Let $R \in \sigma$ and $\bar{t} \in R^\cA$. Since $D$ is full, there is $u \in V(T)$ such that $R(\bar{t}) \in \lambda(u)$ and $\mathsf{set}(\bar{t}) \subseteq B(u)$. Assume that $R(\bar{t})$ is the $i$-th coloured tuple in $\mathsf{Ord}\lambda(u) = R_1(\bar{a}_1), \dots, R_\ell(\bar{a}_\ell)$, where $1 \leq \ell \leq k$. Note that $\bar{a}_i^* = \bar{a}_i = \bar{t}$ since $\mathsf{set}(\bar{t}) \subseteq B(u)$. Let $\profile(u) = (R_1, \dots, R_i, \dots, R_\ell)$ where $R_i = R$. Since, $\hat{g} \in \Hom(\cA^D, \mathbf{B}(\cB, k))$, we have $\hat{g}(u) \in R_1^{\cB} \times \dots \times R^\cB \times \dots \times R_\ell^\cB$ which in turn by \Cref{claim:SurjectivityNonPure} implies that $\rho(\bar{t}) \in R^\cB$ which completes the proof.
\end{proof}

By the combination of \Cref{claim:SurjectivityNonPure,claim:isHom} we deduce that $\hat{g}$ naturally induces $\{\{\exthom{\rho}{u} : u \in V(T)\}\} \in \mathsf{Ext}^D_\cB(\rho)$ such that $\pi(\{\{\exthom{\rho}{u} : u \in V(T)\}\}) = \hat{g}$, which shows that $\pi$ is also surjective. 

Hence, $\#\Hom(\cA^{D},\mathbf{B}(\cB,k)) = |\bigcup_{h \in \Hom(\cA, \cB)}\mathsf{Ext}^D_\cB(h)| = \sum_{h \in \Hom(\cA, \cB)}|\mathsf{Ext}^D_\cB(h)|$, where the last equality follows from the fact that the sets $\{\mathsf{Ext}^D_\cB(h)\}_{h \in \Hom(\cA, \cB)}$ are pair-wise disjoint ---as argued earlier--- and non-empty as shown in \Cref{lem:GHD-non-empty-extension-set}.
\end{proof}

\subsubsection{$\rcr$-equivalence implies homomorphism indistinguishability over structures of generalised hypertreewidth at most $k$}

In this section, we prove the direction $(2) \implies (1)$ of \Cref{maintheorem:boundedGHDforStructures} that also comprises our main novel technical contribution, formally stated in the following lemma.

\begin{lemma}\label{lem:connectedHomInd}
Let $\cA$ be a connected $\sigma$-structure and let $D = (T, B, \lambda)$ be a \emph{full} GHD of $\cA$ with width at most $k$. Let $\cB, \cB'$ be $\sigma$-structures such that $\cB \equiv_{\rcr} \cB'$. Then, $\homs(\cA, \cB) = \homs(\cA, \cB')$.
\end{lemma}

\begin{proof}
We assume wlog that $\cB, \cB'$ have disjoint domains. By \Cref{prop:kRCR.1WL,prop:B-enc-1WL}, we deduce that $\mathbf B(\cB,k)$ and
$\mathbf B(\cB',k)$ are $1$-WL-equivalent. Let $\chi$ denote the stable colouring produced by 1-WL when run on $\mathbf B(\cB, k) \,\bigcup\, \mathbf B(\cB', k)$. 

Let $\cC \in \{\mathbf B(\cB, k), \mathbf B(\cB', k)\}$.  Let $t$ be a tuple in the domain of $\cC$ and let $\alpha$ be a profile such that $t \in (U_\alpha)^{\cC}$. Let also $x$ be a sequence of positive integers $n \geq 1$ appearing in strictly increasing order such that the maximum integer appearing in $x$ is at most the length of $t$. We call $x$ a \textit{grounding} of $t$. For such $t, x, \alpha$, we define $\mathfrak g(t, x, \alpha)$ as the set consisting of all tuples $s$ in the domain of $\cC$ such that $s \in (U_\alpha)^{\cC}$ and for every $i \in x$, we have $(i, i) \in \stp(t, s)$. We call $\mathfrak g(t, x, \alpha)$ the \textit{group of $t$ w.r.t. $x$ and $\alpha$}. 

\begin{observation}
It is easy to see that, for any $s \in \mathfrak g(t, x, \alpha)$, we also have $t \in \mathfrak g(s, x, \alpha)$.    
\end{observation}

We call a set $P$ of pairs of positive integers a \textit{similarity pattern}, if there exist tuples $a, b$ such that $P = \stp(a, b)$. Furthermore, for a stable colour $d \in \mathsf{Im}(\chi)$, where $\mathsf{Im}(\chi)$ is the image of $\chi$, we write $d \implies \alpha$ to denote that for a tuple $q$, $\chi(q) = d \implies \alpha \in (U_\alpha)^\cC$. It can be verified that, 
\begin{equation}\label{eq:gtx}
|\mathfrak g(t, x, \alpha)| = \sum_{\substack{d \in \mathsf{Im}(\chi) \\ d \implies \alpha}}\left(\sum_{P \supseteq \{(i, i) : i \in x\}}\Big|\{s \in \cC \mid \chi(s) = d, \stp(t, s) = P\}\}\Big|\right)\,.
\end{equation}

\begin{claim}\label{claim:auxForGroups}
Let $\cC, \cC' \in \{\mathbf B(\cB, k), \mathbf B(\cB', k)\}$ and let $t \in \cC, t' \in \cC'$ be tuples with $\chi(t) = \chi(t')$. For every stable colour $d$ in the image of $\chi$ and every similarity pattern $P$ satisfying $\{(i, i) : i \in x\} \subseteq P$, \[\Big|\{s \in \cC \mid \chi(s) = d, \stp(t, s) = P\}\Big| = \Big|\{s' \in \cC' \mid \chi(s') = d, \stp(t', s') = P\}\Big|\,,\] 

Furthermore, for any profile $\alpha$ such that $t \in (U_\alpha)^\cC, t' \in (U_\alpha)^{\cC'}$ and any grounding $x$ of $t$ and $t'$, 
\[|\mathfrak g(t, x, \alpha)| = |\mathfrak g(t', x, \alpha)|\,.\]
\end{claim}
\begin{proof}
The first equation follows directly from combining that $\chi(t) = \chi(t')$ and that $\chi$ is a stable colouring. The second equation follows by combining the first equation of the claim with \Cref{eq:gtx}. 
\end{proof}

Let $h \in \Hom(\cA, \cB)$. By \Cref{lem:GHD-non-empty-extension-set}, we know that $\mathsf{Ext}^D_\cB(h) \neq \emptyset$. Let $f^h \in \mathsf{Ext}^D_\cB(h)$. For each $u \in V(T)$, let $x_u$ be a sequence with elements $\{i : \bagtup(u)[i] \neq \varepsilon\}$ in increasing order. Set $\alpha_u \coloneq \profile(u)$. For each $u \in V(T)$, set $t_u \coloneq f^h(u)$. Note, that for each $u \in V(T)$, $x_u$ is a grounding of $t_u$. Take $S = \{s_u\}_{u \in V(T)} \in \varprod_{u \in V(T)}\mathfrak g(t_u, x_u, \alpha_u)$. We define the mapping $f^h_S : u \mapsto s_u$.
 
\begin{claim}\label{claim:allExtensions}
We have, $f^h_S \in \mathsf{Ext}^D_\cB(h)$. In particular, we have $\mathsf{Ext}^D_\cB(h) = \{f^h_S : S \in \varprod_{u \in V(T)}\mathfrak g(t_u, x_u, \alpha_u)\}$. 
\end{claim}
\begin{proof}
To see that $f^h_S \in \mathsf{Ext}^D_\cB(h)$, note that since for each $u \in V(T)$, $s_u$ keeps the entries on indices $i \in x_u$ unchanged, it suffices to show that $s_u \in (U_{\alpha_u})^{\mathbf B(\cB, k)}$, which is guaranteed by construction. Hence, $\{f^h_S : S \in \varprod_{u \in V(T)}\mathfrak g(t_u, x_u, \alpha_u)\} \subseteq \mathsf{Ext}^D_\cB(h)$.

To see that $\mathsf{Ext}^D_\cB(h) = \{f^h_S : S \in \varprod_{u \in V(T)}\mathfrak g(t_u, x_u, \alpha_u)\}$, note that no extension of $h$ can map a node $u$ to a tuple $q$ outside of $\mathfrak g(t_u, x_u, \alpha_u)$ as this would imply that at least one of the following are true (1) $q$ does not agree with the restriction $h|_{B(u)}$ on the indices specified by $x_u$ (2) $q \notin (U_{\alpha_u})^{\mathbf B(\cB, k)}$.    
\end{proof}

Note that a similar analysis as the one above can also be shown w.r.t. $\cB'$ and $\mathbf B(\cB', k)$. Next, we define a mapping $\xi : \Hom(\cA^D, \mathbf B(\cB, k)) \to \Hom(\cA^D, \mathbf B(\cB', k))$, the specifications of which will be made clear momentarily.

Let $u \in V(T)$. Let $d$ be a stable colour in the image of $\chi$ such that for a tuple $t \in \mathbf B(\cB, k)$, $\chi(t) = d \implies t \in (U_{\alpha_u})^{\mathbf B(\cB, k)}$, where recall that $\alpha_u = \profile(u)$. Let $Q(d)$ denote the family of all tuples in $\mathbf B(\cB, k)$ with stable colour $d$. We consider the following ordering $\leq^d_u$ w.r.t. $d$ and $u$ of tuples in $Q(d)$. We wish all tuples within the same group $\mathfrak g(t, x_u, \alpha_u)$ to appear ''next to each other'' in the ordering. So, we internally order all groups $\mathfrak g(t, x_u, \alpha_u)$ according to $\prec_\cB$ and then we order the (internally ordered) groups with each other based on their common string induced by the indices appearing in the grounding $x_u$ (again using $\prec_\cB$). Similarly, we define $Q'(d)$ w.r.t. $\mathbf B(\cB', k)$. For ease of notation, we use the same symbol $\leq^d_u$ for the ordering associated with $Q'(d)$. Clearly, $|Q(d)| = |Q'(d)|$. Furthermore, recall that according to \Cref{claim:auxForGroups}, for any two tuples $t, t'$ of the same colour, their groups contain the same number of tuples of any pre-specified stable colour. So, $Q(d)$ is ``grouped'' into groups of the same size, and similarly $Q'(d)$ is ``grouped'' into equally many groups as $Q(d)$.

Towards defining $\xi$, given $h \in \Hom(\cA, \cB)$ and $f \in \mathsf{Ext}^D_\cB(h)$, we traverse $T$ in a \emph{breadth-first search} (BFS) fashion, assuming that is rooted at node $r$, and perform the following:

\begin{enumerate}
\item Set $t_r \coloneq f(r)$. Let $i_r$ denote the index of $t_r$ within the colour-class $Q(\chi(t_r))$ with respect to the ordering $\leq^r_{\chi(t_r)}$. We let $s_r$ denote the $i_r$-th element of $\mathbf{B}(\cB', k)$ within the corresponding colour-class $Q'(\chi(t_r))$, with respect to the corresponding ordering $\leq^r_{\chi(t_r)}$; such $i_r$ exists since $\chi$ is a stable colouring and so the aforementioned colour classes have the same size.
\item Let $u$ be a child of $r$ and set $t_u \coloneq f(u)$. Since $\cA$ is connected, we have $B(r) \,\cap\, B(u) \neq \emptyset$,\footnote{Note that $B(r) \,\cap\,B(u)$ is a separator of $\cA$. Hence, if $B(r) \,\cap\, B(u) = \emptyset$, then $\cA$ has an empty separator and hence it is not connected, which is a contradiction.} which in turn implies that $\mathsf{stp}(t_r, t_u) \neq \emptyset$.

Let $P(r, u) \subseteq \stp(t_r, t_u)$ denote the non-empty similarity pattern that consists of all pairs of indices $(i, j)$ such that $\bagtup(r)[i] = \bagtup(u)[j] \neq \varepsilon$. Consider the ordered families $Q(d), Q'(d)$ of all tuples of colour $d \coloneq \chi(t_u)$. Let $q \in Q(d)$ be such that $\stp(t_r, q) \not\supseteq P(r, u)$. We filter-out $q$. We follow exactly the same process filtering-out elements from $Q'(d)$. More specifically, we filter out tuples from $Q(d), Q'(d)$ in parallel as follows: first note that for any $q \in Q(d)$, we have $|\{p \in Q(d) : \stp(t_r, p) = \stp(t_r, q)\}| = |\{p \in Q'(d) : \stp(s_r, p) = \stp(t_r, q)\}|$, since $\chi(t_r) = \chi(s_r)$ and $\chi$ is a stable colouring. Hence, for a given $q \in Q(d)$ that should be filtered-out, we may filter-out simultaneously the same number of tuples from each family. Also, the remaining tuples $p \in Q(d)$ (resp. $p' \in Q'(d))$ satisfy $\stp(t_r, p) \neq \stp(t_r, q)$ (resp. $\stp(s_r, p') \neq \stp(t_r, q)$). Hence, we may repeat this process, taking another tuple $q_2 \in Q(d)$ that should be filtered-out, and then filter-out all tuples $p$ (resp. $p'$) that satisfy $\stp(t_r, p) = \stp(t_r, q_2)$ (resp. $\stp(s_r, p') = \stp(t_r, q_2)$), and so on. The resulting families of tuples have the same size. Furthermore, note that, if $q$ is filtered out then the entire group it belongs to is filtered-out as well, which means that each group of $Q(d)$ (resp. $Q'(d)$) is either completely filtered-out or remains intact. Hence, the resulting families have the same number of groups. Let $i_u$ denote the index at which $t_u$ appears in the filtered family $Q(d)$. Let $s_u$ denote the $i_u$-th tuple appearing in the filtered family $Q'(d)$. We compute $s_v$ for each child $v$ of $r$ as described above.

\item We repeat step (2) for every node that we visit in the BFS traversal.

\item After we have visited every node $u \in V(T)$ and computed $s_u$, we output the mapping $\mathsf{out}_f : u \mapsto s_u $.
\end{enumerate}

We write $\xi$ for the mapping that maps $f$ to $\mathsf{out}_f$.

\begin{claim}\label{claim:OutHom}
$\mathsf{out}_f \in \Hom(\cA^D, \mathbf{B}(\cB', k))$.    
\end{claim}
\begin{proof}
Let $u \in V(T)$ and let $\alpha$ such that $u \in U_\alpha^{\cA^D}$. We have $t_u = f(u) \in U_\alpha^{\mathbf{B}(\cB, k)}$. Since $t_u, \mathsf{out}_f(u)$ have been assigned the same stable colour (by construction), it follows from \Cref{def:kRCR} that $\mathsf{out}_f(u) \in U_\alpha^{\mathbf{B}(\cB', k)}$. The same reasoning yields that for any $i, j$ such that $(u, u) \in E_{i,j}^{\cA^D}$, we have $(\mathsf{out}_f(u), \mathsf{out}_f(u)) \in E_{i, j}^{\mathbf{B}(\cB', k)}$.

Next, consider $\{u, v\} \in E(T)$ and assume that $v$ is the parent of $u$ in $T$ rooted at $r$. By construction, we have  $P(v, u) = \{(i, j) : (v, u) \in E_{i,j}^{\cA^D}\}$ and $P(v, u) \subseteq \stp(\mathsf{out}_f(v), \mathsf{out}_f(u))$. We also have $\{(i, j) : (u, v) \in E_{i,j}^{\cA^D}\} \subseteq \mathsf{stp}(\mathsf{out}_f(u), \mathsf{out}_f(v))$, yielding that $\mathsf{out}_f \in \Hom(\cA^D, \mathbf{B}(\cB', k))$.
\end{proof}

\begin{claim}
Let $h \in \Hom(\cA, \cB)$ and let $f, \tilde f \in \mathsf{Ext}^D_\cB(h)$. For $u \in V(T)$, set $t_u \coloneq f(u), \tilde t_u \coloneq \tilde f(u)$. Assume that for each $u$, $f(u)$ and $\tilde f(u)$ have the same stable colour and $t_u \in \mathfrak g(\tilde t_u, x_u, \alpha_u)$. Set also $s_u \coloneq \xi(f)(u), \tilde s_u \coloneq \xi(\tilde f)(u)$. Then, for each $u \in V(T)$, it holds that $s_u \in \mathfrak g(\tilde s_u, x_u, \alpha_u)$ and vice versa. 
\end{claim}
\begin{proof}
Let $u \in V(T)$ and let $d \coloneq \chi(t) = \chi(t')$. To show the claim, it suffices to show that when we filter out tuples from $Q(d)$ (resp. $Q'(d)$) at node $u$, then our algorithm filters out exactly the same tuples, when we execute the corresponding step for $f$ and $\tilde f$ respectively. It is clear, that is true for the root $r$ of $T$. Let $u$ be an internal node and let $v$ denote its parent. Since $\tilde t_v \in \mathfrak{g} (t_v, x_v, \alpha_v)$, $\tilde t_u \in \mathfrak{g} (t_u, x_u, \alpha_u)$, it follows that the similarity type $P(v, u)$ is exactly the same in both executions of the corresponding step for $f, \tilde f$. To be more precise, for any $q \in Q(d)$, $\stp(t_v, q) \supseteq P(v, u)$ if and only if $\stp(\tilde t_v, q) \supseteq P(v, u)$. Finally, since in addition $t_v, \tilde t_v$ have the same stable colour and $t_u, \tilde t_u$ have the same stable colour as well, we deduce that all steps of the algorithm are executed for both mappings $f, \tilde f$ as desired.   
\end{proof}

Let $\mathfrak{Col} \in \varprod_{u \in V(T)}\mathsf{Im}(\chi)$ be a sequence of stable colours indexed by $u \in V(T)$. For such a sequence $\mathfrak{Col}$ and $\cC \in \{\mathbf B(\cB, k), \mathbf B(\cB', k)\}$, we say that a homomorphism $f \in \Hom(\cA^D, \cC)$  is $\mathfrak{Col}$-coloured, if for each $u \in V(T)$, $\chi(f(u)) = \mathfrak{Col}_u$. Let $H(\mathfrak{Col})$ (resp. $H'(\mathfrak{Col})$) denote the set of all homomorphisms in $\Hom(\cA, \cB)$ (resp. $\Hom(\cA, \cB')$) such that each $h \in H(\mathfrak{Col})$ (resp. $h' \in H'(\mathfrak{Col})$) satisfies that there exists a $\mathfrak{Col}$-coloured $D$-extension of $h$ (resp. $h'$).
Towards finishing the proof, we observe that a mapping $\xi$ induces a bijection $\rho : H(\mathfrak{Col}) \to H'(\mathfrak{Col})$, such that for any $h \in H(\mathfrak{Col})$ and any two $\mathfrak{Col}$-coloured extensions $f_1, f_2$ of $h$, we have $\xi(f_1), \xi(f_2)$ are $\mathfrak{Col}$-coloured extensions of $\rho(h)$. We may continue this process, and extend the mapping $\rho$ to the remaining homomorphisms $h \in \Hom(\cA, \cB) \setminus H(\mathfrak{Col})$, by considering a new colour-sequence $\mathfrak{Col^*}$ and establish in a similar fashion a bijection between $H(\mathfrak{Col^*}) \setminus H(\mathfrak{Col})$ and $H'(\mathfrak{Col^*}) \setminus H'(\mathfrak{Col})$ and so on, using the following claim, which follows by the construction of $\xi$ in combination with \Cref{claim:auxForGroups,claim:allExtensions} and states that no extension of a homomorphism $\phi \in H(\mathfrak{Col^*}) \setminus H(\mathfrak{Col})$ is mapped by $\xi$ to an extension of a homomorphism in $H(\mathfrak{Col})$ (that is, of a homomorphism that we have already ``hit'' while defining $\rho$), completing the proof.

\begin{claim}
Let $h, g \in \Hom(\cA, \cB)$ such that $g = \rho(h)$. For every sequence of stable colours $\mathfrak{Col}$ (indexed by $u \in V(T)$, we have \[\Big|\{f \mid f \in \mathsf{Ext}^D_\cB(h), \text{$f$ is $\mathfrak{Col}$-coloured}\}\Big| = \Big|\{f \mid f \in \mathsf{Ext}^D_\cB(g), \text{$f$ is $\mathfrak{Col}$-coloured}\}\Big|\,.\]    
\end{claim}
\end{proof}

\subsubsection{Homomorphism indistinguishability over structures of generalised hypertreewidth at most $k$ implies $\rcr$-equivalence}

In this part, we prove the direction $(1) \implies (2)$ of \Cref{maintheorem:boundedGHDforStructures}. Our proof is a suitable adaptation of the proof of \cite{scheidt2026colorrefinementrelationalstructures} for the same direction in the case $k = 1$, which essentially corresponds to the relational colour refinement algorithm. However, extending the definition of inducing binary structures by tree-decompositions (as it was the case in \cite{scheidt2026colorrefinementrelationalstructures}) to inducing binary structures by hypertree-decompositions (as it is in our case) is non-trivial. Given that these binary structures are substantial for the argumentation, we do provide a self-contained proof.

For the proof, we will need an auxiliary lemma on full and pure GHDs which we state below (and the proof of which is deferred to \Cref{lem:panosproof.ghd.hom.eq}).

\begin{lemma}[Analogue of Lemma 4.5 of \cite{scheidt2026colorrefinementrelationalstructures}]
\label{lem:ghd.hom.eq}
Let $\cA$ be a $\sigma$-structure and let $D$ be a \emph{full} and \emph{pure} GHD of $\cA$ with width at most $k$.
Then for every $\sigma$-structure $\cB$,
\[
\homs(\cA^{D},\enc(\cB,k))=
\homs(\cA^{D},\mathbf{B}(\cB,k)) = \homs(\cA,\cB)\,.
\]
\end{lemma}

The map
\[
g\longmapsto\bigl(u\mapsto(\profile(u),g(u))\bigr)
\]
is a bijection from $\Hom(\cA^D,\mathbf B(\cB,k))$ to
$\Hom(\cA^D,\enc(\cB,k))$; its inverse maps $(\alpha,\bar b)$ to $\bar b$.
This proves the first equality.

Before we proceed with the proof, we recall a result in \cite{Dvorak10} and \cite{dell_et_al:LIPIcs.ICALP.2018.40}  that establish 1-WL-equivalence (equivalently, Colour Refinement-equivalence) for binary structures in terms of homomorphism indistinguishability over \textit{acyclic} binary structures. 

\begin{theorem}[\cite{Dvorak10,dell_et_al:LIPIcs.ICALP.2018.40}]\label{thm:ClassicalColorRefinement}
Let $\sigma$ be a finite signature and $\cA, \cB$ be binary $\sigma$-structures. The following are equivalent.
\begin{enumerate}
\item Color Refinement distinguishes $\cA$ and $\cB$;
\item There exists binary $\sigma$-structure $\mathcal{T}$, the Gaifman graph of which is a tree, such that $\#\Hom(\mathcal{T}, \cA) \neq \#\Hom(\mathcal{T}, \cB)$.
\end{enumerate}
\end{theorem}

\begin{lemma}\label{lem:reverseDirectionGHD}
Let $\cB, \cB'$ be two $\sigma$-structures such that $\cB \not\equiv_{\rcr}\cB'$. There exists a $\sigma$-structure $\cA$ and a pure full GHD $D$ for $\cA$ with width at most $k$ such that $\#\Hom(\cA^D,\enc(\cB,k))\neq\#\Hom(\cA^D,\enc(\cB',k))$. Consequently, $\#\Hom(\cA,\cB)\neq\#\Hom(\cA,\cB')$.
\end{lemma}

\begin{proof}
The proof works in a similar fashion as the proof of \cite[Lemma 4.6]{scheidt2026colorrefinementrelationalstructures}. However there are some key technical differences, and for this reason we provide a self-contained proof. 

Recall that $\enc(\cB,k), \enc(\cB',k)$ are $\hat{\sigma}_k$-structures (see \Cref{def:explodedSignature}). By our assumption and \Cref{prop:kRCR.1WL} it follows that $\enc(\cB,k) \not\equiv_{1\textup{-WL}} \enc(\cB',k)$, which due to \Cref{thm:ClassicalColorRefinement} in turn implies that there is a $\hat{\sigma}_k$-structure $\mathcal{T}$, the Gaifman graph of which is a tree ---which we denote by $T$--- such that $\#\Hom(\mathcal{T}, \enc(\cB,k)) \neq \#\Hom(\mathcal{T}, \enc(\cB',k))$. Since every element of either target belongs to exactly one unary relation (see \Cref{def:OrderExplodedEncoding}), every element of each print defined below has a unique unary profile.

Similarly to the proof of \cite[Lemma 4.6]{scheidt2026colorrefinementrelationalstructures}, we define for each $h \in \Hom(\mathcal{T}, \enc(\cB,k))$ the \emph{print} $\mathsf{Print}(h)$ of $h$ which is a $\hat{\sigma}_k$-structure with the same domain as $\mathcal{T}$ and relations which are given as follows:

\begin{enumerate}
\item[(i)] For each $v \in \dom(\mathcal{T})$ and $\alpha \in \sigma^{(\leq k)}$, we have $v \in (U_\alpha)^{\mathsf{Print}(h)}$ if and only if $h(v) \in (U_\alpha)^{\enc(\cB,k)}$;
\item[(ii)] For each $(u, v) \in E(T) \,\bigcup \,\{(w,w) \mid w \in V(T)\}, i \in [|\mathsf{flat}(h(u))|], j \in [|\mathsf{flat}(h(v))|]$, we have $(u, v) \in (E_{i,j})^{\mathsf{Print}(h)}$ if and only if $(h(u), h(v)) \in (E_{i,j})^{\enc(\cB,k)}$.
\end{enumerate}

We also define the print of each $h \in \Hom(\mathcal{T}, \enc(\cB',k))$ in the same fashion. Note that the Gaifman graph of every print is $T$.

Since $\hat{\sigma}_k$ is a binary signature, it follows from the proof of \cite[Lemma 4.6]{scheidt2026colorrefinementrelationalstructures} that there exists a print $Q \coloneq \mathsf{Print}(h_Q)$ where $h_Q \in \Hom(\mathcal{T}, \enc(\cB,k)) \bigcup \Hom(\mathcal{T}, \enc(\cB',k))$ such that $\#\Hom(Q, \enc(\cB,k)) \neq \#\Hom(Q, \enc(\cB',k))$.

We now proceed with the construction of the claimed $\sigma$-structure $\cA$ and decomposition $D = (T, B, \lambda)$ (where $T$ is the Gaifman graph of $\mathcal{T}$).
In particular, we show that the $\hat{\sigma}_k$-structures $\cA^D$ and $Q$ are isomorphic, which will in turn imply that $\#\Hom(\cA^D, \enc(\cB,k)) \neq \#\Hom(\cA^D, \enc(\cB',k))$. We also show that $D$ is by construction a full and pure GHD and hence due to \Cref{lem:ghd.hom.eq} we derive the last claim of our statement, that is, $\#\Hom(\cA, \cB) \neq \#\Hom(\cA, \cB')$.

To this end, let $\mathsf{\Omega}$ be an countably infinite set such that $\mathsf{\Omega} \,\cap\, (\dom(\cB) \bigcup \dom(\cB')) = \emptyset$. We designate a node $r$ as the root of $T$ and see the rest of the nodes as being directed away from the root. We traverse $T$ in a \textit{breadth-first search} (BFS) fashion and perform the following for each node $v$ that we visit:

\begin{enumerate}
    \item [I.] We write $\alpha_v = (R_1, \dots, R_{\ell_v}) \in \sigma^{(\leq k)}$ for the unique profile for which $v \in (U_{\alpha_v})^Q$ holds, where $\ell_v \in [k]$. We introduce $\ell_v$ tuples $\bar{t}(v;1), \dots, \bar{t}(v;\ell_v)$ where $\bar{t}(v;i) \in \mathsf{\Omega}^{\mathsf{ar}(R_i)},$ for each $i \in [\ell_v]$. Letting $\bar{t}(v) \coloneq \bar{t}(v;1) + \dots + \bar{t}(v;\ell_v)$, we enforce $\mathsf{stp}(\bar{t}(v), \bar{t}(v)) = \{(i, j) \mid (v, v) \in (E_{i, j})^Q\}\}$ which is always possible since $\{(i, j) \mid (v, v) \in (E_{i, j})^Q\}\} = \mathsf{stp}(h_Q(v), h_Q(v))$ and we further ensure that no entry of $\bar{t}(v)$ appears as an entry in a tuple corresponding to any node which we have already visited.
    \item[II.] Let $p_v$ denote the parent of $v$ in $T$ (assuming that $v \neq r$). Then, for each $(i, j) \in \mathsf{stp}(h_Q(p_v), h_Q(v))$ we replace $\bar{t}(v)[j]$ with $\bar{t}(p_v)[i]$ and also replace with $\bar{t}(p_v)[i]$ the content of every other entry $m$ of $\bar{t}(v)$ such that $(v, v) \in (E_{j, m})^Q$. We write $\mathsf{stp}_{u, v}$ for $\mathsf{stp}(\bar t_u, \bar t_v)$, where $\bar t_u, \bar t_v$ as given in this step.
\end{enumerate}

\begin{claim}\label{claim:InvariantSTP}
For each $\{u, v\} \in E(T) \,\bigcup \, \{\{w, w\} : w \in V(T)\}$, it holds $\mathsf{stp}(\bar{t}(u), \bar{t}(v)) = \mathsf{stp}_{u, v}$.
\end{claim}
\begin{proof}
Recall that $\mathsf{stp}_{u, v} = \mathsf{stp}(h_Q(u), h_Q(v))$. Hence, our goal is to show that $\mathsf{stp}(\bar t_u, \bar t_v)$ is invariant under step (II). Let $v$ be a node and let $u$ denote its parent.

Recall that before the execution of step (II) for node $v$, we have $\mathsf{set}(\bar{t}_u) \,\cap\, \mathsf{set}(\bar{t}_v) = \emptyset$. Let $(i, j) \in \mathsf{stp}(h_Q(u), h_Q(v))$. Then, by construction we clearly have $(i, j) \in \mathsf{stp}(\bar{t}_u, \bar{t}_v)$. Assume that there is $i'$ such that $(i', j) \in \mathsf{stp}(\bar{t}_u, \bar{t}_v)$, which together with $(i, j) \in \mathsf{stp}(\bar{t}_u, \bar{t}_v)$ implies that $(i, i') \in \mathsf{stp}(\bar{t}_u, \bar{t}_u)$. We have $(i, i') \in \mathsf{stp}(h_Q(u), h_Q(u))$ and so $(i', j) \in \mathsf{stp}(h_Q(u), h_Q(v))$. Furthermore, for any $j'$ such that $(j, j') \in \mathsf{stp}(h_Q(v), h_Q(v))$ and $i'$ such that $(i, i') \in \mathsf{stp}(h_Q(u), h_Q(u))$, it follows that $(i', j') \in \mathsf{stp}(h_Q(u), h_Q(v))$. Hence, we deduce that while executing step (II), we do not add to $\mathsf{stp}(\bar{t}_u, \bar{t}_v)$ any tuples other than the ones in $\mathsf{stp}_{u, v} = \mathsf{stp}(h_Q(u), h_Q(v))$.

Also, by definition, it follows that while executing step (II), we do not add to $\mathsf{stp}(\bar{t}_v, \bar{t}_v)$ any tuples other than the ones in $\mathsf{stp}(h_Q(v), h_Q(v))$, which completes the proof.
\end{proof}

Next, we define our claimed $\sigma$-structure $\cA$ as follows: 
\begin{enumerate}
\item[(a.)] $\dom(\cA) = \bigcup_{v \in V(T)}\mathsf{set}(\bar{t}(v))$;
\item[(b.)] for each $R \in \sigma$, $R^\cA = \{\bar{a} \mid \text{there are $v \in V(T), i \in [\ell_v]$ such that $\alpha_v[i] = R$ and $\bar{t}(v;i) = \bar{a}$}\}$.
\end{enumerate}

We also consider the decomposition $D = (T, B, \lambda)$ where for each $v \in V(T)$ we define $B(v) \coloneq \mathsf{set}(\bar{t}(v))$ and $\lambda(v) \coloneq (R_i(\bar{t}(v;i)))_{i \in [\ell_v]}$, where $R_i = \alpha_v[i]$. By construction $D$ is full and pure. Hence, for $D$ to be a valid GHD of $\cA$, it remains to show that the connectivity condition is also met. To this end, let $x \in \dom(\cA)$ and $u, v \in V(T)$ such that $\bar{t}(u)[i] = \bar{t}(v)[j] = x$, for some $i, j$. By construction, for any node $w$ and its parent $p_w$ we have that the entries of $\bar{t}(w)$ that appear outside of $T_w$, where $T_w$ is the subtree of $T$ rooted at $w$, must also appear in $B(p_w)$. Hence, $u$ and $v$ must have a least common ancestor $w$ such that $x \in B(w)$ and thus $x$ must also be contained in the bag of every node in the path from $w$ to $u$ as well as in the path from $w$ to $v$. Thus, the subtree of $T$ induced by $\{v \in V(T) \mid x \in B(v)\}$ is connected for each $x \in \dom(\cA)$.

Finally, we show that $\cA^D$ is indeed isomorphic to $Q$. For this, we recall our previous observation according to which by the construction of $\bar{t}(v), v \in V(T)$ it follows that for each $\{u, v\} \in E(T) \,\bigcup \, \{\{w, w\} : w \in V(T)\}$, we have $\mathsf{stp}(\bar{t}(u), \bar{t}(v)) = \mathsf{stp}(h_Q(u), h_Q(v)) = \{(i, j) \mid (u, v) \in (E_{i, j})^Q\}$. Furthermore, it follows by definition that for each $v \in V(T)$ we have $v \in (U_{\alpha_v})^Q$ and $\profile(v) = \alpha_v$ which implies that $v \in (U_{\alpha_v})^{\cA^D}$. Hence we deduce that $\cA^D$ is isomorphic to $Q$, which completes the proof.

\end{proof}

\subsubsection{Putting the pieces together}
We may now prove \Cref{maintheorem:boundedGHDforStructures} as follows.

\begin{proof}[Proof of \Cref{maintheorem:boundedGHDforStructures}]
Direction (1) $\implies $ (2) follows from \Cref{lem:reverseDirectionGHD}. In particular, \Cref{lem:reverseDirectionGHD} does not state that $\cA$ must be connected, however it can be readily verified that since $\#\Hom(\cA, \cB) \neq \#\Hom(\cA, \cB')$, there must already exist a (maximal) connected substructure of $\cA$ with different number of homomorphisms to $\cB$ and $\cB'$ respectively.

Finally, we derive direction $(2) \implies (1)$ by contraposition. To this end, assume that $\cB \equiv_\rcr \cB'$. By \Cref{remark:GHDforRelationalStructures}, we may choose the decomposition to be full without increasing its width. Then, by \Cref{lem:connectedHomInd} it follows that for any connected $\sigma$-structure $\cA$ that admits a GHD of width at most $k$, it holds $\homs(\cA, \cB) = \homs(\cA, \cB')$, which completes the proof.
\end{proof}

\section{Equivalence of \ensuremath{\rcr} and \ensuremath{\hmrcr}}\label{sec:rcr_equivalence}

In this section we show that $\hmrcr$ and $\rcr$ have the same distinguishing power, that is, two structures are distinguished by $\hmrcr$ if and only if they are distinguished by $\rcr$. In a nutshell, our approach consists of defining some special equivalence relations over the elements coloured by $\hmrcr$ (resp. $\rcr$), which we call \textit{robust} (resp. \textit{pro-robust}) that correspond to refinements of the equivalence relation associated with the stable colouring produced by $\hmrcr$ (resp. $\rcr$), and introducing an appropriate type of \textit{isomorphism} between robust and pro-robust equivalence relations that will eventually be used to show  that $\rcr$ and $\hmrcr$ have the same distinguishing power.

\subsection{Robust and pro-robust equivalence classes}
We fix a finite signature $\sigma$ and a positive integer $k$. For a $\sigma$-structure $\cA$, we write $\mathcal{T}_k(\cA)$ for the set of all tuples that we may obtain by concatenating up to $k$ tuples from $\cA$. At times, we refer to such tuples as \emph{flat} tuples.

\begin{definition}[robust equivalence relations]
Let $\cA, \cB$ be two $\sigma$-structures, for which we assume that they have disjoint domains. We say that an equivalence relation $\equiv$ over flat tuples $\mathcal{T}_k(\cA) \,\bigcup \, \mathcal{T}_k(\cB)$ is \emph{robust} if for any $\cC, \cC' \in \{\cA, \cB\}$ and any two tuples $s \in \mathcal{T}_k(\cC), s' \in \mathcal{T}_k(\cC')$ satisfying $s \equiv s'$, the following hold.
\begin{itemize}
    \item $(\atp(s), \stp(s, s)) = (\atp(s'), \stp(s', s'))$;
    \item $\{\{(\stp(s, t), [t]_{\equiv}) \mid t \in \mathcal{T}_k(\cC) : \stp(s, t) \neq \emptyset\}\} = \{\{(\stp(s', t'), [t']_{\equiv}) \mid t' \in \mathcal{T}_k(\cC') : \stp(s', t') \neq \emptyset\}\}$\,, where for a tuple $t \in \mathcal{T}_k(\cA)\,\bigcup\,\mathcal{T}_k(\cB)$, we write $[t]_{\equiv}$ for the (unique) equivalence class of $\equiv$ that $t$ belongs to.
\end{itemize}
\end{definition}

\begin{proposition}\label{prop:RobustAuxProp}
Let $\cA, \cB$ be two $\sigma$-structures, for which we assume that they have disjoint domains. Let $\hmcol_\infty$ denote the stable colouring produced by $\hmrcr$ when run on $\mathcal{T}_k(\cA)\,\bigcup\,\mathcal{T}_k(\cB)$. The following are true.
\begin{enumerate}
    \item The equivalence relation $\equiv_{\infty}$ induced by $\hmcol_\infty$ is robust;
    \item Furthermore, $\equiv_\infty$ is the coarsest robust equivalence relation, that is, every other robust equivalence relation is a subset of $\equiv_\infty$
\end{enumerate}
\end{proposition}

\begin{proof}
The first part of the proposition follows directly from the definitions of robust equivalence relations and stable colourings. We proceed with the proof of the second part. For each positive integer $i \geq 0$, let $\equiv_i$ denote the equivalence relation that is induced by the round-$i$ colouring produced by $\hmrcr$ when run on $\mathcal{T}_k(\cA)\,\bigcup\, \mathcal{T}_k(\cB)$. For every equivalence relation $\equiv$, we show by induction on $i$, that $\equiv$ is a \textit{refinement} of $\equiv_i$, that is, every equivalence class of $\equiv$ is a subset of some equivalence class of $\equiv_i$. Furthermore, there is $i^*$ such that $\equiv_{i^*} \,=\, \equiv_\infty$. Hence, we may deduce that $\equiv$ is a refinement of $\equiv_\infty$, showing the claim. What remains is to show the aforementioned induction, the details of which are carried out as follows.
\medskip

\noindent \textbf{Base case $i = 0$:}  Since $\equiv$ is robust, for any two flat tuples $s, s'$ satisfying $s \equiv s'$, we have $(\atp(s), \stp(s, s)) = (\atp(s'), \stp(s', s'))$. By definition of $\hmrcr$, we also have $s \equiv_0 s'$, and so $\equiv$ is a refinement of $\equiv_0$.
\medskip

\noindent \textbf{Induction hypothesis:} We assume that for every $i < j$, $\equiv$ refines $\equiv_i$. We show, that $\equiv$ also refines $\equiv_j$. To this end, let $s, s'$ be two flat tuples satisfying $s \equiv s'$. Let $\cC, \cC' \in \{\cA, \cB\}$ be such that $s \in \mathcal{T}_k(\cC), \mathcal{T}_k(\cC')$. We need to show that

\begin{multline*}
([s]_{\equiv_{j-1}}, \{\{(\stp(s, t), [t]_{\equiv_{j-1}}) \mid t \in \mathcal{T}_k(\cC) : \stp(s, t) \neq \emptyset\}\}) = \\
([s']_{\equiv_{j-1}}, \{\{(\stp(s', t'), [t']_{\equiv_{j-1}}) \mid t' \in \mathcal{T}_k(\cC') : \stp(s', t') \neq \emptyset\}\}).
\end{multline*}

By the induction hypothesis, we have $[s]_{\equiv_{j-1}} = [s']_{\equiv_{j-1}}$, therefore it suffices to show that there is a bijection $\mu : \{t \in \mathcal{T}_k(\cC) : \stp(s, t) \neq \emptyset\} \to \{t' \in \mathcal{T}_k(\cC') : \stp(s', t') \neq \emptyset\}$ such that for each $t$ in the domain of $\mu$ we have $(\stp(s, t), [t]_{\equiv_{j-1}}) = (\stp(s', \mu(t)), [\mu(t)]_{\equiv_{j-1}})$.

Since $\equiv$ is robust, it follows that there is a bijection $\eta : \{t \in \mathcal{T}_k(\cC) : \stp(s, t) \neq \emptyset\} \to \{t' \in \mathcal{T}_k(\cC') : \stp(s', t') \neq \emptyset\}$ such that for each $t$ in the domain of $\eta$ we have $(\stp(s, t), [t]_\equiv) = (\stp(s', \eta(t), [\eta(t)]_\equiv)$. By the induction hypothesis, we have $t \equiv \eta(t) \implies t \equiv_{j-1} \eta(t)$. Hence, we may take $\mu \coloneq \eta$.

\end{proof}
Recall that we write $\MCT_k(\cA)$ for the set of all elements $(\alpha, t)$ where $\alpha = (R_1, \dots, R_m) \in \sigma^{(\leq k)}$, for $m \leq k$ and $t \in R_1^{\cA} \times \ldots \times R_m^{\cA}$. Note, that we may also see $t$ as a flat tuple, since we may retrieve the tuples it consists of using $\alpha$. We call the elements of $\MCT_k(\cA)$ \emph{pro-tuples}.

\begin{definition}[pro-robust equivalence relations]
Let $\cA, \cB$ be two $\sigma$-structures, for which we assume that they have disjoint domains. We say that an equivalence relation $\approx$ over pro-tuples $\MCT_k(\cA) \,\bigcup \, \MCT_k(\cB)$ is \emph{pro-robust} if for any $\cC, \cC' \in \{\cA, \cB\}$ and any two pro-tuples $(\alpha, s) \in \MCT_k(\cC), (\alpha', s') \in \MCT_k(\cC')$ satisfying $(\alpha, s) \approx (\alpha', s')$, the following hold.
\begin{itemize}
    \item $\alpha = \alpha' \text{ and } \stp(s, s) = \stp(s', s')$;
    \item $\{\{(\stp(s, t), [(\beta, t)]_{\approx}) \mid (\beta, t) \in \MCT_k(\cC) : \stp(s, t) \neq \emptyset\}\} = \{\{(\stp(s', t'), [(\beta', t')]_{\approx}) \mid (\beta', t') \in \MCT_k(\cC') : \stp(s', t') \neq \emptyset\}\}$\,, where for a pro-tuple $(\beta, t) \in \MCT_k(\cA)\,\bigcup\,\MCT_k(\cB)$, we write $[(\beta, t)]_{\approx}$ for the (unique) equivalence class of $\approx$ that $(\beta, t)$ belongs to. \qedhere
\end{itemize}
\end{definition}

\begin{proposition}\label{prop:ProrobustAuxProp}
Let $\cA, \cB$ be two $\sigma$-structures, for which we assume that they have disjoint domains. Let $\rcrcol_\infty$ denote the stable colouring produced by $\rcr$ when run on $\MCT_k(\cA)\,\bigcup\,\MCT_k(\cB)$. The following are true.
\begin{enumerate}
    \item The equivalence relation $\approx_{\infty}$ induced by $\rcrcol_\infty$ is pro-robust;
    \item Furthermore, $\approx_\infty$ is the coarsest pro-robust equivalence relation, that is, every other pro-robust equivalence relation is a subset of $\approx_\infty$.  \qedhere
\end{enumerate}
\end{proposition}

\begin{proof}
The proof follows verbatim the proof of \Cref{prop:RobustAuxProp}, but we provide a proof for completeness.

The first part of the proposition follows directly from the definitions of robust equivalence relations and stable colourings. We proceed with the proof of the second part. For each positive integer $i \geq 0$, let $\approx_i$ denote the equivalence relation that is induced by the round-$i$ colouring produced by $\rcr$ when run on $\MCT_k(\cA)\,\bigcup\, \MCT_k(\cB)$. For every equivalence relation $\approx$, we show by induction on $i$, that $\approx$ is a \textit{refinement} of $\approx_i$, that is, every equivalence class of $\approx$ is a subset of some equivalence class of $\approx_i$. Furthermore, there is $i^*$ such that $\approx_{i^*} \,=\, \approx_\infty$. Hence, we may deduce that $\approx$ is a refinement of $\approx_\infty$, showing the claim. What remains is to show the aforementioned induction, the details of which, are carried out as follows.
\medskip

\noindent \textbf{Base case $i = 0$:}  Since $\approx$ is robust, for any two pro-tuples $(\alpha, s), (\alpha', s')$ satisfying $(\alpha, s) \approx (\alpha', s')$, we have $\alpha = \alpha', \stp(s, s) = \stp(s', s'))$. By definition of $\rcr$, we also have $(\alpha, s) \approx_0 (\alpha', s')$, and so $\approx$ is a refinement of $\approx_0$.
\medskip

\noindent \textbf{Induction hypothesis:} We assume that for every $i < j$, $\approx$ refines $\equiv_i$. We show, that $\approx$ also refines $\approx_j$. To this end, let $(\alpha, s), (\alpha', s')$ be two flat tuples satisfying $(\alpha, s) \approx (\alpha, s')$. Let $\cC, \cC' \in \{\cA, \cB\}$ be such that $(\alpha, s) \in \MCT_k(\cC), \MCT_k(\cC')$. We need to show that

\begin{multline*}
([(\alpha, s)]_{\approx_{j-1}}, \{\{(\stp(s, t), [(\beta, t)]_{\approx_{j-1}}) \mid (\beta, t) \in \MCT_k(\cC) : \stp(s, t) \neq \emptyset\}\}) = \\
([(\alpha', s')]_{\approx_{j-1}}, \{\{(\stp(s', t'), [(\beta', t')]_{\approx_{j-1}}) \mid (\beta', t') \in \MCT_k(\cC') : \stp(s', t') \neq \emptyset\}\})
\end{multline*}

By the induction hypothesis, we have $[(\alpha, s)]_{\approx_{j-1}} = [(\alpha', s')]_{\approx_{j-1}}$, therefore it suffices to show that there is a bijection $\mu : \{(\beta, t) \in \MCT_k(\cC) : \stp(s, t) \neq \emptyset\} \to \{(\beta', t') \in \MCT_k(\cC') : \stp(s', t') \neq \emptyset\}$ such that for each $(\beta, t)$ in the domain of $\mu$ we have $(\stp(s, t), [(\beta, t)]_{\approx_{j-1}}) = (\stp(s', \mu((\beta, t))_2), [\mu(t)]_{\approx_{j-1}})$, where for a tuple $(x, y)$ we define $(x, y)_2 \coloneq y$.

Since $\approx$ is robust, it follows that there is a bijection $\eta : \{(\beta, t) \in \MCT_k(\cC) : \stp(s, t) \neq \emptyset\} \to \{(\beta', t') \in \MCT_k(\cC') : \stp(s', t') \neq \emptyset\}$ such that for each $(\beta, t)$ in the domain of $\eta$ we have $(\stp(s, t), [(\beta, t)]_\approx) = (\stp(s, \eta((\beta, t))_2), [\eta((\beta, t))]_\approx)$. By the induction hypothesis, we have $(\beta, t) \approx \eta((\beta, t)) \implies (\beta, t) \approx_{j-1} \eta((\beta, t))$. Hence, we may take $\mu \coloneq \eta$.
\end{proof}

\subsection{Isomorphisms between robust and pro-robust equivalence relations}

Again, consider two $\sigma$-structures $\cA, \cB$ for which we assume that they have disjoint domains.

\begin{definition}[Isomorphisms between robust and pro-robust equivalence relations]
Let $\mathfrak{P}$ be a mapping that maps robust equivalence relations to pro-robust equivalence relations. We say that $\mathfrak{P}$ is an \textit{isomorphism} if the following conditions are met.
\begin{enumerate}
    \item $\mathfrak{P}$ is a bijection.
    \item For any two robust equivalence relations $\equiv$ and $\equiv'$, we have that $\equiv \, \subseteq \, \equiv'$ if and only if $\mathfrak{P}(\equiv) \, \subseteq \, \mathfrak{P}(\equiv')$.
\end{enumerate}
\end{definition}

\paragraph*{Construction of an isomorphism} Let $\equiv$ be a robust equivalence relation over $\mathcal{T}_k(\cA) \,\bigcup\, \mathcal{T}_k(\cB)$. We construct an equivalence relation $\equiv^+$ over pro-tuples $\MCT_k(\cA) \,\bigcup \, \MCT_k(\cB)$ such that for any two pro-tuples $(\alpha, s), (\alpha', s')$ we have:

\[(\alpha, s) \equiv^+ (\alpha', s') \text{ if and only if } \alpha = \alpha' \text { and } s \equiv s'\,.\]
It is straightforward to verify that $\equiv^+$ is indeed an equivalence relation over pro-tuples.

\begin{proposition}
The equivalence relation $\equiv^+$ is pro-robust.
\end{proposition}
\begin{proof}
Let $(\alpha, s), (\alpha', s')$ be two pro-tuples of $\cC, \cC'$ respectively, where $\cC, \cC' \in \{\cA, \cB\}$, such that $(\alpha, s) \equiv^+ (\alpha, s')$. By construction of $\equiv^+$, we have $\alpha = \alpha'$ and $s \equiv s'$. Since $\equiv$ is robust, we deduce that $\stp(s, s) = \stp(s', s')$.

What remains is to show that there is a bijection $\xi : \{(\beta, t) \in \MCT_k(\cC) \mid \stp(s, t) \neq \emptyset\} \to \{(\beta', t') \in \MCT_k(\cC') \mid \stp(s', t') \neq \emptyset\}$ such that for each $(\beta, t)$ in the domain of $\xi$:
\[
(\stp(s, t), [(\beta, t)]_{\equiv^+}) = (\stp(s', \xi((\beta, t))_2), [\xi((\beta, t))]_{\equiv^+})\,,
\]
where for a tuple $(x, y)$ we define $(x, y)_2 = y$.

Since $\equiv$ is robust and $s \equiv s'$, it follows that there is a bijection $\mu : \{t \in \mathcal{T}_k(\cC) \mid \stp(s, t) \neq \emptyset\} \to \{t' \in \mathcal{T}_k(\cC') \mid \stp(s', t') \neq \emptyset\}$ such that for each $t$ in the domain of $\mu$:
\begin{equation}\label{eq:mu}
(\stp(s, t), [t]_\equiv) = (\stp(s', \mu(t)), [\mu(t)]_\equiv)\,.
\end{equation}

We define $\xi : (\beta, t) \mapsto (\beta, \mu(t))$. To show that $\xi$ is well-defined, we first show the following claim.

\begin{claim}\label{claim:sameProfiles}
\[
\{\gamma : (\gamma, t) \in \MCT_k(\cC)\} = \{\delta : (\delta, \mu(t)) \in \MCT_k(\cC')\}
\]
\end{claim}
\begin{proof}
Let $t$ be the concatenation of $N$ tuples $p_1, \dots, p_N$ of $\cC$. Since $t \equiv \mu(t)$, it follows that for each $i \in [N]$ there is $w_i$, such that $(\stp(t, p_i), [p_i]_\equiv) = (\stp(\mu(t), w_i), [w_i]_\equiv)$. Note that $\atp(w_i) = \atp(p_i) \neq \emptyset$ and so for each $1 \leq i \leq N$, $w_i$ has the same length as $p_i$. Let $\mathsf{len}(\cdot)$ denote a mapping that outputs the length of an input tuple. Furthermore, we have \[\bigcup_{z \in [\sum_{j=1}^{i-1}\mathsf{len}(p_j) + 1, \sum_{j=1}^{i}\mathsf{len}(p_j)]}(z, z - \sum_{j=1}^{i-1}\mathsf{len}(p_j)) \subseteq \stp(t, p_i) = \stp(\mu(t), w_i)\,,\] which implies that $\mu(t)$ can be decomposed into $N$ tuples $q_1, \dots, q_N$ such that $\atp(p_i) = \atp(q_i)$, for each $i \in [N]$, completing the proof (since the above are true for an arbitrary partitioning of $t$ into $p_1, \dots, p_N$).
\end{proof}

From the claim above follows that $\xi$ is well-defined and that is also a bijection. Finally, using \Cref{eq:mu}, we deduce that $\stp(s, t) = \stp(s', \mu(t)) \, (= \stp(s', \xi((\beta, t))_2))$ and since $t \equiv \mu(t)$, in combination with the claim above we also deduce that $(\beta, t) \equiv^+ (\beta, \mu(t)) \, (= \xi((\beta, t)))$, which completes the proof.
\end{proof}

Hence, the mapping $\bij : \equiv \, \mapsto \, \equiv^+$ is a well-defined mapping from robust equivalence relations to pro-robust equivalence relations. Next, we show that $\bij$ is bijective.

\begin{lemma}\label{lem:bijectionEqRel}
$\bij$ is a bijection.
\end{lemma}
\begin{proof}
We first show that $\bij$ is injective. To this end, let $\equiv_1, \equiv_2$ be two robust equivalence relations, such that $\equiv_1 \, \neq \, \equiv_2$. Hence, there exist tuples $t \in \mathcal{T}_k(\cC), t' \in \mathcal{T}_k(\cC')$, where $\cC, \cC' \in \{\cA, \cB\}$ such that $t \equiv_1 t'$ and $t \not\equiv_2 t'$. Observe that $t \equiv_1 t'$ implies that $t, t'$ have the same length $\ell$. To see this, let $\ell_1, \ell_2$ denote the lengths of $t, t'$ respectively. Note that, our assumption $t \equiv_1 t'$ implies that $\stp(t, t) = \stp(t', t')$. Furthermore, $(\ell_1, \ell_1) \in \stp(t, t)$ and $(\ell_2, \ell_2) \in \stp(t', t')$. Thus, we deduce that $\ell_1 = \ell_2$.

\begin{claim}\label{claim:permutations}
Let $\pi : [\ell] \to [\ell]$ be a permutation\footnote{When we refer to a permutation we assume that the permutation is obtained by rearranging the tuples that $t$ (resp. $t'$) consists of.} such that $\pi(t) \in \mathcal{T}_k(\cC), \pi(t') \in \mathcal{T}_k(\cC')$. We have $t \equiv_1 t'$ if and only if $\pi(t) \equiv_1 \pi(t')$.
\end{claim}
\begin{proof}
It suffices to show only the direction $\implies$ as then we could also deduce that $\pi(t) \equiv_1 \pi(t') \implies \hat\pi(\pi(t)) \equiv_1 \hat\pi(\pi(t'))$, for $\hat\pi \coloneq \pi^{-1}$.

Let $\hmcol$ denote the mapping that maps a tuple to the equivalence class it belongs to, with respect to~$\equiv_1$. By definition, $\hmcol(t) = \hmcol(t')$. We observe that $(\stp(t, \pi(t)), \hmcol(\pi(t))$ is unique in $\{\{(\stp(t,u), \hmcol(u)) \mid u \in \mathcal{T}_k(\cC), \stp(t, u) \neq \emptyset\}\}$. Before we show this, first observe that $\stp(t, \pi(t))) = \stp(t', \pi(t'))$. Hence, by definition of robust equivalence relations and our earlier observation we have $(\stp(t, \pi(t)), \hmcol(\pi(t))) = (\stp(t', \pi(t')), \hmcol(\pi(t')))$, from which we deduce that $\pi(t) \equiv_1 \pi(t')$. To complete the proof, we show our latter claim.

First, as we have already argued earlier, $\hmcol(\pi(t)) = \hmcol(u)$ implies that $\pi(t), u$ have the same length~$\ell$. Next, say that for each $i \in [\ell]$, $\pi$ moves the $i$-th entry of $t$ to its $j_i$-th entry, which implies that for all $i\in [\ell]$, we have that $(i, j_i) \in \stp(t, \pi(t)) = \stp(t, u)$. Since $u$ has the same length as $t$ and $\{j_1, j_2, \dots, j_\ell\} = [\ell]$, it follows that $u$ agrees with $\pi(t)$ on all of their coordinates, which completes the proof.
\end{proof}

\begin{claim}\label{claim:existenceofProfPermutation}
There exists a profile $\alpha \in \sigma^{(\leq k)}$ and a permutation $\pi : [\ell] \to [\ell]$ such that $(\alpha, \pi(t)) \in \MCT_k(\cC), (\alpha, \pi(t')) \in \MCT_k(\cC')$.
\end{claim}
\begin{proof}
Let $t$ be the concatenation of $N$ tuples $p_1, \dots, p_N$ of $\cC$ (note that here we refer to tuples that belong to a relation of $\cC$). Let $\hmcol$ denote the mapping that maps a flat tuple to the equivalence class it belongs to, with respect to $\equiv_1$. Since $\hmcol(t) = \hmcol(t')$, it follows that for each $i \in [N]$, there is $w_i$ such that $(\stp(t, p_i), \hmcol(p_i)) = (\stp(t', w_i), \hmcol(w_i))$. Note that $\atp(w_i) = \atp(p_i) \neq \emptyset$ and so for each $1 \leq i \leq N$, $w_i$ has the same length as $p_i$. Let $\mathsf{len}(\cdot)$ denote a mapping that outputs the length of an input tuple. Furthermore, we have \[\bigcup_{z \in [\sum_{j=1}^{i-1}\mathsf{len}(p_j) + 1, \sum_{j=1}^{i}\mathsf{len}(p_j)]}(z, z - \sum_{j=1}^{i-1}\mathsf{len}(p_j)) \subseteq \stp(t, p_i) = \stp(t', w_i)\,,\] which means that $t'$ can be decomposed into $N$ tuples $q_1, \dots, q_N$ such that $\atp(p_i) = \atp(q_i)$. Finally, consider a profile $\alpha = (R_1, \dots, R_N) \in \sigma^{(\leq k)}$ and a permutation $j_1, \dots, j_N$ of $[N]$ such that $p_{j_i} \in (R_{j_i})^\cC$. Clearly, $q_{j_i} \in (R_{j_i})^{\cC'}$. Let $\pi$ be a permutation such that $\pi(t) = p_{j_1} + p_{j_2} + \dots + p_{j_N}$, where the binary operator $+$ concatenates two input tuples. It is easy to verify that $\pi(t') = q_{j_1} + q_{j_2} + \dots + q_{j_N}$, which completes the proof.
\end{proof}

Let $\alpha, \pi$ be given by the claim above. By definition and \Cref{claim:permutations}, we have that $(\alpha, \pi(t)) \equiv_1^+ (\alpha, \pi(t'))$ and $(\alpha, \pi(t)) \not\equiv_2^+ (\alpha, \pi(t'))$, which shows that $\bij$ is injective, as our analysis above applies to any choice of $\cC, \cC' \in \{\cA, \cB\}$.

Towards showing that $\bij$ is surjective, we show that given a pro-robust equivalence relation $\approx$, we can define a robust equivalence relation $\equiv$ such that $\equiv^+ \, = \,\approx$. So given $\approx$, we define $\equiv$ as follows: for any $\cC, \cC' \in \{\cA, \cB\}$, given $t \in \mathcal{T}_k(\cC), t' \in \mathcal{T}_k(\cC')$ of the same length $\ell$, $t \equiv t'$ if and only if there is a profile $\alpha \in \sigma^{(\leq k)}$ and a permutation $\pi : [\ell] \to [\ell]$ such that $(\alpha, \pi(t)) \approx (\alpha, \pi(t'))$. It is straightforward to verify that $\equiv$ is indeed an equivalence relation. Next, we show that $\equiv$ is robust.

\begin{claim}
The equivalence relation $\equiv$ is robust.
\end{claim}
\begin{proof}
For $\cC, \cC' \in \{\cA, \cB\}$, let $t \in \mathcal{T}_k(\cC), t' \in \mathcal{T}_k(\cC')$ such that $t \equiv t'$. By definition, there exist a profile $\alpha$ and a permutation $\pi$ such that $(\alpha, \pi(t)) \approx (\alpha, \pi(t'))$ which implies that $\stp(\pi(t), \pi(t)) = \stp(\pi(t'), \pi(t'))$. Hence, we have $\stp(t, t) = \stp(t', t')$. Next, assume that there is profile $\beta = (R)$, where $R \in \sigma$ such that $R \in \atp(t)$. Following similar arguments as in \Cref{claim:permutations}, we have $(\stp(\pi(t), t), [(\beta, t)]_\approx) = (\stp(\pi(t'), t'), [(\beta, t')]_\approx)$, from which we deduce that $R \in \atp(t')$ and so $\atp(t) = \atp(t')$.

To finish the proof of the claim we show that there is a bijection $\mu : \{s \in \mathcal{T}_k(\cC) : \stp(t, s) \neq \emptyset\} \to \{s' \in \mathcal{T}_k(\cC') : \stp(t', s') \neq \emptyset\}$ such that
\[
(\stp(t, s), [s]_\equiv) = (\stp(t', \mu(s)), [\mu(s)]_\equiv)\,.
\]

Since $(\alpha, \pi(t)) \approx (\alpha, \pi(t'))$, there is a bijection $\xi : \{(\beta, s) \in \MCT_k(\cC) : \stp(\pi(t), s) \neq \emptyset\} \to \{(\beta', s') \in \MCT_k(\cC') : \stp(\pi(t'), s') \neq \emptyset\}$ such that

\[
(\stp(\pi(t), s), [(\beta, s)]_\approx) = (\stp(\pi(t'), \xi((\beta, s))_2), [\xi((\beta, s))]_\approx)\,.
\]

For every equivalence class $E$ of $\equiv$, let $(\beta_E, \pi_E)$ denote the ``representative'' of the class $E$, that is, such that $q, q' \in E$ if $(\beta_E, \pi_E(q)) \approx (\beta_E, \pi_E(q'))$. Let $s \in \mathcal{T}_k(\cC)$ such that $s \in E$. For simplicity, let $(\beta, \hat\pi)$ denote the representative of $E$. Then, we have $(\stp(\pi(t), \hat\pi(s)), [(\beta, \hat\pi(s))]_\approx) = (\stp(\pi(t'), u), [(\beta, u)]_\approx)$, for some pro-tuple $(\beta, u) \in \MCT_k(\cC')$, using the bijection $\xi$ above. Let $s'$ be a tuple such that $\hat{\pi}(s') = u$. Note that such an $s'$ exists, since $u$ and $\hat\pi(s)$ share $\beta$, meaning that we can ''break'' $\hat\pi(s)$ and $u$ into the same number of sub-tuples, each pair of which is of the same length. Hence, we may apply $\hat\pi^{-1}$ on $u$ and get $s' \coloneq \hat\pi^{-1}(u) \in \mathcal{T}_k(\cC')$. We define $\mu : s \mapsto s'$, which is by construction bijective.

Furthermore, note that $\stp(t, s) = \stp(t', s')$ if and only if $\stp(\pi(t), \hat\pi(s)) = \stp(\pi(t'), \hat\pi(s'))$. Finally, by definition of $\equiv$, we have $(\beta, \hat{\pi}(s)) \approx (\beta, \hat\pi(s'))$ implies that $s \equiv s'$.
\end{proof}

Let $(\alpha, s) \in \MCT_k(\cC), (\alpha, s') \in \MCT_k(\cC')$ such that $(\alpha, s) \approx (\alpha, s')$. By definition of $\equiv$, we have that $s \equiv s'$, which in turn implies that $(\alpha, s) \equiv^+ (\alpha, s')$.

Hence, what remains to show is that for any $\cC, \cC' \in \{\cA, \cB\}$, given $(\alpha, s) \in \MCT_k(\cC), (\alpha, s') \in \MCT_k(\cC')$ such that $(\alpha, s) \equiv^+ (\alpha, s')$, we also have $(\alpha, s) \approx (\alpha, s')$. First note that by definition of $\equiv^+$, we have $s \equiv s'$. Next, by definition of $\equiv$, there exists a profile $\beta \in \sigma^{(\leq k)}$ and a permutation $\pi : [\ell] \to [\ell]$ (where $\ell$ is the length of $s$ and $s'$) such that $(\beta, \pi(s)) \approx (\beta, \pi(s'))$. The following claim completes the proof of the lemma.

\begin{claim}
$(\alpha, s) \approx (\alpha, s')$ if and only if $(\beta, \pi(s)) \approx (\beta, \pi(s'))$.
\end{claim}
\begin{proof}
The proof follows the same arguments as the proof of \Cref{claim:permutations}.
It suffices to show only the direction $\implies$ as then we could also deduce that $(\beta, \pi(s)) \approx (\beta, \pi(s')) \implies (\alpha, \hat\pi(\pi(s))) \approx (\alpha, \hat\pi(\pi(s')))$, for $\hat\pi \coloneq \pi^{-1}$.

Let $\rcrcol$ denote the mapping that maps a pro-tuple to the equivalence class that it belongs to, with respect to $\approx$. By definition, we have $\rcrcol(\alpha, s) = \rcrcol(\alpha, s')$. From \Cref{claim:permutations}, follows that $\stp(s, \pi(s))$ is unique in $\{\{\stp(s,u) \mid (\beta, u) \in \MCT_k(\cC), \stp(s, u) \neq \emptyset\}\}$, in the sense that $\stp(s, \pi(s)) = \stp(s, u)$ implies that $u = \pi(s)$. Furthermore, recall that $\stp(s, \pi(s))) = \stp(s', \pi(s'))$. In combination with $\rcrcol(\alpha, s) = \rcrcol(\alpha, s')$, we deduce that $(\stp(s, \pi(s)), \rcrcol(\beta, \pi(s))) = (\stp(s', \pi(s')), \rcrcol(\beta, \pi(s')))$ and so $(\beta, \pi(s)) \approx (\beta, \pi(s'))$.
\end{proof}
\end{proof}

\begin{theorem}
$\bij$ is an isomorphism.
\end{theorem}
\begin{proof}
From \Cref{lem:bijectionEqRel}, we get that $\bij$ is a bijection.
It remains to show that given two robust equivalence relations $\equiv_1, \equiv_2$ such that $\equiv_1 \, \subseteq \, \equiv_2$, it also holds that  $\equiv_1^+ \, \subseteq \, \equiv_2^+$  and vice-versa.
For the direction $\implies$, it suffices to show that for every two pro-tuples $(\alpha, s), (\alpha', s')$ that satisfy $(\alpha, s) \equiv_1^+ (\alpha', s')$, it also holds that $(\alpha, s) \equiv_2^+ (\alpha', s')$. To this end, since $(\alpha, s) \equiv_1^+ (\alpha', s')$, it follows that $\alpha = \alpha'$ and $s \equiv_1 s'$. Furthermore, since $\equiv_1 \, \subseteq \, \equiv_2$, we have $s \equiv_2 s'$. Hence, by definition of $\equiv_2^+$, it follows that $(\alpha, s) \equiv_2^+ (\alpha', s')$.

For the direction $\impliedby$, it suffices to show that for every two tuples $s, s'$ that satisfy $s \equiv_1 s'$, it also holds that $s \equiv_2 s'$. To this end, since $s \equiv_1 s'$, it follows by \Cref{claim:existenceofProfPermutation} that there exist $\alpha, \pi$ such that $(\alpha, \pi(s)) \equiv_1^+ (\alpha, \pi(s'))$. Furthermore, since $\equiv_1^+ \, \subseteq \, \equiv_2^+$, we have $(\alpha, \pi(s)) \equiv_2^+ (\alpha, \pi(s'))$. Hence, by definition of $\equiv_2^+$, it follows that $\pi(s) \equiv_2 \pi(s')$, which in combination with \Cref{claim:permutations} implies that $s \equiv_2 s'$, completing the proof.
\end{proof}

\subsection{Putting all the pieces together}
We are now ready to show the intended equivalence of $\hmrcr$ and $\rcr$. To this end, consider again two $\sigma$-structures $\cA, \cB$ with disjoint domains. We say that a robust (resp. pro-bust) equivalence relation $\equiv$ (resp. $\approx$) \textit{distinguishes} $\cA, \cB$ if there is an equivalence class $E$ of $\equiv$ (resp. $\approx$) that \textit{distinguishes} $\cA, \cB$, that is, $|E \,\cap\, \mathcal{T}_k(\cA)| \neq |E \,\cap\, \mathcal{T}_k(\cB))|$ (resp. $|E \,\cap\, \MCT_k(\cA)| \neq |E \,\cap\, \MCT_k(\cB))|$). Next, let $\equiv_c$ denote a coarsest robust equivalence relation. Clearly, $\equiv^+_c \coloneq \bij(\equiv_c)$ is a coarsest pro-robust equivalence relation, due to the isomorphism $\bij$.

\begin{theorem}[Theorem~\ref{thm:intro_rcr-hmrcr-equivalence}; restated]\label{thm:equivalenceOfAlgorithms}
 The following statements are equivalent.
\begin{enumerate}
    \item $\hmrcr$ distinguishes $\cA, \cB$;
    \item $\equiv_c$ distinguishes $\cA, \cB$;
    \item $\equiv^+_c$ distinguishes $\cA, \cB$;
    \item $\rcr$ distinguishes $\cA, \cB$.
\end{enumerate}
\end{theorem}
\begin{proof}

Let $\equiv_\infty$ (resp. $\approx_\infty$) denote the robust (resp. pro-robust) equivalence relation that corresponds to the stable colouring induced by $\hmrcr$ (resp. $\rcr$) when run on the set $\mathcal{T}_k(\cA) \,\bigcup\, \mathcal{T}_k(\cB)$ (resp. $\MCT_k(\cA) \,\bigcup\, \MCT_k(\cB)$)).
It is clear from \Cref{prop:RobustAuxProp,prop:ProrobustAuxProp}, that $\equiv_c \, = \, \equiv_\infty$ and $\equiv_c^+ \, = \, \approx_\infty$, from which we derive $(1) \Leftrightarrow (2)$ and $(3) \Leftrightarrow (4)$ since every robust (resp. pro-robust) equivalence relation refines $\equiv_\infty$ (resp. $\approx_\infty)$. Hence, it remains to show $(2) \Leftrightarrow (3)$.

Let $E$ be an equivalence class of $\equiv_c$ that distinguishes $\cA, \cB$. From the proof of \Cref{claim:existenceofProfPermutation}, it follows that there exists profile $\alpha \in \sigma^{(\leq k)}$ and a permutation $\pi : [\ell] \to [\ell]$, where $\ell$ is the length of $t$ (and of all tuples in $E$) such that $E^+_\alpha \coloneq \{(\alpha, \pi(t)) \mid t \in E\}$ is an equivalence class of $\equiv_c^+$. Furthermore, for $\cC \in \{\cA, \cB\}$, we have $t \in \mathcal{T}_k(\cC) \implies (\alpha, \pi(t)) \in \MCT_k(\cC)$. Since, by definition $\pi$ rearranges the tuples that $t$ consists of, it follows that the mapping $t \in E \mapsto (\alpha, \pi(t))$ is injective (see, e.g., the proof of \Cref{claim:existenceofProfPermutation} for more details). Hence, we deduce that $E^+_\alpha$ also distinguishes $\cA, \cB$, which shows the direction $(2) \implies (3)$. We show $(3) \implies (2)$ by contraposition. To this end, let $E$ be an arbitrary equivalence class of $\equiv_c$. We know that $E$ \textit{does not} distinguish $\cA, \cB$. Hence, for every $\alpha$ such that $(\alpha, t) \in \MCT_k(\cA)$ where $t$ is an arbitrary tuple in $E$ (if this is true for a tuple $t$, then it is true for any other tuple in $E$ according to the proof of \Cref{claim:existenceofProfPermutation}) we have that $E_\alpha^+$ does not distinguish $\cA, \cB$. Since this is true for every equivalence class $E$, it follows that $\equiv_c^+$ does not distinguish $\cA, \cB$, completing the proof.
\end{proof}

Finally, our main result (Theorem~\ref{maintheorem:boundedGHDforStructures_neq}) follows as an immediate consequence:
\begin{proof}[Proof of \Cref{maintheorem:boundedGHDforStructures_neq}]
Combine Theorem~\ref{maintheorem:boundedGHDforStructures} with Theorem~\ref{thm:equivalenceOfAlgorithms}.
\end{proof}

\section*{Acknowledgements}
The first and third author would like to thank Benjamin Scheidt and Nicole Schweikardt for fruitful discussions on RCR variants.

\section*{Full Declaration on the use of Generative AI}
We follow the format recommended by the AI+TCS Working Group at the Simons Institute for the Theory of Computing~\cite{SimonsInstituteAITCS2026}.
\begin{flushleft}
\renewcommand{\arraystretch}{1.3}
\begin{tabular}{@{}l@{\qquad}c@{}}
\textbf{Activity} & \textbf{AI used} \\
Asking the research question & [No] \\
Coming up with the approach & [No] \\
Proof development & [No] \\
Writing and exposition & [No] \\
Checking for bugs & [Yes] \\
Other supporting tasks & [No] \\
\end{tabular}
\end{flushleft}

\begin{sloppypar}
\paragraph*{Detailed statement}
After the first draft of this manuscript was completed, Claude (Version Opus 5.5) and ChatGPT (GPT-6 Astra) have been used for proof and type checking. Genuine bugs have been discovered and corrected; the fixes have been developed and implemented by the authors without further AI assistance.
\end{sloppypar}

\bibliographystyle{plain}
\bibliography{references}

\newpage

\appendix
\section{The $\rcrFrac$ algorithm and the proof of \Cref{maintheorem:boundedFHDforStructures}}\label{sec:Fractional}

In this section, we define $\rcrFrac$ which is our second relational WL algorithm and prove our second main result (\Cref{maintheorem:boundedFHDforStructures}) both of which can be seen as the ``fractional'' analogues of $\rcr$ and our first main homomorphism-indistinguishability characterisation respectively.

\subsection{Additional technical background and notation}

First, we note that we will follow the same ordering convention of relational symbols that also extends to coloured tuples as it was described in \Cref{sec:orderingSetUp}.

The following lemmas will be useful for adapting the definition of the $\rcr$ algorithm (see, \Cref{def:kRCR}) accordingly and derive its ``fractional analogue'' which is the $\rcrFrac$ algorithm as well as for the adaptation of the main technical argumentation used in the proof of \Cref{maintheorem:boundedGHDforStructures}.

\begin{lemma}\label{lem:UpperBoundFractional}
For each $k \in \mathbb{N}$ and signature $\sigma$ there is a number $\upp(k; \sigma) \in \mathbb{N}$ such that for any $\sigma$-structure $\cA$, the following holds: if $\cA$ has fractional edge-cover number at most $k$, then $\cA$ contains at most $\upp(k; \sigma)$ coloured tuples.
\end{lemma}

\begin{proof}

Recall that we assume that $\cA$ features no isolated domain elements (see, \Cref{remark:no_isolated_elements}). Let $\rho : \mathsf{CT}(\cA) \to \mathbb{Q}^+$ be a fractional edge cover of $\cA$ of size at most $k$. Let $\bar{t}$ be a tuple made up of $\ell$ domain elements of $\cA$ (recall that we write $A$ for the domain of $\cA$). Let also $R_1, \dots, R_q \in \sigma$ be all the relation symbols such that $\bar{t} \in (R_i)^\cA$ (ordered arbitrarily). Consider the following mapping $\hat\rho : \mathsf{CT}(\cA) \to \mathbb{Q}^*$ obtained by $\rho$, by setting for each tuple $\bar{t}$, 

\[
\hat\rho(R_i(\bar{t})) = \left\{
\begin{array}{cc}
    \sum_{j = 1}^q\rho(R_j(\bar{t})) & \mbox{ if $i = 1$} \\
    0 & \mbox{ if $i > 1$} 
\end{array}
\right.
\]  

It follows from the definition of fractional edge covers (see, \Cref{def:fracEdgeCoverNumberPrelims}), that $\hat{\rho}$ is also a fractional edge cover of $\cA$ of size at most $k$. 

Now, as observed in \cite{Gottlob2023,10.1145/3801900}, we may upper-bound the number of domain elements of $\cA$ as follows:

\[|A| \leq \sum_{x \in A}\underbrace{\sum_{R(\bar{t}) : x \in \bar{t}}\hat{\rho}(R(\bar{t}))}_{\geq 1} \leq \sum_{R(\bar{t})}\sum_{x \in \bar{t}}\hat{\rho}(R(\bar{t})) \leq \ar(\sigma)\cdot\sum_{R(\bar{t})}\hat{\rho}(R(\bar{t})) \leq \ar(\sigma) \cdot k\,.\]

Finally, since $|A| \leq k\cdot \ar(\sigma)$, it follows that the maximum number of coloured tuples that $\cA$ may feature, depends only on $k$, the arity of $\sigma$ and the number of relation symbols in $\sigma$.
\end{proof}

For a $\sigma$-structure $\cA$ and a tuple $\bar{a}$ consisting of coloured tuples of $\cA$, we write $\cA[\bar{a}]$ for the substructure of $\cA$ induced by the coloured tuples in $\bar{a}$ (and with no isolated domain elements). By \Cref{lem:UpperBoundFractional}, it follows that if $\cA[\bar{a}]$ has fractional edge-cover number at most $k$, then $\bar a$ contains at most $\upp(k;\sigma)$ \textit{distinct} coloured tuples. Since we consider $\sigma$ to be fixed, we will write $\upp(k)$ instead of $\upp(k;\sigma)$ for convenience. 

\begin{lemma}[{\cite[Lemma 7]{Chen2020}}]\label{lem:auxFracCoverNumPreserved}
Let $\cA, \cB$ be $\sigma$-structures and let $h \in \Hom(\cA, \cB)$. For a tuple $\bar{a}$ consisting of coloured tuples of $\cA$, we write $h(\bar{a})$ for the tuple obtained by applying $h$ entry-wise on each coloured tuple of $\bar{a}$. For any $\bar{a}$, if $\cA[\bar{a}]$ has fractional edge-cover number at most $k$, then $\cB[h(\bar{a})]$ has fractional edge-cover number at most $k$.
\end{lemma}

\subsection{The $\rcrFrac$ algorithm}
Towards the definition of $\rcrFrac$, fix a signature $\sigma$. For a $\sigma$-structure $\cA$, let
$\mathsf{MFCT}_k(\cA)$ contain the objects
$\Omega=(\alpha;\bar t_1,\ldots,\bar t_m)\in
\MCT_{\upp(k)}(\cA)$ for which the substructure induced by the
coloured tuples $R_1(\bar t_1),\ldots,R_m(\bar t_m)$ has fractional
edge-cover number at most $k$.  Set $\flatt(\Omega):=\bar t_1+\cdots+\bar
t_m$ and $L(\Omega):=|\flatt(\Omega)|$.

\begin{definition}[Atomic \& Similarity Types for $\rcrFrac$]
For $\Omega,\Psi\in\mathsf{MFCT}_k(\cA)$, we define
$\stp(\Omega,\Psi):=\{(p,q)\in [L(\Omega)]\times[L(\Psi)]\mid \flatt(\Omega)[p]=\flatt(\Psi)[q]\}$,
and $\stp(\Omega):=\stp(\Omega,\Omega)$.
\end{definition}

\begin{definition}[$k$-Relational colour refinement based on fractional edge covers (\rcrFrac)]
\label{def:FrackRCR}
We fix $k\ge 1$ and define the $\rcrFrac$ algorithm that iteratively colours the elements of $\mathsf{MFCT}_k(\cA)$ by the colouring $\chi_i^\cA$ computed as follows (where $i$ is the iteration counter)

\[\chi_0^\cA(\Omega):=(\profile(\Omega),\stp(\Omega));
\]
\[\chi_{t+1}^\cA(\Omega):=
\Bigl(
\chi_t^\cA(\Omega),\
\{\!\{\,(\stp(\Omega,\Psi),\chi_t^\cA(\Psi))\mid
\Psi\in\mathsf{MFCT}_k(\cA),\ \stp(\Omega,\Psi)\neq\emptyset\,\}\!\}
\Bigr).
\]
We say that $\rcrFrac$ \emph{stabilises} after iteration $j$ if the colour classes formed before iteration $j$ do not change after iteration $j$.
Let $\chi_\infty^\cA$ be the stable coloring, and for a stable color $c$ let
$\mathsf{mult}_k^\cA(c):=|\{\Omega\in\mathsf{MFCT}_k(\cA)\mid \chi_\infty^\cA(\Omega)=c\}|$.

For $\sigma$-structures $\cA,\cB$, write $\cA\equiv_{\rcrFrac}\cB$ iff
$\mathsf{mult}_k^\cA(c)=\mathsf{mult}_k^\cB(c)$ for all stable colors $c$, in which case we say that $\rcrFrac$ \emph{cannot distinguish} structures $\mathcal{A}$ and $\mathcal{B}$.
\end{definition}

\subsection{Towards the proof of \Cref{maintheorem:boundedFHDforStructures}}

Similarly to inducing binary structures by a GHD as in \Cref{def:ghd.bin.struct}, we may also induce binary structures by a FHD as follows. Essentially, the only difference is that the induced structure $\cA^D$ is now a $\widehat{\sigma}_{\upp(k)}$-structure (and not a $\widehat{\sigma}_k$-structure) so as to ensure that all possible tuples $\bar{a}$ of size $\upp(k)$ such that $\cA[\bar{a}]$ (more precisely, the substructure induced by the distinct tuples appearing in $\bar a$) has fractional edge cover number at most $k$ are considered, according to \Cref{lem:UpperBoundFractional}.

\begin{definition}[Binary structure induced by a FHD]
\label{def:ghd.bin.struct.fractional}

Let $\cA$ be a $\sigma$-structure and let $D=(T,B,\lambda)$ be a FHD of $\cA$ of width at most $k \in \mathbb{N}$.
Define $\cA^D$ as the $\widehat{\sigma}_{\upp(k)}$-structure with universe
 $V(T)$ and the minimal interpretation s.t. (we omit the superscript $\cA^D$):
\begin{enumerate}
\item $u \in U_\alpha$ iff $\alpha = \profile(u)$.
\item $E_{i,j}(u,w)$ if $\{u,w\}\in E(T)\,\cup\,\{\{x,x\} \mid x \in V(T)\}$, and $\bar a[i] = \bar b[j]$ and $\bar a [i] \neq \varepsilon$, where $\bar a = \bagtup(u)$ and $\bar b = \bagtup(w)$.
\end{enumerate}
\end{definition}

Similarly to \Cref{prop:kRCR.1WL}, we also show that $\rcrFrac$-indistinguishability can be interpreted in terms of 1-WL-indistinguishability for the respective ``fractional'' $k$-exploded encodings (which can be seen as the analogues of the canonical $k$-exploded encodings), stated formally below.

To this end, fix again a signature $\sigma$ and a positive integer $k\ge 1$ and first observe that with \Cref{lem:UpperBoundFractional} in hand, we may now define the aforementioned ``fractional analogue'' of a canonical $k$-exploded encoding (from \Cref{def:ExplodedEncoding}, in the case of the $\rcr$ algorithm) as follows.

\begin{definition}[Fractional $k$-Exploded Encoding]\label{def:FracEncoding}
For $\alpha=(R_1,\ldots,R_m)\in\sigma^{(\leq\upp(k))}$ and
$\bar a=\bar t_1+\cdots+\bar t_m\in R_1^\cA\times\cdots\times R_m^\cA$,
call $(\alpha,\bar a)$ $k$-fractional if the substructure induced by the
coloured tuples $R_1(\bar t_1),\ldots,R_m(\bar t_m)$ has fractional
edge-cover number at most $k$. Define $\mathbf F(\cA,k)$ to have all
$\bar a$ for which $(\alpha,\bar a)$ is $k$-fractional for some $\alpha$,
put $\bar a\in U_\alpha$ iff $(\alpha,\bar a)$ is $k$-fractional, and retain
the relations $E_{i,j}(\bar a,\bar b)$ iff $\bar a[i]=\bar b[j]$.  Define
$\fenc(\cA,k)$ as the substructure of $\enc(\cA,\upp(k))$ induced by its
$k$-fractional elements.
\end{definition}

\begin{remark}
Note that all elements $\bar a \in \dom(\mathbf{F}(\cA, k))$ use at most $\upp(k)$ relation tuples.
\end{remark}

\begin{proposition}
\label{prop:FractionalkRCR.1WL}
Let $\cA,\cB$ be $\sigma$-structures and let $\fenc(\cA,k),\fenc(\cB,k)$
be the profile-split fractional $k$-exploded encodings.
Then
\(
\cA \equiv_{\rcrFrac} \cB
\quad\Longleftrightarrow\quad
\fenc(\cA,k) \equiv_{\textup{1-WL}} \fenc(\cB,k).
\)
\end{proposition}

\begin{proposition}[Fractional split/non-split bridge]\label{prop:F-split-bridge}
For all $\sigma$-structures $\cA,\cB$,
\[
\fenc(\cA,k)\equiv_{1\textup{-WL}}\fenc(\cB,k)
\quad\Longleftrightarrow\quad
\mathbf F(\cA,k)\equiv_{1\textup{-WL}}\mathbf F(\cB,k).
\]
\end{proposition}

\begin{proof}
Apply the proof of \Cref{prop:B-enc-1WL} to the pairs
$(\alpha,\bar a)$ that are $k$-fractional. By definition, these are exactly
the elements of $\fenc(\cA,k)$, while their projections $\bar a$, with the
corresponding unary predicates, form $\mathbf F(\cA,k)$.
\end{proof}

\subsection{Proof of \Cref{maintheorem:boundedFHDforStructures}}

In this section, we prove \Cref{maintheorem:boundedFHDforStructures} with which we establish a characterisation of our $\rcrFrac$ algorithm in terms of homomorphism indistinguishability over relational structures of \textit{semi-pure} fractional hypertreewidth at most $k$. Our extended technical set-up will now allow to write the proof following the same lines as the proof of our first main theorem \Cref{maintheorem:boundedGHDforStructures}. However, there are several further technical observations that are substantial and non-trivial for the adaptation and hence we provide a self-contained proof.

\begin{remark}
Unless stated otherwise, for what follows a FHD will always be assumed to be a semi-pure FHD.  
\end{remark}

Similarly to \Cref{def:Extensions}, we may define $D$-extensions of homomorphisms with respect to a FHD $D$, as follows.

\begin{definition}\label{def:FractionalExtensions}
Let $\cA, \cB$ be $\sigma$-structures and let $D = (T, B, \lambda)$ be a full FHD of $\cA$ with width at most $k \in \mathbb{N}$. For $u \in V(T)$, write $\bagtup(u) = \bar{a}^*_1 + \dots + \bar{a}^*_\ell$, \footnote{Recall that $\bagtup(u)$ may contain $\varepsilon$-entries.} and $\mathsf{Ord}\lambda(u) = R_1(\bar{a}_1), \dots, R_\ell(\bar{a}_\ell)$, where $1 \leq \ell \leq \upp(k)$. Given a mapping $h : \dom(\cA) \to \dom(\cB)$ and $u \in V(T)$, let $\exthom{h}{u}$ be a mapping that maps the tuple $\bar{a}^*_1 + \dots + \bar{a}^*_\ell$ entry-wise to some $\bar{t} \in \dom(\mathbf{F}(\cB, k))$ of the same size such that $(1)$ $\exthom{h}{u}$ agrees with $h$ on all entries $x \neq \varepsilon$ of $\bagtup(u)$ and $(2)$ $\exthom{h}{u}$ maps all $\varepsilon$'s to arbitrary elements in $\dom(\cB)$ that need not be equal. 

For $h \in \Hom(\cA, \cB)$, we say that the multiset $\{\{\exthom{h}{u} : u \in V(T)\}\}$ $D$-\emph{extends} $h$ w.r.t $\cB$ if and only if for each $u \in V(T)$, it holds  $\exthom{h}{u}(\bagtup(u)) \in U^{\mathbf{F}(\cB, k)}_{\profile(u)}$, where $\profile(u) = (R_1, \dots, R_\ell)$. 

We write $\mathsf{Ext}^D_\cB(h)$ for all $D$-extensions of $h$ (w.r.t $\cB$).
\end{definition}

\begin{remark}[On \Cref{def:FractionalExtensions}]
Note that the mapping $\exthom{h}{u}$ naturally induces a mapping $\exthom{h}{u,i}$ acting on each individual tuple $\bar{a}_i^*, 1 \leq i \leq \ell$ such that, $\exthom{h}{u,i}(\bar{a}_i^*) \in R_i^\cB$ and $\exthom{h}{u}(\bar{a}_1^* + \dots + \bar{a}_\ell^*)$ may equivalently be written as $\exthom{h}{u,1}(\bar{a}_1^*) + \dots + \exthom{h}{u,\ell}(\bar{a}_\ell^*)$.
Compared to \Cref{def:Extensions}, where every possible combination of (individual) mappings $\exthom{h}{u,i}, 1 \leq i \leq \ell$ yielded a valid mapping $\exthom{h}{u}$ for each $u \in V(T)$ and thus a valid $D$-extension of $h$, the situation now in \Cref{def:FractionalExtensions} is different. In particular, it may be the case that the substructure of $\cB$ induced by the coloured tuples $\exthom{h}{u,1}(\bar{a}_1^*), \dots, \exthom{h}{u,\ell}(\bar{a}_\ell^*)$ has fractional edge cover number which is larger than $k$ and thus $\exthom{h}{u}$ is not well-defined. However, we can still show that $\mathsf{Ext}^D_\cB(h) \neq \emptyset$, for every $h \in \Hom(\cA, \cB)$, which is crucial for stating a 'fractional analogue' of \Cref{lem:HomSumOfExt}.
\end{remark}

\begin{lemma}\label{lem:non-empty-extension-set}
Let $\cA$ be a $\sigma$-structure and let $D = (T, B, \lambda)$ be a \emph{full} and \emph{semi-pure} FHD of $\cA$ with width at most $k \in \mathbb{N}$. For every $\sigma$-structure $\cB$ and $h \in \Hom(\cA, \cB)$, it holds $\mathsf{Ext}^D_\cB(h) \neq \emptyset$.    
\end{lemma}
\begin{proof}
For $u \in V(T)$, write $\fulltup(u) = \bar{a}_1 + \dots + \bar{a}_\ell$, $\bagtup(u) = \bar{a}^*_1 + \dots + \bar{a}^*_\ell$ and $\mathsf{Ord}\lambda(u) = R_1(\bar{a}_1), \dots, R_\ell(\bar{a}_\ell)$, where $1 \leq \ell \leq \upp(k)$. Let $\exthom{h}{u}$ denote the mapping that maps $\bagtup(u)$ entry-wise to $h(\fulltup(u))$, that is, in particular if $\bar{a}_i^*[j] = \varepsilon$, then $\exthom{h}{u}$ maps $\bar{a}^*_i[j]$ to $h(\bar{a}_i[j])$.

From the combination of \Cref{lem:auxFracCoverNumPreserved,lem:UpperBoundFractional}, we deduce that $\exthom{h}{u}(\bagtup(u)) \in \mathbf{F}(\cB, k)$. Furthermore, since $h$ is a homomorphism we have $h(\bar{a}_i) \in U_{R_i}^\cB$, for each $1 \leq i \leq \ell$ which in turn implies that $h(\fulltup(u)) \in U_{\profile(u)}^{\mathbf{F}(\cB, k)}$. Finally, since $\exthom{h}{u}(\bagtup(u)) = h(\fulltup(u))$, we also have $\exthom{h}{u}(\bagtup(u)) \in U_{\profile(u)}^{\mathbf{F}(\cB, k)}$ and so $\mathsf{Ext}^D_\cB(h) \neq \emptyset$.
\end{proof}

With \Cref{lem:non-empty-extension-set} in hand, we may derive an analogue of \Cref{lem:HomSumOfExt} now based on FHD, the proof of which follows the exact same lines as the proof of \Cref{lem:HomSumOfExt}.

\begin{lemma}\label{lem:HomSumOfExtFrac}
Let $\cA$ be a $\sigma$-structure and let $D = (T, B, \lambda)$ be a \emph{full} and \emph{semi-pure} FHD of $\cA$ with width at most $k \in \mathbb{N}$.
Then for every $\sigma$-structure $\cB$,
\[
\homs(\cA^{D},\mathbf{F}(\cB,k)) = \sum_{h \in \Hom(\cA, \cB)}|\mathsf{Ext}^D_\cB(h)|\,. 
\]
In particular, $\{\mathsf{Ext}^D_\cB(h)\}_{h \in \Hom(\cA, \cB)}$ partition $\Hom(\cA^D, \mathbf{F}(\cB, k)$ into $\homs(\cA, \cB)$ equivalence classes.
\end{lemma}
\begin{proof}

Let $h \in \Hom(\cA, \cB)$ and assume that a given multiset $\{\{\exthom{h}{u} : u \in V(T)\}\}$ of mappings $D$-extends $h$ w.r.t. $\cB$. Note that $\{\{\exthom{h}{u} : u \in V(T)\}\}$ induces a mapping $\hat{h} : \cA^D \to \mathbf{F}(\cB, k)$, where $\hat{h}(u) = \exthom{h}{u}(\bagtup(u))$. We have
\begin{enumerate}
\item for each $u \in V(T)$, $\hat{h}(u) = \exthom{h}{u}(\bagtup(u)) \in U^{\mathbf{F}(\cB, k)}_{\profile(u)}$, following from the definition of a $D$-extension of $h$;
\item for each $(u, v)$ such that $u = v$ or $\{u, v\} \in E(T)$, if $(u, v) \in E_{i,j}^{\cA^D}$, then $(\hat{h}(u), \hat{h}(v)) \in E_{i, j}^{\mathbf{F}(\cB, k)}$. To see this, recall that $(u, v) \in E_{i,j}^{\cA^D}$ if and only if $\bagtup(u)[i] = \bagtup(v)[j] \neq \varepsilon$. Equivalently, $\bagtup(u)[i] = \bagtup(v)[j] \in B(u) \, \cap \, B(v)$. Since $\exthom{h}{u}, \exthom{h}{v}$ agree on $B(u)\,\cap\,B(v)$, we deduce that $(\hat{h}(u), \hat{h}(v)) = (\exthom{h}{u}(\bagtup(u)), \exthom{h}{v}(\bagtup(v))) \in E_{i, j}^{\mathbf{F}(\cB, k)}$. 
\end{enumerate}
Hence, we deduce that $\hat{h} \in \Hom(\cA^D, \mathbf{F}(\cB, k))$ which in turn implies that there is a natural mapping $\pi : \bigcup_{h \in \Hom(\cA, \cB)}\mathsf{Ext}^D_\cB(h) \to \Hom(\cA^D, \mathbf{F}(\cB, k))$ that maps $\{\{\exthom{h}{u} : u \in V(T)\}\}$ to $\hat{h}$ which is injective by definition. 

Also, it is clear that if $\cA$ has no isolated domain elements, then for any two $h, h' \in \Hom(\cA, \cB)$ such that $\mathsf{Ext}^D_\cB(h) \,\cap\, \mathsf{Ext}^D_\cB(h') \neq \emptyset$, we have $h = h'$. To see this, note that for each $u \in V(T)$, the mapping of an intersecting $D$-extension, that corresponds to $u$, must agree on $B(u)$ with both $h$ and $h'$. Since, $\bigcup_{u \in V(T)}B(u) = \dom(\cA)$, we deduce that $h = h'$. Hence, the sets $\{\mathsf{Ext}^D_\cB(h)\}_{h \in \Hom(\cA, \cB)}$ are pair-wise disjoint.

Next let $\hat{g} \in \Hom(\cA^D, \mathbf{F}(\cB, k))$. Recall that by assumption, each element $z \in \dom(\cA)$ is contained in some tuple of $\cA$ and thus it is also contained in some bag $B(u)$ of $T$. For each $z \in \dom(\cA)$, we fix a node $u_z \in V(T)$ such that $z \in B(u_z)$. Recall that if $\bagtup(u_z)[i] = \bagtup(u_z)[j] \neq \varepsilon$ then we have $(u_z, u_z) \in E^{\cA^D}_{i,j}$. Since $\hat{g}$ is a homomorphism, it also follows that $(\hat{g}(u_z), \hat{g}(u_z)) \in E^{\mathbf{F}(\cB, k)}_{i,j}$ implying that there is $x_z \in \dom(\cB)$ such that for any index $i$ with $\bagtup(u_z)[i] = z$, we have $\hat{g}(u_z)[i] = x_z$. We consider the well-defined mapping $\rho : \dom(\cA) \to \dom(\cB)$ that maps each $z \in \dom(\cA)$ to its corresponding element $x_z \in \dom(\cB)$ obtained from the procedure described above.

\begin{claim}\label{claim:SurjectivityNonPureFrac}
For $u \in V(T)$ with $\bar{a} = \bagtup(u)$ such that $\hat{g}(u) = \bar{d} \in \dom(\mathbf{F}(\cB, k))$, we have $\rho(\bar{a}[i]) = \bar{d}[i]$, whenever $\bar{a}[i] \neq \varepsilon$.    
\end{claim}

\begin{proof}
Let $z \in \bar{a}$ (where $z \neq \varepsilon$). Recall that we have fixed $u_z \in V(T)$ such that $z \in B(u_z)$. By the definition of $D$, it follows that for the path $(u_0, u_1, \dots, u_{q-1}, u_q)$ in $T$ connecting $u_0 = u$ and $u_q = u_z$, it holds that $z \in B(u_j)$, for each $0 \leq j \leq q$. For each $0 \leq j \leq q$, write $\bar{t}_j  = \bagtup(u_j)$ and let $i_j$ denote any index such that $\bar{t}_j[i_j] = z$. Then, it easy to see that for each $0 \leq j < q$, we have by definition that $(u_j, u_{j+1}) \in E^{\cA^D}_{i_j,i_{j+1}}$ which implies that we also have $(\hat{g}(u_j), \hat{g}(u_{j+1})) \in E^{\mathbf{F}(\cB, k)}_{i_j, i_{j+1}}$, since $\hat{g}$ is a homomorphism. Concretely we have,
\begin{enumerate}
\item $\bar{t}_q[i_q] = \bar{t}_0[i_0] = \bar{a}[i_0] = z$;
\item $\hat{g}(u_0)[i_0] = \hat{g}(u_z)[i_q]$;
\item $\hat{g}(u_z)[i_q] = x_z$ due to the fact that whenever $\bagtup(u_z)[i] = \bar{t}_q[i] = z$, we have $\hat{g}(u_z)[i] = x_z$.
\end{enumerate}
Hence, $\bar{d}[i_0] \overset{def}= \hat{g}(u_0)[i_0] \overset{(2)}= \hat{g}(u_z)[i_q] \overset{(3)}= x_z \overset{def}= \rho(\bar{a}[i_0])$ which also shows the claim.
\end{proof}

\begin{claim}\label{claim:isHomFrac}
$\rho \in \Hom(\cA, \cB)$.    
\end{claim}
\begin{proof}
Let $R \in \sigma$ and $\bar{t} \in R^\cA$. Since $D$ is full, there is $u \in V(T)$ such that $R(\bar{t}) \in \lambda(u)$ and $\mathsf{set}(\bar{t}) \subseteq B(u)$. Assume that $R(\bar{t})$ is the $i$-th coloured tuple in $\mathsf{Ord}\lambda(u) = R_1(\bar{a}_1), \dots, R_\ell(\bar{a}_\ell)$, where $1 \leq \ell \leq k$. Note that $\bar{a}_i^* = \bar{a}_i = \bar{t}$ since $\mathsf{set}(\bar{t}) \subseteq B(u)$. Let $\profile(u) = (R_1, \dots, R_i, \dots, R_\ell)$ where $R_i = R$. Since, $\hat{g} \in \Hom(\cA^D, \mathbf{F}(\cB, k))$, we have $\hat{g}(u) \in R_1^{\cB} \times \dots \times R^\cB \times \dots \times R_\ell^\cB$ which in turn by \Cref{claim:SurjectivityNonPureFrac} implies that $\rho(\bar{t}) \in R^\cB$ which completes the proof.
\end{proof}

By the combination of \Cref{claim:SurjectivityNonPureFrac,claim:isHomFrac} we deduce that $\hat{g}$ naturally induces $\{\{\exthom{\rho}{u} : u \in V(T)\}\} \in \mathsf{Ext}^D_\cB(\rho)$ such that $\pi(\{\{\exthom{\rho}{u} : u \in V(T)\}\}) = \hat{g}$, which shows that $\pi$ is also surjective. Hence, $\#\Hom(\cA^{D},\mathbf{F}(\cB,k)) = |\bigcup_{h \in \Hom(\cA, \cB)}\mathsf{Ext}^D_\cB(h)| = \sum_{h \in \Hom(\cA, \cB)}|\mathsf{Ext}^D_\cB(h)|$, where the last equality follows from the fact that the sets $\{\mathsf{Ext}^D_\cB(h)\}_{h \in \Hom(\cA, \cB)}$ are pair-wise disjoint ---as argued earlier--- and non-empty as shown in \Cref{lem:non-empty-extension-set}.
\end{proof}

\subsubsection{$\rcrFrac$-equivalence implies homomorphism indistinguishability over structures of semi-pure fractional hypertreewidth at most $k$}

In this section, we prove the direction $(2) \implies (1)$ of \Cref{maintheorem:boundedFHDforStructures}, formally stated in the following lemma. The proof follows verbatim the proof of \Cref{lem:connectedHomInd}, but we provide a proof for the sake of completeness.

\begin{lemma}\label{lem:connectedHomIndFrac}
Let $\cA$ be a connected $\sigma$-structure and let $D = (T, B, \lambda)$ be a \emph{full} and \emph{semi-pure} FHD of $\cA$ with width at most $k \in \mathbb{N}$. Let $\cB, \cB'$ be $\sigma$-structures such that $\cB \equiv_{\rcrFrac} \cB'$. Then, $\homs(\cA, \cB) = \homs(\cA, \cB')$.   
\end{lemma}

\begin{proof}
We assume wlog that $\cB, \cB'$ have disjoint domains. By By \Cref{prop:FractionalkRCR.1WL,prop:F-split-bridge}, we deduce that $\mathbf F(\cB,k)$ and
$\mathbf F(\cB',k)$ are $1$-WL-equivalent. Let $\chi$ denote the stable colouring produced by 1-WL when run on $\mathbf F(\cB, k) \,\bigcup\, \mathbf F(\cB', k)$. 

Let $\cC \in \{\mathbf F(\cB, k), \mathbf F(\cB', k)\}$.  Let $t$ be a tuple in the domain of $\cC$ and let $\alpha$ be a profile such that $t \in (U_\alpha)^{\cC}$. Let also $x$ be a sequence of positive integers $n \geq 1$ appearing in strictly increasing order such that the maximum integer appearing in $x$ is at most the length of $t$. We call $x$ a \textit{grounding} of $t$. For such $t, x, \alpha$, we define $\mathfrak g(t, x, \alpha)$ as the set consisting of all tuples $s$ in the domain of $\cC$ such that $s \in (U_\alpha)^{\cC}$ and for every $i \in x$, we have $(i, i) \in \stp(t, s)$. We call $\mathfrak g(t, x, \alpha)$ the \textit{group of $t$ w.r.t. $x$ and $\alpha$}. 

\begin{observation}
It is easy to see that, for any $s \in \mathfrak g(t, x, \alpha)$, we also have $t \in \mathfrak g(s, x, \alpha)$.    
\end{observation}

We call a set $P$ of pairs of positive integers a \textit{similarity pattern}, if there exist tuples $a, b$ such that $P = \stp(a, b)$. Furthermore, for a stable colour $d \in \mathsf{Im}(\chi)$, where $\mathsf{Im}(\chi)$ is the image of $\chi$, we write $d \implies \alpha$ to denote that for a tuple $q$, $\chi(q) = d \implies \alpha \in (U_\alpha)^\cC$. It can be verified that, 
\begin{equation}\label{eq:gtx}
|\mathfrak g(t, x, \alpha)| = \sum_{\substack{d \in \mathsf{Im}(\chi) \\ d \implies \alpha}}\left(\sum_{P \supseteq \{(i, i) : i \in x\}}\Big|\{s \in \cC \mid \chi(s) = d, \stp(t, s) = P\}\}\Big|\right)\,.
\end{equation}

\begin{claim}\label{claim:auxForGroups}
Let $\cC, \cC' \in \{\mathbf F(\cB, k), \mathbf F(\cB', k)\}$ and let $t \in \cC, t' \in \cC'$ be tuples with $\chi(t) = \chi(t')$. For every stable colour $d$ in the image of $\chi$ and every similarity pattern $P$ satisfying $\{(i, i) : i \in x\} \subseteq P$, \[\Big|\{s \in \cC \mid \chi(s) = d, \stp(t, s) = P\}\Big| = \Big|\{s' \in \cC' \mid \chi(s') = d, \stp(t', s') = P\}\Big|\,,\] 

Furthermore, for any profile $\alpha$ such that $t \in (U_\alpha)^\cC, t' \in (U_\alpha)^{\cC'}$ and any grounding $x$ of $t$ and $t'$, 
\[|\mathfrak g(t, x, \alpha)| = |\mathfrak g(t', x, \alpha)|\,.\]
\end{claim}
\begin{proof}
The first equation follows directly from combining that $\chi(t) = \chi(t')$ and that $\chi$ is a stable colouring. The second equation follows by combining the first equation of the claim with \Cref{eq:gtx}. 
\end{proof}

Let $h \in \Hom(\cA, \cB)$. By \Cref{lem:GHD-non-empty-extension-set}, we know that $\mathsf{Ext}^D_\cB(h) \neq \emptyset$. Let $f^h \in \mathsf{Ext}^D_\cB(h)$. For each $u \in V(T)$, let $x_u$ be a sequence with elements $\{i : \bagtup(u)[i] \neq \varepsilon\}$ in increasing order. Set $\alpha_u \coloneq \profile(u)$. For each $u \in V(T)$, set $t_u \coloneq f^h(u)$. Note, that for each $u \in V(T)$, $x_u$ is a grounding of $t_u$. Take $S = \{s_u\}_{u \in V(T)} \in \varprod_{u \in V(T)}\mathfrak g(t_u, x_u, \alpha_u)$. We define the mapping $f^h_S : u \mapsto s_u$.
 
\begin{claim}\label{claim:allExtensions}
We have, $f^h_S \in \mathsf{Ext}^D_\cB(h)$. In particular, we have $\mathsf{Ext}^D_\cB(h) = \{f^h_S : S \in \varprod_{u \in V(T)}\mathfrak g(t_u, x_u, \alpha_u)\}$. 
\end{claim}
\begin{proof}
To see that $f^h_S \in \mathsf{Ext}^D_\cB(h)$, note that since for each $u \in V(T)$, $s_u$ keeps the entries on indices $i \in x_u$ unchanged, it suffices to show that $s_u \in (U_{\alpha_u})^{\mathbf F(\cB, k)}$, which is guaranteed by construction. Hence, $\{f^h_S : S \in \varprod_{u \in V(T)}\mathfrak g(t_u, x_u, \alpha_u)\} \subseteq \mathsf{Ext}^D_\cB(h)$.

To see that $\mathsf{Ext}^D_\cB(h) = \{f^h_S : S \in \varprod_{u \in V(T)}\mathfrak g(t_u, x_u, \alpha_u)\}$, note that no extension of $h$ can map a node $u$ to a tuple $q$ outside of $\mathfrak g(t_u, x_u, \alpha_u)$ as this would imply that at least one of the following are true (1) $q$ does not agree with the restriction $h|_{B(u)}$ on the indices specified by $x_u$ (2) $q \notin (U_{\alpha_u})^{\mathbf F(\cB, k)}$.    
\end{proof}

Note that a similar analysis as the one above can also be shown w.r.t. $\cB'$ and $\mathbf F(\cB', k)$. Next, we define a mapping $\xi : \Hom(\cA^D, \mathbf F(\cB, k)) \to \Hom(\cA^D, \mathbf F(\cB', k))$, the specifications of which will be made clear momentarily.

Let $u \in V(T)$. Let $d$ be a stable colour in the image of $\chi$ such that for a tuple $t \in \mathbf F(\cB, k)$, $\chi(t) = d \implies t \in (U_{\alpha_u})^{\mathbf F(\cB, k)}$, where recall that $\alpha_u = \profile(u)$. Let $Q(d)$ denote the family of all tuples in $\mathbf F(\cB, k)$ with stable colour $d$. We consider the following ordering $\leq^d_u$ w.r.t. $d$ and $u$ of tuples in $Q(d)$. We wish all tuples within the same group $\mathfrak g(t, x_u, \alpha_u)$ to appear ''next to each other'' in the ordering. So, we internally order all groups $\mathfrak g(t, x_u, \alpha_u)$ according to $\prec_\cB$ and then we order the (internally ordered) groups with each other based on their common string induced by the indices appearing in the grounding $x_u$ (again using $\prec_\cB$). Similarly, we define $Q'(d)$ w.r.t. $\mathbf F(\cB', k)$. For ease of notation, we use the same symbol $\leq^d_u$ for the ordering associated with $Q'(d)$. Clearly, $|Q(d)| = |Q'(d)|$. Furthermore, recall that according to \Cref{claim:auxForGroups}, for any two tuples $t, t'$ of the same colour, their groups contain the same number of tuples of any pre-specified stable colour. So, $Q(d)$ is ``grouped'' into groups of the same size, and similarly $Q'(d)$ is ``grouped'' into equally many groups as $Q(d)$.

Towards defining $\xi$, given $h \in \Hom(\cA, \cB)$ and $f \in \mathsf{Ext}^D_\cB(h)$, we traverse $T$ in a \emph{breadth-first search} (BFS) fashion, assuming that is rooted at node $r$, and perform the following:

\begin{enumerate}
\item Set $t_r \coloneq f(r)$. Let $i_r$ denote the index of $t_r$ within the colour-class $Q(\chi(t_r))$ with respect to the ordering $\leq^r_{\chi(t_r)}$. We let $s_r$ denote the $i_r$-th element of $\mathbf{F}(\cB', k)$ within the corresponding colour-class $Q'(\chi(t_r))$, with respect to the corresponding ordering $\leq^r_{\chi(t_r)}$; such $i_r$ exists since $\chi$ is a stable colouring and so the aforementioned colour classes have the same size.
\item Let $u$ be a child of $r$ and set $t_u \coloneq f(u)$. Since $\cA$ is connected, we have $B(r) \,\cap\, B(u) \neq \emptyset$,\footnote{Note that $B(r) \,\cap\,B(u)$ is a separator of $\cA$. Hence, if $B(r) \,\cap\, B(u) = \emptyset$, then $\cA$ has an empty separator and hence it is not connected, which is a contradiction.} which in turn implies that $\mathsf{stp}(t_r, t_u) \neq \emptyset$.

Let $P(r, u) \subseteq \stp(t_r, t_u)$ denote the non-empty similarity pattern that consists of all pairs of indices $(i, j)$ such that $\bagtup(r)[i] = \bagtup(u)[j] \neq \varepsilon$. Consider the ordered families $Q(d), Q'(d)$ of all tuples of colour $d \coloneq \chi(t_u)$. Let $q \in Q(d)$ be such that $\stp(t_r, q) \not\supseteq P(r, u)$. We filter-out $q$. We follow exactly the same process filtering-out elements from $Q'(d)$. More specifically, we filter out tuples from $Q(d), Q'(d)$ in parallel as follows: first note that for any $q \in Q(d)$, we have $|\{p \in Q(d) : \stp(t_r, p) = \stp(t_r, q)\}| = |\{p \in Q'(d) : \stp(s_r, p) = \stp(t_r, q)\}|$, since $\chi(t_r) = \chi(s_r)$ and $\chi$ is a stable colouring. Hence, for a given $q \in Q(d)$ that should be filtered-out, we may filter-out simultaneously the same number of tuples from each family. Also, the remaining tuples $p \in Q(d)$ (resp. $p' \in Q'(d))$ satisfy $\stp(t_r, p) \neq \stp(t_r, q)$ (resp. $\stp(s_r, p') \neq \stp(t_r, q)$). Hence, we may repeat this process, taking another tuple $q_2 \in Q(d)$ that should be filtered-out, and then filter-out all tuples $p$ (resp. $p'$) that satisfy $\stp(t_r, p) = \stp(t_r, q_2)$ (resp. $\stp(s_r, p') = \stp(t_r, q_2)$), and so on. The resulting families of tuples have the same size. Furthermore, note that, if $q$ is filtered out then the entire group it belongs to is filtered-out as well, which means that each group of $Q(d)$ (resp. $Q'(d)$) is either completely filtered-out or remains intact. Hence, the resulting families have the same number of groups. Let $i_u$ denote the index at which $t_u$ appears in the filtered family $Q(d)$. Let $s_u$ denote the $i_u$-th tuple appearing in the filtered family $Q'(d)$. We compute $s_v$ for each child $v$ of $r$ as described above.

\item We repeat step (2) for every node that we visit in the BFS traversal.

\item After we have visited every node $u \in V(T)$ and computed $s_u$, we output the mapping $\mathsf{out}_f : u \mapsto s_u $.
\end{enumerate}

We write $\xi$ for the mapping that maps $f$ to $\mathsf{out}_f$.

\begin{claim}\label{claim:OutHom}
$\mathsf{out}_f \in \Hom(\cA^D, \mathbf{F}(\cB', k))$.    
\end{claim}
\begin{proof}
Let $u \in V(T)$ and let $\alpha$ such that $u \in U_\alpha^{\cA^D}$. We have $t_u = f(u) \in U_\alpha^{\mathbf{F}(\cB, k)}$. Since $t_u, \mathsf{out}_f(u)$ have been assigned the same stable colour (by construction), it follows from \Cref{def:kRCR} that $\mathsf{out}_f(u) \in U_\alpha^{\mathbf{F}(\cB', k)}$. The same reasoning yields that for any $i, j$ such that $(u, u) \in E_{i,j}^{\cA^D}$, we have $(\mathsf{out}_f(u), \mathsf{out}_f(u)) \in E_{i, j}^{\mathbf{F}(\cB', k)}$.

Next, consider $\{u, v\} \in E(T)$ and assume that $v$ is the parent of $u$ in $T$ rooted at $r$. By construction, we have  $P(v, u) = \{(i, j) : (v, u) \in E_{i,j}^{\cA^D}\}$ and $P(v, u) \subseteq \stp(\mathsf{out}_f(v), \mathsf{out}_f(u))$. We also have $\{(i, j) : (u, v) \in E_{i,j}^{\cA^D}\} \subseteq \mathsf{stp}(\mathsf{out}_f(u), \mathsf{out}_f(v))$, yielding that $\mathsf{out}_f \in \Hom(\cA^D, \mathbf{F}(\cB', k))$.
\end{proof}

\begin{claim}
Let $h \in \Hom(\cA, \cB)$ and let $f, \tilde f \in \mathsf{Ext}^D_\cB(h)$. For $u \in V(T)$, set $t_u \coloneq f(u), \tilde t_u \coloneq \tilde f(u)$. Assume that for each $u$, $f(u)$ and $\tilde f(u)$ have the same stable colour and $t_u \in \mathfrak g(\tilde t_u, x_u, \alpha_u)$. Set also $s_u \coloneq \xi(f)(u), \tilde s_u \coloneq \xi(\tilde f)(u)$. Then, for each $u \in V(T)$, it holds that $s_u \in \mathfrak g(\tilde s_u, x_u, \alpha_u)$ and vice versa. 
\end{claim}
\begin{proof}
Let $u \in V(T)$ and let $d \coloneq \chi(t) = \chi(t')$. To show the claim, it suffices to show that when we filter out tuples from $Q(d)$ (resp. $Q'(d)$) at node $u$, then our algorithm filters out exactly the same tuples, when we execute the corresponding step for $f$ and $\tilde f$ respectively. It is clear, that is true for the root $r$ of $T$. Let $u$ be an internal node and let $v$ denote its parent. Since $\tilde t_v \in \mathfrak{g} (t_v, x_v, \alpha_v)$, $\tilde t_u \in \mathfrak{g} (t_u, x_u, \alpha_u)$, it follows that the similarity type $P(v, u)$ is exactly the same in both executions of the corresponding step for $f, \tilde f$. To be more precise, for any $q \in Q(d)$, $\stp(t_v, q) \supseteq P(v, u)$ if and only if $\stp(\tilde t_v, q) \supseteq P(v, u)$. Finally, since in addition $t_v, \tilde t_v$ have the same stable colour and $t_u, \tilde t_u$ have the same stable colour as well, we deduce that all steps of the algorithm are executed for both mappings $f, \tilde f$ as desired.   
\end{proof}

Let $\mathfrak{Col} \in \varprod_{u \in V(T)}\mathsf{Im}(\chi)$ be a sequence of stable colours indexed by $u \in V(T)$. For such a sequence $\mathfrak{Col}$ and $\cC \in \{\mathbf F(\cB, k), \mathbf F(\cB', k)\}$, we say that a homomorphism $f \in \Hom(\cA^D, \cC)$  is $\mathfrak{Col}$-coloured, if for each $u \in V(T)$, $\chi(f(u)) = \mathfrak{Col}_u$. Let $H(\mathfrak{Col})$ (resp. $H'(\mathfrak{Col})$) denote the set of all homomorphisms in $\Hom(\cA, \cB)$ (resp. $\Hom(\cA, \cB')$) such that each $h \in H(\mathfrak{Col})$ (resp. $h' \in H'(\mathfrak{Col})$) satisfies that there exists a $\mathfrak{Col}$-coloured $D$-extension of $h$ (resp. $h'$).
Towards finishing the proof, we observe that a mapping $\xi$ induces a bijection $\rho : H(\mathfrak{Col}) \to H'(\mathfrak{Col})$, such that for any $h \in H(\mathfrak{Col})$ and any two $\mathfrak{Col}$-coloured extensions $f_1, f_2$ of $h$, we have $\xi(f_1), \xi(f_2)$ are $\mathfrak{Col}$-coloured extensions of $\rho(h)$. We may continue this process, and extend the mapping $\rho$ to the remaining homomorphisms $h \in \Hom(\cA, \cB) \setminus H(\mathfrak{Col})$, by considering a new colour-sequence $\mathfrak{Col^*}$ and establish in a similar fashion a bijection between $H(\mathfrak{Col^*}) \setminus H(\mathfrak{Col})$ and $H'(\mathfrak{Col^*}) \setminus H'(\mathfrak{Col})$ and so on, using the following claim, which follows by the construction of $\xi$ in combination with \Cref{claim:auxForGroups,claim:allExtensions} and states that no extension of a homomorphism $\phi \in H(\mathfrak{Col^*}) \setminus H(\mathfrak{Col})$ is mapped by $\xi$ to an extension of a homomorphism in $H(\mathfrak{Col})$ (that is, of a homomorphism that we have already ``hit'' while defining $\rho$), completing the proof.

\begin{claim}
Let $h, g \in \Hom(\cA, \cB)$ such that $g = \rho(h)$. For every sequence of stable colours $\mathfrak{Col}$ (indexed by $u \in V(T)$, we have \[\Big|\{f \mid f \in \mathsf{Ext}^D_\cB(h), \text{$f$ is $\mathfrak{Col}$-coloured}\}\Big| = \Big|\{f \mid f \in \mathsf{Ext}^D_\cB(g), \text{$f$ is $\mathfrak{Col}$-coloured}\}\Big|\,.\]    
\end{claim}
\end{proof}

\subsubsection{Homomorphism indistinguishability over structures of fractional hypertreewidth at most $k$ implies $\rcrFrac$-equivalence}

As in the case of \Cref{lem:reverseDirectionGHD}, the proof works in a similar fashion as the proof of \cite[Lemma 4.6]{scheidt2026colorrefinementrelationalstructures}. However there are some key technical differences, and for this reason we provide a self-contained proof as well. Before we proceed with the proof, we state an auxiliary lemma on full and pure FHDs.

\begin{lemma}[Analogue of Lemma 4.5 of \cite{scheidt2026colorrefinementrelationalstructures}]\label{lem:FHD.hom.eq}
Let $\cA$ be a $\sigma$-structure and let $D = (T, B, \lambda)$ be a \emph{full} and \emph{pure} FHD of $\cA$ with width at most $k$.
Then for every $\sigma$-structure $\cB$,
\[
\homs(\cA^{D},\fenc(\cB,k))=
\homs(\cA^{D},\mathbf{F}(\cB,k)) = \homs(\cA,\cB)\,.
\]
\end{lemma}
\begin{proof}

Let $u \in \dom(\cA^D)$ with $\mathsf{Ord}\lambda(u) = R_1(\bar{a}_1),\dots,R_\ell(\bar{a}_\ell)$, where $1 \leq \ell \leq \upp(k)$\footnote{Recall the function $\upp(k)$ from \Cref{lem:UpperBoundFractional}.}. Since $D$ is pure we have that $\bagtup(u) = \bar{a}_1 + \dots+\bar{a}_\ell$, where $+$ is overloaded so as to be also used as a binary operator that concatenates two tuples. Also, recall that $\profile(u) = (R_1, \dots, R_\ell)$. Given a mapping $h : \dom(\cA) \to \dom(\cB)$, we write $h(\bar{a}_i)$ for the tuple obtained by applying $h$ entry-wise on $\bar{a}_i$. We also write $h(\bagtup(u)) = h(\bar{a}_1) + \dots + h(\bar{a}_\ell)$.

\begin{claim}
If $h$ is a homomorphism (from $\cA$ to $\cB$), then the mapping $h' : u \in V(T) \mapsto h(\bagtup(u))$ is a homomorphism in $\Hom(\cA^D, \mathbf{F}(\cB,k))$. 
\end{claim}
\begin{claimproof}
For $u \in \dom(\cA^D)$ as defined earlier with $\profile(u) = (R_1, \dots, R_\ell)$ we have that $u \in U^{\cA^D}_{(R_1, \dots, R_\ell)}$. Hence, we first need to show that $h'(u) \in U^{\mathbf{F}(\cB,k)}_{(R_1, \dots, R_\ell)}$. Since $h \in \Hom(\cA, \cB)$ it follows that $\bar{a}_i \in R^{\cA}_i \Rightarrow h(\bar{a}_i) \in R^{\cB}_i$, for each $i \in [\ell]$. Hence, $h'(u) = h(\bar{a}_1)+ \dots +h(\bar{a}_\ell)$ which can be equivalently seen as an element of $R^{\cB}_1 \times \dots \times R^{\cB}_\ell \subseteq \mathsf{Prod}^{(\leq \upp(k))}(\cB)$. 
Since $D$ is pure, $\cA[\lambda(u)]$ has universe $B(u)$ and its coloured tuples are exactly $R_1(\bar{a}_1),\ldots,R_\ell(\bar{a}_\ell)$. As $D$ has width at most $k$, this substructure has fractional edge-cover number at most $k$. By \Cref{lem:auxFracCoverNumPreserved}, the substructure of $\cB$ induced by the coloured tuples $R_1(h(\bar{a}_1)),\ldots,R_\ell(h(\bar{a}_\ell))$ also has fractional edge-cover number at most $k$. Therefore $h'(u)\in U^{\mathbf{F}(\cB,k)}_{(R_1,\ldots,R_\ell)}$.
Next, we take $v \in \dom(\cA^D)$ and write $\bar{a} = \bagtup(u)$ and $\bar{b} = \bagtup(v)$. Consider $(u, v) \in E^{\cA^D}_{i, j}$ which holds if $\bar{a}[i] = \bar{b}[j]$. Clearly, $h(\bar{a}[i]) = h(\bar{b}[j])$ and so $(h'(u), h'(v)) \in E^{\mathbf{F}(\cB,k))}_{i, j}$. Hence, the mapping $h' : u \mapsto h(\bagtup(u))$ is a well-defined mapping that preserves relations, which completes the proof.   
\end{claimproof}

Let $\pi : \Hom(\cA, \cB) \to \Hom(\cA^D, \mathbf{F}(\cB, k))$ be the mapping that maps $h$ to $h'$ with $h'$ as defined above w.r.t $h$. We show that $\pi$ is bijective, which would conclude the proof.

\begin{claim}
$\pi$ is injective.   
\end{claim}

\begin{claimproof}
Let $h_1, h_2 \in \Hom(\cA, \cB)$ such that $h_1 \neq h_2$ which means that there is $x \in \dom(\cA)$ such that $h_1(x) \neq h_2(x)$. Recall that we have assumed that there is $R \in \sigma$ and $\bar{c} \in R^\cA$ such that $x \in \bar{c}$ (see, \Cref{remark:no_isolated_elements}). Furthermore, by the definition of $D$, there is $u \in T$ such that $\mathsf{set}({\bar{c}}) \subseteq B(u)$ and so $x \in \bagtup(u)$. It is then easy to verify that $\pi(h_1)(u) \neq \pi(h_2)(u)$ and so $\pi$ is injective.
\end{claimproof}

Next we need to show that, for each $h'' \in \Hom(\cA^D, \mathbf{F}(\cB, k))$, there is $h \in \Hom(\cA, \cB)$ such that $\pi(h) = h''$. To this end, recall that by assumption, each element $z \in \dom(\cA)$ is contained in some tuple and thus it is also contained in some bag $B(u)$. For each $z \in \dom(\cA)$, we fix a node $u_z \in V(T)$ such that $z \in \bagtup(u_z)$. Recall that if $\bagtup(u_z)[i] = \bagtup(u_z)[j]$ then $(u_z, u_z) \in E^{\cA^D}_{i,j}$. Since $h''$ is a homomorphism it also follows that $(h''(u_z), h''(u_z)) \in E^{\mathbf{F}(\cB, k)}_{i,j}$ implying that there is $x_z \in \dom(\cB)$ such that for any index $i$, if $\bagtup(u_z)[i] = z$, then $h''(u_z)[i] = x_z$. We consider the well-defined mapping $h : \dom(\cA) \to \dom(\cB)$ that maps $z$ to $x_z$.

\begin{claim}\label{claim:SurjectivityFHD}
For $u \in V(T)$ with $\bar{a} = \bagtup(u)$ such that $h''(u) = \bar{d} \in \dom(\mathbf{F}(\cB, k))$, we have $h(\bar{a}) = \bar{d}$.    
\end{claim}

\begin{claimproof}
Let $z \in \bar{a}$. Since $D$ is pure it follows that $z \in B(u)$. Recall that we have fixed $u_z \in V(T)$ such that $z \in B(u_z)$. By the definition of $D$, it follows that for the path $(u_0, u_1, \dots, u_{q-1}, u_q)$ in $T$ connecting $u_0 = u$ and $u_q = u_z$, it holds that $z \in B(u_j)$, for each $0 \leq j \leq q$. For each $0 \leq j \leq q$, write $\bar{t}_j  = \bagtup(u_j)$ and let $i_j$ denote any index such that $\bar{t}_j[i_j] = z$. Then, it easy to see that for each $0 \leq j < q$, we have by definition that $(u_j, u_{j+1}) \in E^{\cA^D}_{i_j,i_{j+1}}$ which implies that we also have $(h''(u_j), h''(u_{j+1})) \in E^{\mathbf{F}(\cB, k)}_{i_j, i_{j+1}}$. Concretely we have
\begin{enumerate}
\item $\bar{t}_q[i_q] = \bar{t}_0[i_0] = \bar{a}[i_0] = z$;
\item $h''(u_0)[i_0] = h''(u_z)[i_q]$;
\item $h''(u_z)[i_q] = x_z$ due to (1) and the definition of $u_z$.
\end{enumerate}

Hence, $\bar{d}[i_0] \overset{def}= h''(u_0)[i_0] \overset{(2)}= h''(u_z)[i_q] \overset{(3)}= x_z \overset{def}= h(\bar{a}[i_0])$ which also shows that $h(\bar{a}) = \bar{d}$.
\end{claimproof}

\begin{claim}
$h \in \Hom(\cA, \cB)$.    
\end{claim}
\begin{claimproof}
Let $R \in \sigma$ and $\bar{t} \in R^\cA$. Since $D$ is full, there is $u \in V(T)$ such that $R(\bar{t}) \in \lambda(u)$. Assume that $R(\bar{t})$ is the $i$-th coloured tuple in $\mathsf{Ord}\lambda(u) = R_1(\bar{a}_1), \dots, R_\ell(\bar{a}_\ell)$, where $1 \leq \ell \leq \upp(k)$. Let $\profile(u) = (R_1, \dots, R_i, \dots, R_\ell)$ where $R_i = R$. Since, $h'' \in \Hom(\cA^D, \mathbf{F}(\cB, k))$, we have $h''(u) \in R_1^{\cB} \times \dots \times R^\cB \times \dots \times R_\ell^\cB$ which implies that $h(\bar{t}) \in R^\cB$ since $h''(u) = h(\bar{a}_1) + \dots + h(\bar{a}_i) + \dots + h(\bar{a}_\ell)$ and $\bar{t} = \bar{a}_i$ which completes the proof.
\end{claimproof}

Finally, it is easy to verify that $\pi(h) = h''$ which follows from \Cref{claim:SurjectivityFHD}, showing that $\pi$ is also surjective which completes the proof.

\end{proof}

The map $g\mapsto(u\mapsto(\profile(u),g(u)))$ is a bijection from
$\Hom(\cA^D,\mathbf F(\cB,k))$ to $\Hom(\cA^D,\fenc(\cB,k))$, with inverse
$(\alpha,\bar b)\mapsto\bar b$. This proves the first equality.

Furthermore, the following claim will be useful. 

\begin{claim}\label{claim:stpFrac}
Let $\cA$ be a $\sigma$-structure and let $\bar{a} = (R_1(\bar{t}_1), \dots, R_\ell(\bar{t}_\ell))$ where for each $i \in [\ell]$, $R_i(\bar{t}_i)$ is a coloured tuple in $\cA$. Similarly, let also $\bar{a}' = (R_1(\bar{t}_1'), \dots, R_\ell(\bar{t}_\ell'))$. Letting $\bar{a}^+ \coloneq R_1(\bar{t}_1) + \dots + R_\ell(\bar{t}_\ell)$, it holds that if $\mathsf{stp}(\bar{a}^+) = \mathsf{stp}((\bar{a}')^+)$, then $\cA[\bar{a}]$ and $\cA[\bar{a}']$ have the same fractional edge-cover number.   
\end{claim}
\begin{proof}
The equality $\mathsf{stp}(\bar{a}^+)=\mathsf{stp}((\bar{a}')^+)$ identifies exactly the same pairs of positions in the two flattened tuples. Since the relation-symbol sequence $(R_1,\ldots,R_\ell)$ is the same, mapping each equivalence class of positions in $\bar{a}^+$ to the corresponding equivalence class in $(\bar{a}')^+$ gives an isomorphism between the two incidence hypergraphs. Thus the fractional edge-cover LPs for $\cA[\bar{a}]$ and $\cA[\bar{a}']$ are identical up to renaming variables and constraints. 
\end{proof}

\begin{lemma}\label{lem:reverseFracHD}
Let $\cB, \cB'$ be two $\sigma$-structures such that $\cB \not\equiv_{\rcrFrac}\cB'$. There exists a $\sigma$-structure $\cA$ and a \emph{full} and \emph{pure} FHD $D$ for $\cA$ with width at most $k$ such that $\#\Hom(\cA^D,\fenc(\cB,k))\neq\#\Hom(\cA^D,\fenc(\cB',k))$. Consequently, $\#\Hom(\cA,\cB)\neq\#\Hom(\cA,\cB')$.
\end{lemma}

\begin{proof}
Recall that $\fenc(\cB,k), \fenc(\cB',k)$ are $\hat{\sigma}_{\upp(k)}$-structures (see \Cref{def:FracEncoding}). By our assumption and \Cref{prop:FractionalkRCR.1WL} it follows that $\fenc(\cB,k) \not\equiv_{1\textup{-WL}} \fenc(\cB',k)$, which due to \Cref{thm:ClassicalColorRefinement} in turn implies that there is a $\hat{\sigma}_{\upp(k)}$-structure $\mathcal{T}$, the Gaifman graph of which is a tree ---which we denote by $T$--- such that $\#\Hom(\mathcal{T}, \fenc(\cB,k)) \neq \#\Hom(\mathcal{T}, \fenc(\cB',k))$. Since every element of either target belongs to exactly one unary relation, every element of each print defined below has a unique unary profile.

Similarly to the proof of \cite[Lemma 4.6]{scheidt2026colorrefinementrelationalstructures}, we define for each $h \in \Hom(\mathcal{T}, \fenc(\cB,k))$ the \emph{print} $\mathsf{Print}(h)$ of $h$ which is a $\hat{\sigma}_{\upp(k)}$-structure with the same domain as $\mathcal{T}$ and relations which are given as follows:

\begin{enumerate}
\item[(i)] For each $v \in \dom(\mathcal{T})$ and $\alpha \in \sigma^{(\leq\upp(k))}$, we have $v \in (U_\alpha)^{\mathsf{Print}(h)}$ if and only if $h(v) \in (U_\alpha)^{\fenc(\cB,k)}$;
\item[(ii)] For each $(u, v) \in E(T) \,\bigcup \,\{(w,w) \mid w \in V(T)\}, i \in [|\mathsf{flat}(h(u))|], j \in [|\mathsf{flat}(h(v))|]$, we have $(u, v) \in (E_{i,j})^{\mathsf{Print}(h)}$ if and only if $(h(u), h(v)) \in (E_{i,j})^{\fenc(\cB,k)}$.
\end{enumerate}

We also define the print of each $h \in \Hom(\mathcal{T}, \fenc(\cB',k))$ in the same fashion. Note that the Gaifman graph of every print is $T$.

Since $\hat{\sigma}_{\upp(k)}$ is a binary signature, it follows from the proof of \cite[Lemma 4.6]{scheidt2026colorrefinementrelationalstructures} that there exists a print $Q \coloneq \mathsf{Print}(h_Q)$ where $h_Q \in \Hom(\mathcal{T}, \fenc(\cB,k)) \bigcup \Hom(\mathcal{T}, \fenc(\cB',k))$ such that $\#\Hom(Q, \fenc(\cB,k)) \neq \#\Hom(Q, \fenc(\cB',k))$.

We now proceed with the construction of the claimed $\sigma$-structure $\cA$ and decomposition $D = (T, B, \lambda)$ (where $T$ is the Gaifman graph of $\mathcal{T}$).
In particular, we show that the $\hat{\sigma}_{\upp(k)}$-structures $\cA^D$ and $Q$ are isomorphic, which will in turn imply that $\#\Hom(\cA^D, \fenc(\cB,k)) \neq \#\Hom(\cA^D, \fenc(\cB',k))$. We also show that $D$ is by construction a full and pure FHD and hence due to \Cref{lem:FHD.hom.eq} we derive the last claim of our statement, that is, $\#\Hom(\cA, \cB) \neq \#\Hom(\cA, \cB')$.

To this end, let $\mathsf{\Omega}$ be an countably infinite set such that $\mathsf{\Omega} \,\cap\, (\dom(\cB) \bigcup \dom(\cB')) = \emptyset$. We designate a node $r$ as the root of $T$ and see the rest of the nodes as being directed away from the root. We traverse $T$ in a \textit{breadth-first search} (BFS) fashion and perform the following for each node $v$ that we visit:

\begin{enumerate}
    \item [I.] We write $\alpha_v = (R_1, \dots, R_{\ell_v}) \in \sigma^{(\leq\upp(k))}$ for the unique profile for which $v \in (U_{\alpha_v})^Q$ holds, where $\ell_v \in [\upp(k)]$.\footnote{Uniqueness follows from the definition of $\fenc$ in \Cref{def:FracEncoding}.} We introduce $\ell_v$ tuples $\bar{t}(v;1), \dots, \bar{t}(v;\ell_v)$ where $\bar{t}(v;i) \in \mathsf{\Omega}^{\mathsf{ar}(R_i)},$ for each $i \in [\ell_v]$. Letting $\bar{t}(v) \coloneq \bar{t}(v;1) + \dots + \bar{t}(v;\ell_v)$, we enforce $\mathsf{stp}(\bar{t}(v), \bar{t}(v)) = \{(i, j) \mid (v, v) \in (E_{i, j})^Q\}\}$ which is always possible since $\{(i, j) \mid (v, v) \in (E_{i, j})^Q\}\} = \mathsf{stp}(h_Q(v), h_Q(v))$ and we further ensure that no entry of $\bar{t}(v)$ appears as an entry in a tuple corresponding to any node which we have already visited. We write $\mathsf{stp}_{u, v}$ for $\mathsf{stp}(\bar t_u, \bar t_v)$, where $\bar t_u, \bar t_v$ as given in this step.
    \item[II.] Let $p_v$ denote the parent of $v$ in $T$ (assuming that $v \neq r$). Then, for each $(i, j) \in \mathsf{stp}(h_Q(p_v), h_Q(v))$ we replace $\bar{t}(v)[j]$ with $\bar{t}(p_v)[i]$ and also replace with $\bar{t}(p_v)[i]$ the content of every other entry $m$ of $\bar{t}(v)$ such that $(v, v) \in (E_{j, m})^Q$.
\end{enumerate}

\begin{claim}\label{claim:InvariantSTPFrac}
For each $\{u, v\} \in E(T) \,\bigcup \, \{\{w, w\} : w \in V(T)\}$, we have $\mathsf{stp}(\bar{t}(u), \bar{t}(v)) = \mathsf{stp}(h_Q(u), h_Q(v))$.
\end{claim}
\begin{proof}
The proof follows verbatim the proof of \Cref{claim:InvariantSTP}.
\end{proof}

Next, we define our claimed $\sigma$-structure $\cA$ as follows: 
\begin{enumerate}
\item[(a.)] $\dom(\cA) = \bigcup_{v \in V(T)}\mathsf{set}(\bar{t}(v))$;
\item[(b.)] for each $R \in \sigma$, $R^\cA = \{\bar{a} \mid \text{there are $v \in V(T), i \in [\ell_v]$ such that $\alpha_v[i] = R$ and $\bar{t}(v;i) = \bar{a}$}\}$.
\end{enumerate}

We also consider the decomposition $D = (T, B, \lambda)$ where for each $v \in V(T)$ we define $B(v) \coloneq \mathsf{set}(\bar{t}(v))$ and $\lambda(v) \coloneq (R_i(\bar{t}(v;i)))_{i \in [\ell_v]}$, where $R_i = \alpha_v[i]$. We have the following: 
\begin{enumerate}
    \item By construction $D$ is full and pure. 
    \item By definition, the substructure of $\cB$ induced by the coloured tuples that $h_Q(v)$ is comprised of has fractional edge cover at most $k$. Furthermore, by \Cref{claim:InvariantSTPFrac}, we have $\mathsf{stp}(\fulltup(v), \fulltup(v)) = \mathsf{stp}(\bar{t}_v, \bar{t}_v) = \mathsf{stp}(h_Q(v), h_Q(v))$ and by definition we have that $v$ and $h_Q(v)$ have the same profile. The above combined with \Cref{claim:stpFrac} yield that $\cA[\lambda(v)]$ has fractional edge cover number at most $k$.
    \item The subtree of $T$ induced by $\{v \in V(T) \mid x \in B(v)\}$ is connected for each $x \in \dom(\cA)$. To see this, let $x \in \dom(\cA)$ and $u, v \in V(T)$ such that $\bar{t}(u)[i] = \bar{t}(v)[j] = x$, for some $i, j$. By construction, for any node $w$ and its parent $p_w$ we have that the entries of $\bar{t}(w)$ that appear outside of $T_w$, where $T_w$ is the subtree of $T$ rooted at $w$, must also appear in $B(p_w)$. Hence, $u$ and $v$ must have a least common ancestor $w$ such that $x \in B(w)$ and thus $x$ must also be contained in the bag of every node in the path from $w$ to $u$ as well as in the path from $w$ to $v$.
\end{enumerate}

Hence, $D$ is a valid full (and pure) FHD of $\cA$ with width at most $k$. Finally, we show that $\cA^D$ is indeed isomorphic to $Q$. For this, we recall our previous observation according to which by the construction of $\bar{t}(v), v \in V(T)$ it follows that for each $\{u, v\} \in E(T) \,\bigcup \, \{\{w, w\} : w \in V(T)\}$, we have $\mathsf{stp}(\bar{t}(u), \bar{t}(v)) = \mathsf{stp}(h_Q(u), h_Q(v)) = \{(i, j) \mid (u, v) \in (E_{i, j})^Q\}$. Furthermore, it follows by definition that for each $v \in V(T)$ we have $v \in (U_{\alpha_v})^Q$ and $\profile(v) = \alpha_v$ which implies that $v \in (U_{\alpha_v})^{\cA^D}$. Hence we deduce that $\cA^D$ is isomorphic to $Q$, which completes the proof.

\end{proof}

\subsubsection{Putting the pieces together}

We may now prove \Cref{maintheorem:boundedFHDforStructures} as follows.

\begin{proof}[Proof of \Cref{maintheorem:boundedFHDforStructures}]
Direction (1) $\implies $ (2) follows from \Cref{lem:reverseFracHD}. In particular, \Cref{lem:reverseFracHD} does not state that $\cA$ must be connected, however it can be readily verified that since $\#\Hom(\cA, \cB) \neq \#\Hom(\cA, \cB')$, there must already exists a (maximal) connected substructure of $\cA$ with different number of homomorphisms to $\cB$ and $\cB'$ respectively.

Finally, we derive direction $(2) \implies (1)$ by contraposition. To this end, assume that $\cB \equiv_\rcrFrac \cB'$. By \Cref{remark:GHDforRelationalStructures}, we may choose the decomposition to be full without increasing its width and while preserving semi-purity. Then, by \Cref{lem:connectedHomIndFrac} it follows that for any connected $\sigma$-structure $\cA$ that admits a semi-pure FHD of width at most $k$, it holds $\homs(\cA, \cB) = \homs(\cA, \cB')$, which completes the proof.
\end{proof}

\section{Proof of \Cref{lem:ghd.hom.eq}}
\label{lem:panosproof.ghd.hom.eq}
\begin{lemma}[Analogue of Lemma 4.5 of \cite{scheidt2026colorrefinementrelationalstructures}]\label{lem:PureAppendix}
Let $\cA$ be a $\sigma$-structure and let $D = (T, B, \lambda)$ be a \emph{full} and \emph{pure} GHD of $\cA$ with width at most $k$.
Then for every $\sigma$-structure $\cB$,
\[
\homs(\cA^{D},\mathbf{B}(\cB,k)) = \homs(\cA,\cB)\,.
\]
\end{lemma}
\begin{proof}

Let $u \in \dom(\cA^D)$ with $\mathsf{Ord}\lambda(u) = R_1(\bar{a}_1),\dots,R_\ell(\bar{a}_\ell)$, where $1 \leq \ell \leq k$. Since $D$ is pure we have that $\bagtup(u) = \bar{a}_1 + \dots+\bar{a}_\ell$, where $+$ is overloaded so as to be also used as a binary operator that concatenates two tuples. Also, recall that $\profile(u) = (R_1, \dots, R_\ell)$. Given a mapping $h : \dom(\cA) \to \dom(\cB)$, we write $h(\bar{a}_i)$ for the tuple obtained by applying $h$ entry-wise on $\bar{a}_i$. We also write $h(\bagtup(u)) = h(\bar{a}_1) + \dots + h(\bar{a}_\ell)$.

\begin{claim}
If $h$ is a homomorphism (from $\cA$ to $\cB$), then the mapping $h' : u \in V(T) \mapsto h(\bagtup(u))$ is a homomorphism in $\Hom(\cA^D, \mathbf{B}(\cB,k))$. 
\end{claim}
\begin{proof}
For $u \in \dom(\cA^D)$ as defined earlier with $\profile(u) = (R_1, \dots, R_\ell)$ we have that $u \in U^{\cA^D}_{(R_1, \dots, R_\ell)}$. Hence, we first need to show that $h'(u) \in U^{\mathbf{B}(\cB,k))}_{(R_1, \dots, R_\ell)}$. Since $h \in \Hom(\cA, \cB)$ it follows that $\bar{a}_i \in R^{\cA}_i \Rightarrow h(\bar{a}_i) \in R^{\cB}_i$, for each $i \in [\ell]$. Hence, $h'(u) = h(\bar{a}_1)+ \dots +h(\bar{a}_\ell)$ which can be equivalently seen as an element of $R^{\cB}_1 \times \dots \times R^{\cB}_\ell \subseteq \mathsf{Prod}^{\leq k}(\cB)$ and by definition it follows that $h'(u) \in U^{\mathbf{B}(\cB,k))}_{(R_1, \dots, R_\ell)}$.  Next, we take $v \in \dom(\cA^D)$ and write $\bar{a} = \bagtup(u)$ and $\bar{b} = \bagtup(v)$. Consider $(u, v) \in E^{\cA^D}_{i, j}$ which holds if $\bar{a}[i] = \bar{b}[j]$. Clearly, $h(\bar{a}[i]) = h(\bar{b}[j])$ and so $(h'(u), h'(v)) \in E^{\mathbf{B}(\cB,k))}_{i, j}$. Hence, the mapping $h' : u \mapsto h(\bagtup(u))$ is a well-defined mapping that preserves relations, which completes the proof.   
\end{proof}

Let $\pi : \Hom(\cA, \cB) \to \Hom(\cA^D, \mathbf{B}(\cB, k))$ be the mapping that maps $h$ to $h'$ with $h'$ as defined above w.r.t $h$. We show that $\pi$ is bijective, which would conclude the proof.

\begin{claim}
$\pi$ is injective.   
\end{claim}

\begin{proof}
Let $h_1, h_2 \in \Hom(\cA, \cB)$ such that $h_1 \neq h_2$ which means that there is $x \in \dom(\cA)$ such that $h_1(x) \neq h_2(x)$. Recall that we have assumed that there is $R \in \sigma$ and $\bar{c} \in R^\cA$ such that $x \in \bar{c}$ (see, \Cref{remark:no_isolated_elements}). Furthermore, by the definition of $D$, there is $u \in T$ such that $\mathsf{set}({\bar{c}}) \subseteq B(u)$ and so $x \in \bagtup(u)$. It is then easy to verify that $\pi(h_1)(u) \neq \pi(h_2)(u)$ and so $\pi$ is injective.
\end{proof}

Next we need to show that, for each $h'' \in \Hom(\cA^D, \mathbf{B}(\cB, k))$, there is $h \in \Hom(\cA, \cB)$ such that $\pi(h) = h''$. To this end, recall that by assumption, each element $z \in \dom(\cA)$ is contained in some tuple and thus it is also contained in some bag $B(u)$. For each $z \in \dom(\cA)$, we fix a node $u_z \in V(T)$ such that $z \in \bagtup(u_z)$. Recall that if $\bagtup(u_z)[i] = \bagtup(u_z)[j]$ then $(u_z, u_z) \in E^{\cA^D}_{i,j}$. Since $h''$ is a homomorphism it also follows that $(h''(u_z), h''(u_z)) \in E^{\mathbf{B}(\cB, k)}_{i,j}$ implying that there is $x_z \in \dom(\cB)$ such that for any index $i$, if $\bagtup(u_z)[i] = z$, then $h''(u_z)[i] = x_z$. We consider the well-defined mapping $h : \dom(\cA) \to \dom(\cB)$ that maps $z$ to $x_z$.

\begin{claim}\label{claim:Surjectivity}
For $u \in V(T)$ with $\bar{a} = \bagtup(u)$ such that $h''(u) = \bar{d} \in \dom(\mathbf{B}(\cB, k))$, we have $h(\bar{a}) = \bar{d}$.    
\end{claim}

\begin{proof}
Let $z \in \bar{a}$. Since $D$ is pure it follows that $z \in B(u)$. Recall that we have fixed $u_z \in V(T)$ such that $z \in B(u_z)$. By the definition of $D$, it follows that for the path $(u_0, u_1, \dots, u_{q-1}, u_q)$ in $T$ connecting $u_0 = u$ and $u_q = u_z$, it holds that $z \in B(u_j)$, for each $0 \leq j \leq q$. For each $0 \leq j \leq q$, write $\bar{t}_j  = \bagtup(u_j)$ and let $i_j$ denote any index such that $\bar{t}_j[i_j] = z$. Then, it easy to see that for each $0 \leq j < q$, we have by definition that $(u_j, u_{j+1}) \in E^{\cA^D}_{i_j,i_{j+1}}$ which implies that we also have $(h''(u_j), h''(u_{j+1})) \in E^{\mathbf{B}(\cB, k)}_{i_j, i_{j+1}}$. Concretely we have
\begin{enumerate}
\item $\bar{t}_q[i_q] = \bar{t}_0[i_0] = \bar{a}[i_0] = z$;
\item $h''(u_0)[i_0] = h''(u_z)[i_q]$;
\item $h''(u_z)[i_q] = x_z$ due to (1) and the definition of $u_z$.
\end{enumerate}

Hence, $\bar{d}[i_0] \overset{def}= h''(u_0)[i_0] \overset{(2)}= h''(u_z)[i_q] \overset{(3)}= x_z \overset{def}= h(\bar{a}[i_0])$ which also shows that $h(\bar{a}) = \bar{d}$.
\end{proof}

\begin{claim}
$h \in \Hom(\cA, \cB)$.    
\end{claim}
\begin{proof}
Let $R \in \sigma$ and $\bar{t} \in R^\cA$. Since $D$ is full, there is $u \in V(T)$ such that $R(\bar{t}) \in \lambda(u)$. Assume that $R(\bar{t})$ is the $i$-th coloured tuple in $\mathsf{Ord}\lambda(u) = R_1(\bar{a}_1), \dots, R_\ell(\bar{a}_\ell)$, where $1 \leq \ell \leq k$. Let $\profile(u) = (R_1, \dots, R_i, \dots, R_\ell)$ where $R_i = R$. Since, $h'' \in \Hom(\cA^D, \mathbf{B}(\cB, k))$, we have $h''(u) \in R_1^{\cB} \times \dots \times R^\cB \times \dots \times R_\ell^\cB$ which implies that $h(\bar{t}) \in R^\cB$ since $h''(u) = h(\bar{a}_1) + \dots + h(\bar{a}_i) + \dots + h(\bar{a}_\ell)$ and $\bar{t} = \bar{a}_i$ which completes the proof.
\end{proof}

Finally, it is easy to verify that $\pi(h) = h''$ which follows from \Cref{claim:Surjectivity}, showing that $\pi$ is also surjective which completes the proof.

\end{proof}

\section{Proof of Theorem~\ref{thm:intro_rt}}\label{appendix_rt}
\begin{proof}
In a nutshell, we construct an \textit{augmented} structure $\cA^+$ with the same domain as $\cA$ and tuples given by $\mathcal{T}_k(\cA)$, where we enforce each flat tuple to appear in some relation, by adding a constant number of suitably defined fresh relations, and show that $\hmrcr$ when run on $\cA$ is equivalent to RCR when run on $\cA^+$. In other words, we show that the partitioning of $\mathcal{T}_k(\cA)$ into the stable colour classes produced by $\hmrcr$ (on $\cA$) is precisely the same as the one produced by RCR (on $\cA^+)$. It is then easy to verify that this suffices to show the theorem, since $\|\cA^+\| = |\mathcal{T}_k(\cA)| = O(\|A\|^k)$.

Let $r_\mathsf{max}$ be the maximum length of a flat tuple in $\mathcal{T}_k(\cA)$ and note that it is constant, as it only depends on $k$ as well as the maximum arity among relations in $\sigma$. For each $1 \leq \ell \leq r_{\mathsf{max}}$ we introduce a fresh relation symbol $R_\ell$ (that is, $R_\ell \notin \sigma$), which we fill in with all flat tuples in $\mathcal{T}_k(\cA)$ of length exactly $\ell$. Let $\sigma^+ \coloneq \sigma \,\cup \left(\,\bigcup_{1 \leq \ell \leq r_{\mathsf{max}}}\{R_\ell\}\right)$. We define the \textit{augmented} $\sigma^+$-structure $\cA^+$, which has the same domain as $\cA$ and with relations $R^\cA$, for each $R \in \sigma^+$ (that is, we also preserve the original relations of $\cA$).

We write $\kappa_i \coloneq \kappa_i^\cA$ (resp. $\kappa_i^+ \coloneq \kappa_i^{\cA^+})$ for the round-$i$ colouring produced by $\hmrcr$ on $\cA$ (resp. by RCR on $\cA^+$). Recall that both of these colourings colour flat tuples in $\mathcal{T}_k(\cA)$. We show by induction on the iteration number $i$ of the algorithms, that for any two tuples $s, s' \in \mathcal{T}_k(\cA)$, $\kappa_i(s) = \kappa_i(s')$ if and only if $\kappa_i^+(s) = \kappa_i^+(s')$. We write $\atp_{\cA^+}(s)$ for the atomic type of a tuple $s$ w.r.t. $\cA^+$.

\medskip

\noindent\textbf{Base case $i = 0$}: If $\kappa_i(s) = \kappa_i(s')$, then $\atp(s) = \atp(s')$ and $\stp(s,s) = \stp(s', s')$. We have $\atp_{\cA+}(s) = \atp(s) \,\dot\cup\, \{R_{\ell_s}\}$, where $\ell_s$ is the length of $s$. Furthermore, $\ell_s = \ell_{s'}$, since $\stp(s, s) = \stp(s', s')$. Hence, we deduce that $\kappa_i^+(s) = \kappa_i^+(s')$. Similarly, it can be shown that $\kappa_i^+(s) = \kappa_i^+(s') \implies \kappa_i(s) = \kappa_i(s')$.
\medskip

\noindent\textbf{Induction hypothesis}: We assume that the statement is true for every $i < j$. We show, that it is also true for $i = j$. To this end, let $s, s'$ be two flat tuples satisfying $\kappa_j(s) = \kappa_j(s')$. We need to show that

\begin{multline*}
(\kappa_{j-1}^+(s), \{\{(\stp(s, t), \kappa_{j-1}^+(t)) \mid t \in \mathcal{T}_k(\cA) : \stp(s, t) \neq \emptyset\}\}) = \\ 
(\kappa_{j-1}^+(s'), \{\{(\stp(s', t'), \kappa_{j-1}^+(t')) \mid t' \in \mathcal{T}_k(\cA) : \stp(s', t') \neq \emptyset\}\})
\end{multline*}

By the induction hypothesis, since $\kappa_{j-1}(s) = \kappa_{j-1}(s')$ we derive $\kappa^+_{j-1}(s) = \kappa^+_{j-1}(s')$, therefore it suffices to show that there is a bijection $\mu : \{t \in \mathcal{T}_k(\cA) : \stp(s, t) \neq \emptyset\} \to \{t' \in \mathcal{T}_k(\cA) : \stp(s', t') \neq \emptyset\}$ such that for each $t$ in the domain of $\mu$ we have $(\stp(s, t), \kappa^+_{j-1}(t)) = (\stp(s', \mu(t)), \kappa^+_{j-1}(\mu(t)))$.

Since $\kappa_j(s) = \kappa_j(s')$, it follows that there is a bijection $\eta : \{t \in \mathcal{T}_k(\cA) : \stp(s, t) \neq \emptyset\} \to \{t' \in \mathcal{T}_k(\cA') : \stp(s', t') \neq \emptyset\}$ such that for each $t$ in the domain of $\eta$ we have $(\stp(s, t), \kappa_{j-1}(t)) = (\stp(s, \eta(t)), \kappa_{j-1}(\eta(t)))$. By the induction hypothesis, we have $\kappa^+_{j-1}(t) = \kappa_{j-1}^+(\eta(t))$. Hence, we may take $\mu \coloneq \eta$. Similarly, it can be shown that $\kappa_i^+(s) = \kappa_i^+(s') \implies \kappa_i(s) = \kappa_i(s')$.
\end{proof}

\end{document}